\documentclass[11pt]{article}
\usepackage[utf8]{inputenc}
\usepackage{amsmath}
\usepackage{amsfonts}
\usepackage{mathtools}
\usepackage{amsthm}
\usepackage{calc}
\usepackage{amssymb}
\usepackage{subcaption}
\usepackage{xcolor}
\usepackage[linktocpage]{hyperref}
\usepackage{amsmath,amssymb}
\usepackage{bbm}
\newcommand{\diag}[1]{\texttt{Diag}\big({#1}\big)}
\newcommand{\diagtwo}[1]{\texttt{diag}\big({#1}\big)}
\newcommand\tab[1][1cm]{\hspace*{#1}}

\newcommand{\vardbtilde}[1]{\tilde{\raisebox{0pt}[0.85\height]{$\tilde{#1}$}}}
\newcommand{\logdet}[1]{\text{logdet}(#1)}
\newcommand\numberthis{\addtocounter{equation}{1}\tag{\theequation}}
 \newcommand{\polylog}{\mathrm{polylog}} 

\usepackage{graphicx}
\graphicspath{ {./images/} }
\usepackage[margin=1in]{geometry}
\usepackage{cancel}

\newtheorem{theorem}{Theorem}[section]
\newtheorem{corollary}{Corollary}[theorem]
\newtheorem{lemma}[theorem]{Lemma}
\newtheorem{definition}{Definition}

\newcommand{\tr}{\texttt{tr}}
\newcommand{\indic}{\mathbbm{1}}
\newcommand{\R}{\mathbb R}
\newcommand{\Rn}{\mathbb R^n}
\newcommand{\A}{\mathrm{A}_x}
\newcommand{\G}{\mathbf{G}_x}
\newcommand{\W}{\mathbf{W}_x}
\newcommand{\Lambdax}{\mathbf{\Lambda}_x}
\newcommand{\Ptwo}{\mathbf{P}^{(2)}_x}
\newcommand{\Px}{\mathbf{P}_x}
\newcommand{\Pxv}{\mathbf{\tilde P}_{x,v}}
\newcommand{\Pxz}{\mathbf{\tilde P}_{x,z}}

\newcommand{\Rxv}{{\mathrm{R}_{x,v}}}
\newcommand{\Rxz}{{\mathrm{R}_{x,z}}}
\newcommand{\Rxu}{{\mathrm{R}_{x,u}}}
\newcommand{\Sxv}{{\mathrm{S}_{x,v}}}
\newcommand{\Sxz}{{\mathrm{S}_{x,z}}}
\newcommand{\Sxu}{{\mathrm{S}_{x,u}}}
\newcommand{\sxv}{{s_{x,v}}}

\newcommand{\sxz}{{s_{x,z}}}
\newcommand{\rxv}{{r_{x,v}}}
\newcommand{\rxu}{{r_{x,u}}}
\newcommand{\rxz}{{r_{x,z}}}
\newcommand{\D}{\mathrm D}
\newcommand{\trr}{\triangleright}
\newcommand{\trrr}{\triangleright\triangleright}
\newcommand{\Wxv}{\mathbf{W'}_{x,v}}
\newcommand{\Wxz}{\mathbf{W'}_{x,z}}
\newcommand{\brackets}[1]{\Big[ #1\Big]}
\newcommand{\bracketss}[1]{\Big[ #1\Big]^{1/2}}
\newcommand{\qq}{\mathrm{q}}
\newcommand{\Ricci}{\texttt{Ricci}}
\newcommand{\Rr}{\mathrm{R}}
\newcommand{\gv}{y}
\newcommand{\sv}{s}
\newcommand{\Pxhat}{\mathbf{\hat P}_x}
\newcommand{\jon}[1]{}
\newcommand{\khasha}[1]{}
\newcommand{\santosh}[1]{}

\title{Sampling with Barriers: Faster Mixing via Lewis Weights}
\author{Khashayar Gatmiry\thanks{MIT, gatmiry@mit.edu. Part of this work was done while visiting Georgia Tech and supported by NSF award CCF-2007443.}, Jonathan Kelner\thanks{MIT, kelner@mit.edu}, Santosh S. Vempala\thanks{Georgia Tech. vempala@gatech.edu. Supported in part by NSF awards CCF-2007443 and CCF-2106444.}}
\date{}

\begin{document}

\maketitle

\begin{abstract}
    We analyze Riemannian Hamiltonian Monte Carlo (RHMC) for sampling a polytope defined by $m$ inequalities in $\R^n$ endowed with the metric defined by the Hessian of a convex barrier function. The advantage of RHMC over Euclidean methods such as the ball walk, hit-and-run and the Dikin walk is in its ability to take longer steps. However, in all previous work, the mixing rate has a linear dependence on the number of inequalities. We introduce a hybrid of the Lewis weights barrier and the standard logarithmic barrier and prove that the mixing rate for the corresponding RHMC is bounded by $\tilde O(m^{1/3}n^{4/3})$, improving on the previous best bound of $\tilde O(mn^{2/3})$ (based on the log barrier). This continues the general parallels between optimization and sampling, with the latter typically leading to new tools and more refined analysis. To prove our main results, we have to overcomes several challenges relating to the smoothness of Hamiltonian curves and the self-concordance properties of the barrier. In the process, we give a general framework for the analysis of Markov chains on Riemannian manifolds, derive new smoothness bounds on Hamiltonian curves, a central topic of comparison geometry, and extend self-concordance to the infinity norm, which gives sharper bounds; these properties appear to be of independent interest.
\end{abstract}

\newpage

\tableofcontents
\thispagestyle{empty}
\newpage
\pagenumbering{arabic}

\section{Introduction}
Generating nearly uniform random samples from a high-dimensional polytope is a fundamental algorithmic problem with a rich history and powerful applications, notably including the only known 
fully polynomial-time approximation schemes for computing a polytope's volume.
All efficient algorithms known for this problem work by designing a Markov chain whose stationary distribution is uniform over the polytope and showing that it mixes in a small number of steps.

In this paper, our main result is that we can construct such a Markov chain with an improved bound on its mixing time.
For a polytope given by $m$ linear inequalities in $\R^{n}$, we describe chain that mixes in $\tilde{O}\left(m^{1/3}n^{4/3}\right)$ steps, improving on the best previous bound of  $\tilde{O}\left(mn^{2/3}\right)$.
This allows us to approximate the volume within relative error $\epsilon$ using $\tilde{O}\left(m^{1/3}n^{4/3}/\epsilon^2\right)$ steps, which is a similar improvement over the best existing bound of $\tilde{O}\left(mn^{2/3}/\epsilon^2\right)$. 

\subsection{Background and Related Work}
In their seminal work~\cite{dyer1991random}, Dyer, Frieze and Kannan gave the first polynomial-time algorithm for this problem, as well as for the more general problem of sampling from a convex body specified by a membership oracle.
The Markov chain in their algorithm was a \emph{grid walk}, which takes steps along the edges of the graph obtained by intersecting the convex body with a discrete grid supported on $\delta \mathbb{Z}^n$ for some $\delta=1/\mathrm{poly}(n)$.
This graph is heavily dependent on the coordinate system---its diameter is proportional to the diameter of the convex body, and its conductance can be arbitrarily small if the convex body is scaled so that is very long in some directions but short in others.
However, they showed that, if one changes to a basis in which the convex body is appropriately ``well-rounded,'' the grid walk mixes in polynomial time and that one can use a random sample from the grid to obtain a one from the convex body.

The polynomial for the mixing time in~\cite{dyer1991random} was quite large, and a sequence of later papers improved this by modifying the Markov chains and refining the analysis.  
Because one often wants to draw many samples from the body, these papers typically provide two bounds on the number of steps required: a bound when starting from an arbitrary point and including the cost of any preprocessing; and a bound when given a \emph{warm start}, where the preprocessing has already been performed and the starting point is drawn from a distribution that is not too far from uniform.

In~\cite{kannan1997random}, Kannan, Lov{\'a}sz, and Simonovits showed that a \emph{ball walk} whose steps are chosen uniformly from a Euclidean ball around the current point mixes in $\tilde{O}(n^3)$ steps from a warm start and $\tilde{O}(n^5)$ steps from an arbitrary starting point and including preprocessing.
Later, Lov{\'a}sz and Vempala~\cite{lovasz2006hit} studied the ``hit-and-run'' walk, which chooses a line in a random direction from the current point and then picks the next point randomly from the intersection of this line with the body, and they showed it also mixed in $\tilde{O}(n^3)$ steps from a warm start but needed only $\tilde{O}(n^4)$ steps for first sample and preprocessing.
These algorithms work on general convex bodies presented by oracles, but like the grid walk, they are strongly coordinate dependent, and they thus require strong additional assumptions about the coordinate system.
In particular, analyses of these algorithms typically assume that body is close to isotropic, i.e., that the covariance matrix of a random sample from the body is approximately the identity, and applying these algorithms to more general bodies requires costly preprocessing.

The dependence on the coordinate system in the aforementioned Markov chains comes from the dependence of the transition probabilities on the extrinsic geometry of the ambient Euclidean space.
The impact of this extends beyond the overhead from the isotropy requirements.
The geometry of the ambient space does not incorporate any information about how close a point is to the boundary, which typically leads to difficulties making progress with steps near the boundary.
For example, if one is running a ball walk with step radius $\delta$ an $n$-dimensional cube, and the current point is some distance $d\ll \delta$ from one of the corners, a random point from the radius $\delta$ ball will lie outside the cube with probabability exponentially close to 1, so naively trying random points until obtaining one in the cube would take a large number of tries.
Moreover, even if one could sample a random point in the intersection of the ball with the cube, restricting the step to points inside the cube would distort the stationary distribution, and it would no longer be uniform.
Remedying such difficulties typically involves (depending on the paper) some combination of taking smaller steps, enlarging the convex body (and failing if the walk ends up at a point outside the original body), and employing rejection sampling or a Metropolis filter to correct the stationary probabilities, all of which increase the required number of steps.

For polytopes specified by an explicit collection of linear constraints, one can use the barrier functions employed by interior point methods to design random walks whose steps depend only on the intrinsic geometry of the polytope and are independent of the basis chosen for the ambient space.
The idea behind these random walks is to use the Hessian of the barrier function to define a local norm/Riemannian metric on the interior of the polytope and specify the steps in terms of the resulting geometry.
This mitigates some of the problems described above and has led to Markov chains whose mixing times grow with the number of constraints but depend more mildly on the dimension.

In the first such work, Kannan and Narayanan~\cite{kannan2012random} introduced the \emph{Dikin walk} and gave a mixing time bound of $O(mn)$ from a warm start for a polytope with $m$ facets in $\R^n$.
This walk is similar to the ball walk, but it chooses its steps from Dikin ellipsoids, which are balls with respect to the Hessian of the standard logarithmic barrier function on the polytope.
In \cite{laddha2020strong}, Laddha, Lee, and Vempala studied the analogous walk with respect to any self-concordant barrier and showed that it mixes in $\tilde{O}(n\bar{\nu})$ steps, where $\bar{\nu}$ is a parameter they called the 
\emph{barrier parameter}. By bounding this parameter for a different barrier function (a variant of a barrier due to Lee and Sidford \cite{lee2014path}), they obtained an improved mixing rate bound of $\tilde{O}(n^2)$.

In 2017, Lee and Vempala~\cite{lee2017geodesic} reduced the mixing rate to $\tilde{O}\left(mn^{3/4}\right)$ using  a process they called the \emph{geodesic walk}.
Like in the Dikin Walk, the steps are constructed using the Hessian of a barrier function.
However, instead of using this to define a Euclidean ellipse, they use it to define a Riemannian metric, and they then solve a differential equation in each step to follow geodesics on the resulting manifold.
These geodesics tend to curve away from the polytope's boundary, which lets them take longer steps in each iteration.

In 2018, Lee and Vempala~\cite{lee2018convergence} improved this to $\tilde{O}\left(mn^{2/3}\right)$ using Riemannian Hamiltonian Monte Carlo (RHMC)~\cite{girolami2011riemann}, which is the class of processes we'll use in this paper.
While there is a large literature on using RHMC and related methods to sample \emph{smooth} densities~\cite{dalalyan2017further,durmus2019analysis,chewi2021analysis,vempala2019rapid,li2022mirror,chewi2022log}, there are relatively few provable results about applying it in constrained non-smooth settings like polytope sampling. 
Roughly speaking, this improvement over the geodesic walk came from RHMC's ability to avoid the use of a Metropolis filter, which the geodesic walk requires in order to obtain the correct stationary distribution (even when the target distribution is uniform).
RHMC chooses its trajectories according to a different differential equation that, remarkably, yields a reversible random walk with the desired stationary distribution, thus eliminating the need for a Metropolis filter and allowing greater progress in each step.

Advances in self-concordant barriers in the past decade as well as the improvement in the analysis of the Dikin walk suggest that a smaller dependence on $m$, the number of inequalities, which can be much higher than the dimension, should be possible. Nevertheless, improving on the bound of $mn^{2/3}$ has been a major open problem for the past 5 years. Moreover, the new techniques developed as a result of progress on non-Euclidean algorithms suggest that this is a fertile area for further TCS research.

\begin{table}[h]
    \centering{}%
    \begin{tabular}{|c|l|c|}
    \hline 
    Year  & Algorithm & Steps\tabularnewline
    \hline 
    \hline 
    1997 \cite{kannan1997random} & Ball walk\textsuperscript{\#} & $n^{3}$ (+$n^5$) \tabularnewline
    \hline 
    2003 \cite{lovasz2006hit} & Hit-and-run\textsuperscript{\#} & $n^{3}$ (+$n^4$) \tabularnewline
    \hline 
    2009 \cite{kannan2012random} & Dikin walk & $mn$ \tabularnewline
    \hline 
    2017 \cite{lee2017geodesic} & Geodesic walk & $mn^{3/4}$ \tabularnewline
    \hline
    2018 \cite{lee2018convergence} & RHMC with log barrier & $mn^{2/3}$ \tabularnewline
    \hline
    2020 \cite{laddha2020strong} & Weighted Dikin walk & $n^2$ \tabularnewline
    \hline 
    2021 \cite{jia2021reducing} & Ball walk\textsuperscript{\#} & $n^2$ (+$n^3$)\tabularnewline
    \hline 
    \textcolor{red}{This paper} & \textcolor{red}{RHMC with Hybrid barrier} & \textcolor{red}{$m^{1/3}n^{4/3}$} \tabularnewline
    \hline 
    \end{tabular}\caption{\label{tab:sampling}The complexity of uniformly sampling a polytope from
    a warm start. All algorithms have a logarithmic dependence on the warm start parameter and each uses $\widetilde{O}(n)$
    bit of randomness. The entries marked \protect\textsuperscript{\#}
    are for general convex bodies presented by oracles, while the rest
    are for polytopes. The additive terms are pre-processing costs for rounding the polytope.}
\end{table}

\subsection{Background on Riemannian Hamiltonian Monte Carlo}

The motivation for RHMC comes from the Hamiltonian formulation of classical Newtonian mechanics.
Hamiltonian mechanics parameterizes a physical system in terms of a \emph{position} vector $x$ and a corresponding \emph{momentum} vector  $v$ (which is also referred to as ``velocity'' in some prior work on sampling polytopes with RHMC).
The physics of the system are encoded in its \emph{Hamiltonian} $H(x,v)$, which is simply the energy of the system written as a function of $x$ and $v$, and its time evolution is determined by \emph{Hamilton's equations}:
\begin{align*}
    \frac{dx}{dt} & =\frac{\partial H}{\partial v}(x,v)\\
    \frac{dv}{dt} & =-\frac{\partial H}{\partial x}(x,v).
\end{align*}
With the appropriate choice of $H$, these reproduce Newton's laws of motion, but they also generalize quite broadly, including to Riemannian manifolds.

In RHMC, one defines a Markov chain by choosing a Hamiltonian that appropriately encodes the target distribution.
At each step, the Markov chain chooses a random momentum vector and then finds the next point by numerically solving a differential equation to follow the trajectory given by Hamilton's equations.

One can show that the value of the Hamiltonian (i.e., the energy) and the volume element in the space of pairs $(x,v)$ are conserved along the trajectory, which can be used to show that the trajectories are preserved by time reversal (i.e., running time backwards).
One can then use this to show that, if one uses the Hamiltonian defined below, the marginal distribution of $x$ will converge to the desired target distribution without requiring a Metropolis filter.  (See~\cite{girolami2011riemann} for the derivation for general RHMC and~\cite{lee2018convergence} for the specific class of Hamiltonians given below.) 

More precisely, let the Hamiltonian at a point $x\in \R^n$ for a vector $v \in \R^n$ be defined as 
\begin{align}
    H(x,v)=f(x)+\frac{1}{2} v^{\top}g^{-1}(x) v+\frac{1}{2}\log\det g(x)\label{eq:hmcode},
\end{align}
where $g(x)$ is a positive definite matrix defining a Riemannian metric at each point $x$ as $\|u\|_g \triangleq \|u\|_{g(x)} \triangleq \sqrt{u^\top g(x) u}$, and the target density to be sampled is proportional to $e^{-f}$ restricted to the support of $g$. One step of RHMC consists of the following: first pick $v$ from the Gaussian $\mathcal N(x,g(x)^{-1})$. Then for time $\delta$ follow the Hamiltonian curve jointly on $(x,v)$:
\begin{align}
\frac{dx}{dt} & =\frac{\partial H}{\partial v}(x,v)=g^{-1}(x)v\nonumber \\
\frac{dv}{dt} & =-\frac{\partial H}{\partial x}(x,v)=-\nabla f(x)+\frac{1}{2}\mbox{tr}(g(x)^{-1}Dg(x))-\frac{1}{2} Dg(x)\left[\frac{dx}{dt},\frac{dx}{dt}\right].\label{eq:ham}
\end{align}

The final $x$ at time $\delta$ is the sampled point from the Markov Kernel. A natural choice for the metric $g$ turns out to be the Hessian of a self-concordant barrier function inside the polytope $\mathcal P$.
The standard logarithmic barrier, $\phi_\ell(x)=-\sum_{i=1}^m \log (a_i^\top x - b_i)$, was used in \cite{lee2018convergence} to prove that the resulting RHMC mixes in $mn^{2/3}$ steps. Improving on this bound is our motivating open problem. 

Using the log barrier implies that the mixing rate has a linear dependence on $m$, the number of inequalities. So we have to look for a ``better" barrier, and what exactly this entails will become clear presently. 
As we will see below, the barrier parameter of the self-concordant function, which is $m$ for the logarithmic barrier, plays an important role in the mixing time of this Markov chain.
Given that there are efficiently-computable barriers for which this parameter is $O(n)$~\cite{lee2014path}, one might hope to obtain faster mixing by simply replacing the logarithmic barrier with one of these.
However, it turns out that just bounding the barrier parameter is insufficient, and we need to choose a barrier that also possesses certain stronger smoothness and stability properties.
One of our primary technical challenges will be to define a notion that is stringent enough to guarantee the stronger properties required while still admitting a construction that improves upon the logarithmic barrier.

\subsection{Results}

In this paper, we use a hybrid barrier based on the $p$ Lewis weight barrier $\phi_p$ defined as
\begin{align}
    \phi_p(x) \triangleq \log\det{ \left( \A^\top\W^{1-2/p}\A \right) },\label{eq:lewisbarrier}
\end{align}
where $\W$ is a diagonal matrix whose diagonal entries are the $p$-Lewis weights of the rescaled matrix $\A = S_x^{-1}\mathrm A$ and $S_x$ is the diagonal matrix whose entries are the slacks at point $x$, i.e., $(S_x)_{ii} = a_{i}^\top x-b_i$.

We define a hybrid barrier $\phi$ for a polytope as follows.
\begin{definition}[Hybrid barrier]\label{def:hybridbarrier}
We define the {\em hybrid} barrier $\phi$ inside a polytope $Ax \ge b$ as
\begin{equation*}
    \phi(x) \triangleq  -\left(\frac{m}{n}\right)^\frac{2}{p+2}\left(\log\det{\A^\top\W^{1-2/p}\A} + \frac{n}{m} \sum_i \log(s_i)\right),\numberthis\label{eq:hybridbarrier}
\end{equation*}
where $s_i = a_i^\top x - b_i$ are the slacks at point $x$. We denote the normalizing factor of $\phi$ by $\alpha_0\triangleq (\frac{m}{n})^{\frac{2/p}{1+2/p}}$. 
\end{definition}
For background on Lewis weights see Section~\ref{sec:prelim}.
Our main theorem is a bound on the mixing rate of RHMC with this hybrid barrier. 

\begin{theorem}[Mixing]\label{thm:mixing}
Given a polytope 
$\mathcal P$, let $\pi$ be the distribution with density proportional to $e^{-\alpha \phi(x)}$ over the open set inside $\mathcal P$. Then,
RHMC with stationary distribution $\pi$ on the manifold of the open set inside $P$ equipped with metric $g$ defined by the Hessian of the hybrid barrier $\phi$ 
with $p=4-(1/\log(m))$ has mixing rate bounded by 
\[
\min\{\alpha^{-1}n^{2/3}+\alpha^{-1/3}n^{5/9}m^{1/9} + n^{1/3}m^{1/6}, n^{4/3}m^{1/3}\}.
\]
In particular, for the uniform distribution over $\mathcal P$ (with $\alpha=0$), the mixing rate is
\begin{align*}
    \tilde O\left(m^{1/3}n^{4/3}\right).
\end{align*}
More specifically, 
the Markov chain starting at $\pi_0$ reaches $\pi_t$ with TV-distance at most $\epsilon$ to the target after
\begin{align*}
    O\left(m^{1/3}n^{4/3}\log(M/\epsilon) \log(M)\right)
\end{align*}
steps, where $M \triangleq sup_{x \in P} \frac{d\pi_0(x)}{d\pi(x)}$ and $\tilde O$ hide $\polylog(m)$ factors.
\end{theorem}

Note that without a warm start, the $\log(M)$ dependence in Theorem~\ref{thm:mixing} could be another factor of $n$ to the mixing time. However, applying the Gaussian Cooling framework~\cite{cousins2018gaussian} extended to manifolds~\cite{lee2018convergence} lets us sample from $e^{-\alpha \phi}$ for any $\alpha$ without a warm start penalty, and also allows us to compute the volume of the polytope without a significant overhead.

\begin{corollary}[Any start; Volume]\label{cor:warmstart}
For the manifold Gaussian Cooling scheme in~\cite{lee2018convergence} with the hybrid barrier \eqref{eq:hybridbarrier} applied to sample from the density $e^{-\alpha \phi(x)}$ inside a given polytope starting from $\arg\min \phi(x)$, the total number of RHMC steps for any $\alpha \geq 0$ is bounded by
\begin{align*}
\tilde{O}\left(m^{1/3}n^{4/3}\log(1/\epsilon)\right),
\end{align*}
Moreover, to compute the integral of $e^{-\alpha \phi}$ in the polytope and in particular the volume of the polytope up to multiplicative error $1\pm \epsilon'$, the total number of RHMC steps is bounded by $\tilde O(m^{1/3}n^{4/3}/\epsilon'^2)$.
\end{corollary}

This improves on the previous best bound of $mn^{2/3}$ due to \cite{lee2018convergence} based on the standard logarithmic barrier. The proof of Theorem~\ref{thm:mixing} requires the development of several technical ingredients. We summarize a few that are likely to be of independent interest.

The first is a new isoperimetric inequality for this hybrid barrier (see Section~\ref{sec:markovchains} for the definition of isoperimetry).

\begin{theorem}\label{thm:hybrid-iso}[Isoperimetry of Hybrid Barrier]
Let $g$ be a metric corresponding to Hessian of the hybrid barrier, with support given by a polytope defined by $m$ inequalities in $\R^n$. 

Then for $\alpha \ge 0$, the distribution with density proportional to $e^{-\alpha\phi}$  has isoperimetric constant at least 
\[
\max\{\frac{1}{\sqrt{n}}(\frac{n}{m})^{\frac{1}{p+2}} , \sqrt{\alpha}\}.
\]
\end{theorem}

As part of the proof, we develop stronger self-concordance properties of the Lewis weight barrier. The usual self-concordance~\cite{nesterov1994interior} for barrier $\phi$ implies a control on the third order derivative of $\phi$ by its second derivative, which can be seen as a property of the metric $g = \nabla^2 \phi'$,
\[
-\|v\|_g g \preccurlyeq \D g(v) \preccurlyeq \|v\|_g g,
\]
where $\D g(v)$ is the directional derivative of $g$ along direction $v$. 
We will need to extend this self-concordance to third-order derivatives of $g$. 
These types of estimates for the derivatives of the metric are known as Calabi estimates in the Differential Geometry literature~\cite{szekelyhidi2014introduction,wang2006schauder}.

\begin{lemma}[Manifold self-concordance of Hybrid barrier]\label{lem:calabitype}
The hybrid barrier is third-order self-concordant with respect to the manifold's metric $g$, namely 
\begin{gather*}
    -\|v\|_{g}g \preccurlyeq \D g(v) \preccurlyeq \|v\|_{g} g\\ 
    -\|v\|_{g}\|z\|_{g} g \preccurlyeq \D^2g(v,z) \preccurlyeq \|v\|_{g}\|z\|_{g} g\\
    -\|v\|_{g}\|z\|_{g}\|u\|_{g} g \preccurlyeq \D^3g(v,z,u) \preccurlyeq \|v\|_{g}\|z\|_{g}\|u\|_{g} g.
\end{gather*}
\end{lemma}
Here $\preccurlyeq$ is the L\"{o}wner ordering between matrices ignoring logarithmic factors.
The Calabi-type estimates in Lemma~\ref{lem:calabitype} turn out to be insufficient to improve the mixing rate. Hence, as one of our main contributions, we develop a new type of self-concordance, where instead of the local norm $\|.\|_g$, we measure the spectral change of the metric in a different local norm $\|.\|_{x,\infty}$. An intuitive description of $\|.\|_{x,\infty}$ is via its unit ball; namely, $\|.\|_{x,\infty}$ is the unique norm whose unit ball is the symmetrized polytope $\mathcal P \cap 2x - \mathcal P$ around $x$, as illustrated in Figure~\ref{fig:symmetrizedpolytope}. ($2x - \mathcal P$ is the reflection of $\mathcal P$ around $x$.) 
\begin{lemma}[Infinity norm Third-order Self-concordance of Hybrid barrier]\label{lem:thirdorderself}
The hybrid barrier, defined in~\eqref{eq:hybridbarrier}, is third-order self-concordant with respect to the local infinity norm $\|.\|_{x,\infty}$. Namely, 
\begin{gather*}
    -\|v\|_{x,\infty}g \preccurlyeq \D g(v) \preccurlyeq \|v\|_{x,\infty} g,\\ 
    -\|v\|_{x,\infty}\|z\|_{x,\infty} g \preccurlyeq \D^2 g(v,z) \preccurlyeq \|v\|_{x,\infty}\|z\|_{x,\infty} g,\\
    -\|v\|_{x,\infty}\|z\|_{x,\infty}\|u\|_{x,\infty} g \preccurlyeq \D^3g(v,z,u) \preccurlyeq \|v\|_{x,\infty}\|z\|_{x,\infty}\|u\|_{x,\infty} g.\numberthis\label{eq:infinityselfconcordance}
\end{gather*}
\end{lemma}

In fact, the norm $\|.\|_{x,\infty}$ measures the ratio of the change of the distance to the $i$th facet after taking step $v$ divided by the distance to facet $i$, then taking maximum of this ratio over all facets. 
These estimates will allow us to prove important smoothness properties of certain quantities on the manifold that we are interested in. In the following, we sometimes refer to our notion of strong third-order self-concordance as infinity norm self-concordance, as it involves the local norm $\|.\|_{x,\infty}$.

\begin{figure}[th]
\centering
\begin{subfigure}{0.48\textwidth}
\centering
\includegraphics[height=1.23\linewidth]{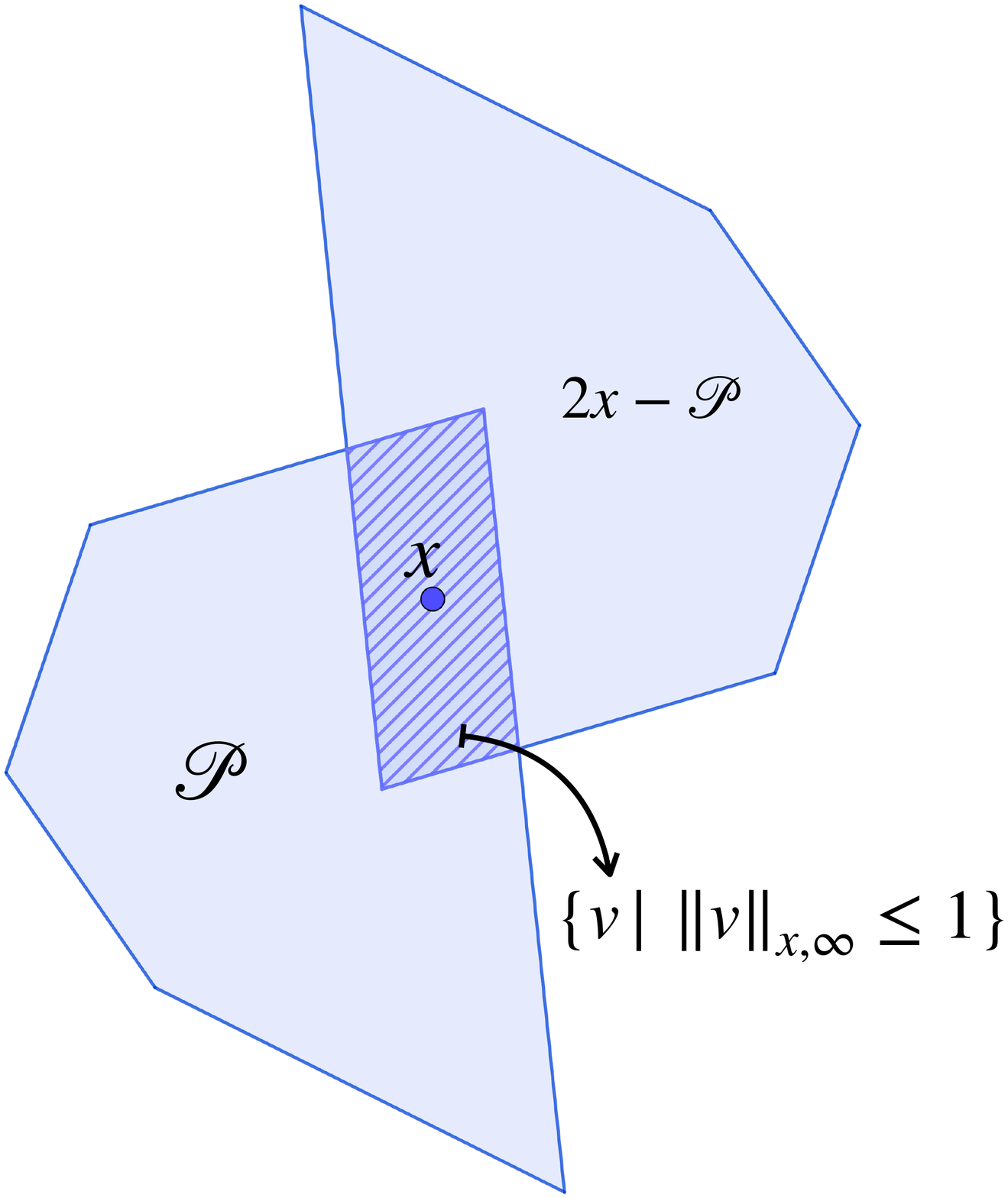}
\caption{The unit ball of the local norm $\|.\|_{x,\infty}$ is the symmetrized polytope around $x \in \mathcal P$.}
\label{fig:symmetrizedpolytope}
\end{subfigure}
\begin{subfigure}{0.48\textwidth}
\centering
\includegraphics[width=0.95\linewidth]{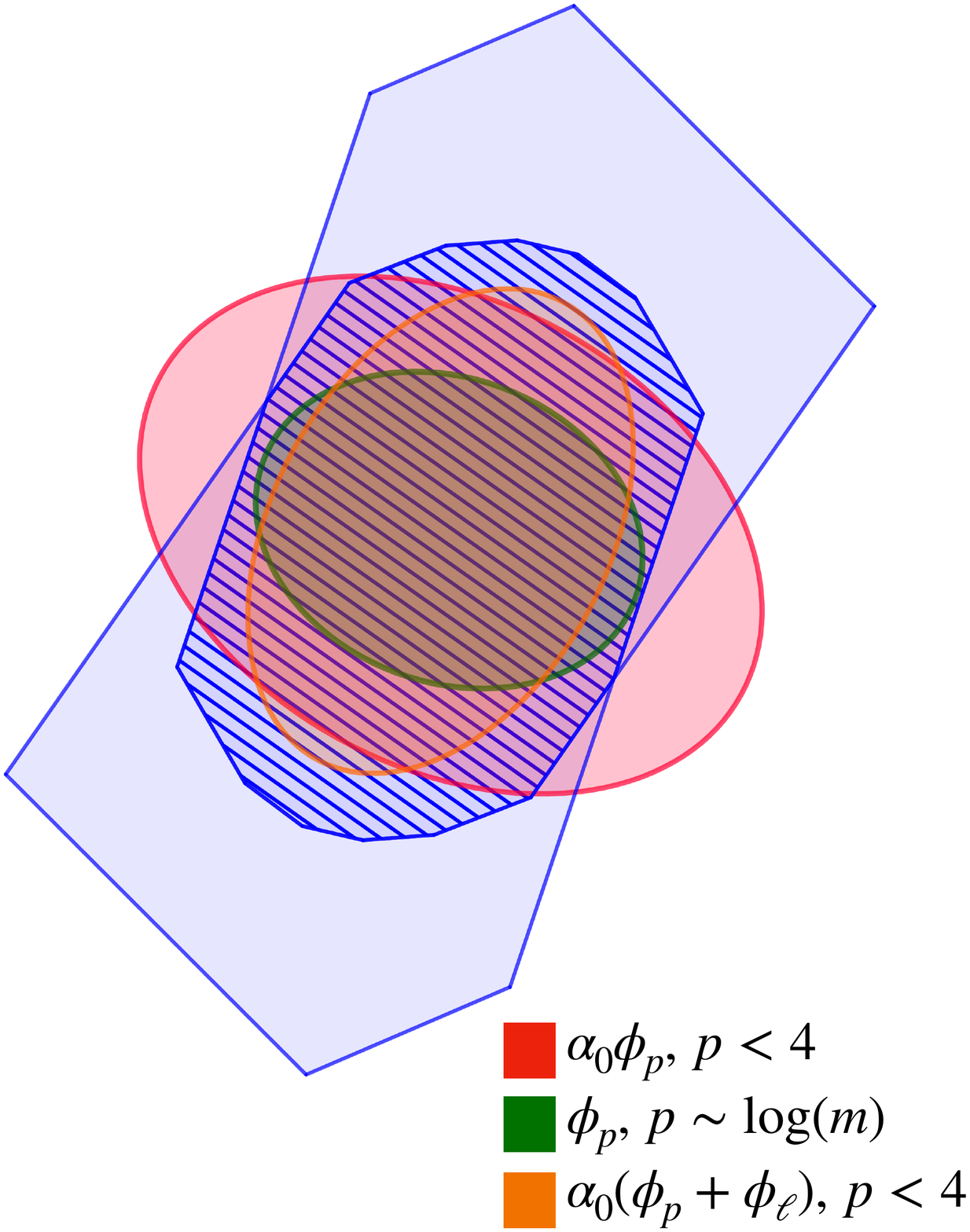}
\caption{The Lewis weights barrier with $p\simeq \log(n)$, with $p < 4$, and our hybrid barrier $\phi$ which is regularized with the log barrier.}
\label{fig:comparisonofellipsoids}
\end{subfigure}
\end{figure}


\subsection{Technical overview}\label{sec:technicaloverview}

\paragraph{Mixing and Conductance.} Our general approach to bounding the mixing rate is based on bounding the conductance~\cite{lovasz1993random}. 
The standard approach to bounding the conductance of geometric walks of this type
is to show an isoperimetric inequality for the underlying metric space and then prove that steps of the random walk behave well with respect to the underlying metric.  
Formally, we show two properties for the manifold $\mathcal M$ obtained by equipping the interior of the polytope $\mathcal P$ with the metric $g = \nabla^2 \phi$: 
\begin{itemize}
   \item {\bf Isoperimetry.} 
The target density  $e^{-\alpha \phi(x)}$ has a good isoperimetry constant on $\mathcal M$.
 \item {\bf One-step Coupling.} The one-step distributions of the Markov chain given two close-by points $x_0, x_1$ on the manifold are close in TV-distance. Namely, for some parameter $\delta > 0$, after excluding a tiny set $S^c \subseteq \mathcal M$, given any two points $x_0, x_1 \in \mathcal S$ with $d(x_0, x_1) \leq \delta$   we show 
    \begin{align}
     TV(\mathcal{T}_{x_0}, \mathcal{T}_{x_1}) \leq 0.01,\label{eq:onestepcoupling}
    \end{align}
    where  $\mathcal{T}_x$ denotes the Markov kernel starting from $x$.   
\end{itemize}

\paragraph{Isoperimetry.}  The log barrier metric gives an isoperimetric coefficient of $1/\sqrt{m}$, which leads to a factor of $m$ in the conductance. In principle, this can be improved to $\tilde{O}(n)$ by using a barrier with barrier parameter $\nu = \tilde{O}(n)$, as the general bound on the isoperimetry is $1/\sqrt{\nu}$ for any strongly self-concordant barrier with barrier parameter $\nu$~\cite{laddha2020convergence}. While the universal and entropic barriers have $\nu=O(n)$, they are expensive to compute. The LS barrier~\cite{lee2014path} has $\nu=\tilde{O}(n)$ while being efficient to compute. However, as we will see in more detail, as far as we know, the metric and its derivatives are not ``smooth" enough in most of the directions in the tangent space, which means we would have to take rather small steps while running RHMC. 

We will prove that the hybrid barrier has significantly better isoperimetry (Thm.~\ref{thm:hybrid-iso}) than the log barrier while maintaining sufficient smoothness.

\paragraph{Smoothness of Hamiltonian Curves and Comparison Geometry.}
The starting point of our analysis is the fact that one can look at the ordinary differential equation of RHMC in Equation~\eqref{eq:ham} as a second-order ODE on the manifold $\mathcal M$ of the open set inside the polytope with metric $g$. We will introduce this alternative form shortly. Looking at the Markov Kernel $\mathcal{T}_{x_0}$ of RHMC for a fixed point $x_0$, the randomness to define this kernel comes from the initial velocity $v_0$, which can be viewed as a vector on the tangent space of $x_0$ on the manifold $\mathcal M$ distributed as a standard Gaussian with respect to the local metric, namely $\mathcal N(0,g(x)^{-1})$ in the Euclidean chart. In order to show the One-step Coupling (Lemma ~\ref{eq:onestepcoupling}) for the Markov kernel of RHMC, we bound the difference between the densities $\mathcal{T}_{x_0}(y)$ and $\mathcal{T}_{x_1}(y)$ at a given point $y$ on the manifold. These densities are the pushforwards of the Gaussian density in the tangent space of $x_0$ and $x_1$ respectively, onto the manifold through the Hamiltonian map $Ham^{\delta}(x_0, v_{x_0})$ for some fixed time $\delta$, which maps the initial velocity $v_{x_0}$ to the solution of the ODE $y = x(\delta)$ at time $\delta$. The key to bound the change of density is to understand how the Hamiltonian curves vary as we change the initial point from $x_0$ to $x_1$ for a  fixed destination $y$, given the particular geometry imposed by our hybrid barrier inside a polytope. In fact, understanding the extremal scenarios of the behavior of geometric quantities on a certain class of manifolds is the topic of Comparison Geometry~\cite{cheeger1975comparison}~\cite{petersen2006riemannian}~\cite{ballmann2000riemannian}. In particular, to argue that the Hamiltonian curve changes sufficiently slowly, we need the metric $g$ of the manifold and its derivatives to be ``stable". 
 The simplest form of stability of the metric is the so-called self-concordance property, namely, $g$ is self-concordant if the derivative of $g(x)$ in a unit direction in the tangent space is controlled by $g$ itself. This type of self-concordance for the first derivative of the metric is already known for the $p$-Lewis weights barrier~\cite{lee2019solving}. However, this notion of stability is too weak for our use since a typical Gaussian vector $v$ in the tangent space of $x$ has norm of order $\|v\|_g \sim \sqrt n$. Nonetheless, one can hope to obtain estimates for $Dg(v)$ with respect to a different norm whose value is typically much smaller than the $\|.\|_g$ norm. 
 We show that self-concordance of the metric of the $p$-Lewis weights barrier for $p < 4$ with respect to the infinity norm of a re-parameterized version of $v$ is effective for characterizing the stability of Hamiltonian curves. 
 This local infinity norm, which we denote by $\|.\|_{x,\infty}$, can be regarded as the maximum ratio of the length of $v$ projected onto the normal of a facet divided by the distance of $x$ from that facet; its unit ball is the symmetrized polytope $\mathcal P \cap 2x-\mathcal P$ around $x$. Importantly, 
 one can see that for a typical Gaussian vector $v \sim \mathcal N(0,g^{-1})$, $\|v\|_{x,\infty}$ is of order $\tilde O(1)$ instead of $\sqrt n$. In fact, the $\|.\|_{x,\infty}$ norm of the tangent vector to the RHMC curve remains small for all times with high probability. This is favorable as we need a bound on the rate of change of the density only for typical values of $v$ and can ignore sets with small probability in bounding the conductance.  
 An important part of our contribution is to derive self-concordance estimates for the derivatives of the metric of the $p$-Lewis weights for $p < 4$ up to third order, with respect to this $\|.\|_{x,\infty}$ local norm. We introduce our approach up to second order self-concordance in Section~\ref{sec:secondorderselfconcordance} and defer the third-order self-concordance to Appendix~\ref{sec:thirdorderselfconcordance}. 
Although the number of terms that are created from differentiating the Lewis weights metric up to third order grows quite large, many subtensors are common, which enables us to treat in a similar fashion. 
To avoid repetition, we gather the common L\"{o}wner inequalities that we use for various matrices in section~\ref{sec:derivativestability} which we reuse to prove the self-concordance of the $p$ Lewis weights barrier. The infinity norm third-order self-concordance of the hybrid barrier follows from combining the infinity norm third-order self-concordance of the $p$-Lewis weights barrier and the log barrier (see section~\ref{sec:secondorderselfconcordance}).

The $p < 4$ threshold is essential to obtain our estimates. In particular, we can still control the derivative of the metric $Dg(v)$ with respect to $\|v\|_g$ for the LS barrier, which is a $p$ Lewis weights barrier for polylogarithmically large $p$, but it is an overestimate of the $\|.\|_{x,\infty}$ norm with high probability for a Gaussian vector in the tangent space of $x$. 
Nonetheless, for small $p$'s the ellipsoid of the $p$-Lewis weights does not approximate the symmetrized polytope as well as larger $p$'s; in particular a large portion of the ellipsoid lies outside the symmetrized polytope. This means that we need to scale down the unit norm ellipsoid so that it fits inside the polytope, which then means we have to to scale it up by a larger constant to make it contain the symmetrized polytope. As a result, the barrier parameter is large (see~\cite{laddha2020strong} for definition of barrier parameter), which in turn results in a poor isoperimetric constant.

We would like to have an ellipsoid at each point $x$ inside the polytope that approximates the symmetrized polytope around $x$ more accurately and is also stable as $x$ moves in random directions. For this, we go back to an idea of Vaidya from optimization and use a hybrid barrier by ``regularizing" the $p$-Lewis weight barrier for $p < 4$ with the standard log barrier
We can give a better bound on the barrier parameter of this hybrid barrier compared to the log barrier, which implies that the corresponding metric has better isoperimetry. Moroever, the regularization does not harm the stability of the metric as the log barrier already enjoys stability with respect to the local infinity norm $\|\|_{x,\infty}$. In particular, we show that our hybrid barrier has stable higher-order derivatives in arbitrary directions based on the local norm $\|.\|_{x,\infty}$. The particular choice of our barrier is essential to simultaneously prove third order infinity-norm self-concordance and good isoperimetry.

\paragraph{Hamiltonian curves and variations.} 
To see the high-level idea of how we show the one-step coupling of the Markov kernel, consider the shortest path between two points $x_0$ and $x_1$, which is a geodesic on the manifold. Geodesics are generalization of straight lines in the Euclidean space to arbitrary manifolds and naturally define the curve with the smallest possible length between two points on the manifold. Let the curve $\gamma_s$, parameterized by $s \in [0,s']$, be a length-minimizing geodesic connecting $x_0 = \gamma_0$ to $x_1 = \gamma_{s'}$ with distance $d(x_0, x_1)$. 
Suppose that running the Hamiltonian ODE with initial location $x_0 \in \mathcal P$ and initial velocity $v_{x_0}$ up to time $\delta$ takes us to a point $y$ on the manifold. As we start moving toward $x_1$ on the geodesic, $\gamma_s$ parameterized by $s \in [0,s']$, we consider the variation of the initial Hamiltonian curve; namely a family of Hamiltonian curves parameterized by $s$, where the $s$-curve starts from point $\gamma_s$, perhaps with a different initial velocity $v_{\gamma_s}$, but ends up to the same destination $y$ at time $\delta$.  The geodesic $\gamma_s$ from $x_0$ to $x_1$ and the corresponding Hamiltonian curves are illustrated in Figure~\ref{fig:hamiltonianfamily}. 

Looking at the the value of the density $\mathcal{T}_{\gamma_s}(y)$ at point $y$ after taking one step of the Markov chain starting from $\gamma_s$, we observe it depends on two major components: (1) the Gaussian density of the initial velocity $v_{\gamma_s}$ which is proportional to $\exp{\{-\frac{\|v_{\gamma_s}\|_{g}^2}{2}\}}$, and (2) the determinant of the Jacobian or the differential of the map from the initial velocity $v_{\gamma_s}$ to the destination point $y$, denoted by $J^{v_{\gamma_s}}_{y}$. Therefore, to study how quickly the density $\mathcal{T}_{\gamma_s}(y)$ changes from $x_0$ to $x_1$, we need to study the rate of change of the initial velocities $v_{\gamma_s}$ and the Jacobians $J^{v_{\gamma_s}}_{y}$; the latter  will depend on the rate of change of the Ricci tensor on the manifold. To study the variation of the Hamiltonian curve, we start by defining these manifold concepts.

As we mentioned earlier, one can identify the location variable $x$ in the Hamiltonian ODE~\eqref{eq:ham} as a point on the manifold $\mathcal M$ with metric $g$, and the velocity variable $v$ as a vector in the tangent space of $x$, $T_x(\mathcal M)$. Then, one can write the Hamiltonian ODE in Equation~\eqref{eq:ham} as a second-order ODE on the manifold $\mathcal M$ using the covariant derivative of $\mathcal M$, illustrated in Lemma~\ref{lem:ham_curve}. For background on Riemannian geometry and covariant differentiation, we refer the reader to Appendix~\ref{sec:riemanniangeometry}.

\begin{lemma}\label{lem:ham_curve} 
The Hamiltonian ODE in Equation~\ref{eq:ham} can be written using the covariant derivative of the manifold in a simplified form:
\begin{align}
    \nabla_{\gamma'(t)} \gamma'(t) = \mu(\gamma(t)).\label{eq:hamiltonianmanifold}
\end{align}
Above, $\nabla$ is the covariant derivative and $\mu(x)$ is the bias (drift) vector field of the Hamiltonian curve, defined as
\begin{align}
    \mu(x) \triangleq g^{-1}\D f(x) - \frac{1}{2}g(x)^{-1}\tr[g(x)^{-1}\D g(x)],\label{eq:biasdefinition}.
\end{align}
\end{lemma}
In the above notation, $\tr[g(x)^{-1}\D g(x)]$ is a vector whose $i$th entry is $\tr[g(x)^{-1}\D_ig(x)]$. See Appendix~\ref{sec:hmc} for a proof of Lemma~\ref{lem:ham_curve}. The above ODE~\eqref{eq:hamiltonianmanifold} for Hamiltonian curves is similar to the second order ODE for geodesics; for the latter the bias vector $\mu$ is zero, i.e., the geodesic Equation is given by~\cite{do2016differential}
\begin{align}
    \nabla_{\gamma'(t)}\gamma'(t) = 0.\label{eq:geodesiceq}
\end{align}
In physics, the Hamiltonian ODE in Equation~\ref{eq:hamiltonianmanifold} is important as it models the motion of a particle on a manifold acting under a force field devised by $\mu$. Next, we define the notion of a family of Hamiltonian curves.


\begin{definition}[Family of Hamiltonian curves]
We say $\big(\gamma_s(t)\big)$ is a family of Hamiltonian curves ending at some fixed $y$ whose starting point varies from $x_0 = \gamma_0(0)$ to $x_1 = \gamma_{s_1}(0)$ if for every fixed time $0 \leq s \leq s_1$, $\gamma_s(t)$ is a Hamiltonian curve in $t$, and $\gamma_s(0)$ as a function of $s$ is a geodesic on $\mathcal M$ from $x_0$ to $x_1$. Unless specified otherwise, whenever we talk about the curve $\gamma_s(t)$ we mean the curve $\gamma_s(t)$ as a function of $t$ for a fixed $s$. We write $\gamma'_s(t) = \partial_t \gamma_s(t)$ to refer to the derivative of the curve with respect to $t$.
\end{definition}

Before studying the variations of Hamiltonian fields, to given some high level intuition, we start by variations of geodesics here. More precisely, suppose $\gamma_s(t)$ is a variation of geodesics, i.e. $\gamma_s(t)$ is a geodesic in $t$ for every fixed $s \in [0,s']$ (recall that the curve $\gamma_0(s)$ in parameter $s$ is also a geodesic from $x_0$ to $x_1$). For brevity, we sometimes refer to the curve $\gamma_0(t)$ by $\gamma(t)$. To see how fast the geodesics $\gamma_s(t)$ changes as a function of $s$ at time $s=0$, for a fixed $t$ we take the derivative of $\gamma_s(t)$ with respect to $s$ at time $s=0$; this gives us a vector field $J(t)$ along $\gamma_0(t)$:
\begin{align*}
    J(t) = \partial_s \gamma_0(t) = \partial_s \gamma_s(t)\Big|_{s = 0},
\end{align*}
This vector field, called a {\em Jacobi field}, is a fundamental object in studying the variations of geodesics. Importantly, one can write a second-order ODE to describe how $J(t)$ evolves along the geodesic given initial conditions $J(0), J'(0)$
\begin{align}
    D^2_tJ(t) = R(J,\gamma'(t))\gamma'(t).\label{eq:jacobiODE}
\end{align}
where the second derivative $J''(t)$ is the covariant derivative on the manifold with respect to $\gamma'_0(t)$, i.e., $D_t \triangleq \nabla_{\gamma'_0(t)}$, and $R$ is the Riemann tensor. We will provide some intuition on the role of Riemann tensor and its role in the behavior of geodesics presently. An important point to observe here is that the covariant derivative of $J$ at $t = 0$ is equal to the covariant derivative of the initial velocity of the geodesic, namely $\frac{d}{dt}\gamma_s(t)$, with respect to $s$ (see Lemma~\ref{lem:commutingderivatives} for a proof):

\begin{align}
    J'(0) = D_s \frac{d}{dt}\gamma_s(t)\Big|_{s=0,t = 0} = D_s v_{\gamma_s}\Big|_{s=0}.\label{eq:jacobiinitial}
\end{align}
So the initial values that uniquely specify the Jacobi field $J$ are $J(0)$, which specifies how fast we change the starting point of the geodesic, and $D_s v_{\gamma_s}$, which is how fast we change the initial velocity of the geodesic. This means that one can study the Jacobi field ODE to obtain estimates on how fast the initial velocity should change along the geodesic from $x_0$ to $x_1$, for this family of Hamiltonian curves with the same destination $y$. Now consider a direction $e$ perpendicular to the velocity $\gamma'(t) = \gamma'_0(t)$ of the geodesic at time $t$, i.e., $\langle \gamma'(t), e\rangle_g = 0$. Looking at the dot product of the vector $R(e,\gamma'(t))\gamma'(t)$ on the right hand side of the Jacobi field ODE in~\eqref{eq:jacobiODE} to $e$ itself, the quantity $\langle e, R(e,\gamma'(t))\gamma'(t) \rangle$ is intuitively measuring how much the Jacobi field is growing or shrinking in direction $e$, meaning whether the geodesics $\gamma_s(t)$ parameterized by $s$ are converging or diverging in direction $e$ at time $s=0$. This quantity is known as the sectional curvature of the plane spanned by $e$ and $\gamma'(t)$. Now consider a unit orthonormal parallelepiped at time $t = 0$, denoted by a set of orthonormal vectors $\{e_i\}_{i=1}^n$ in the tangent space of $\gamma(0)$, where $e_1 = \gamma'(0)$, and look at the evolution of its volume along the geodesic when each $e_i$ evolves according to the Jacobi Equation; in each directions $e_i$, the parallelepiped is either expanding or squeezing, depending on if the geodesics are converging or diverging in that direction which depends on the sign of the sectional curvature $\langle e_i, R(e_i, \gamma'(0))\gamma'(0) \rangle$. Indeed, one can characterize the rate of change of this parallelepiped along the geodesic by summing the sectional curvatures for all $\{e_i\}_{i=2}^n$; this is the Ricci curvature of the manifold at $\gamma(0)$ in the direction $\gamma'(0)$:
\begin{align*}
    \Ricci(\gamma'(0),\gamma'(0)) = \sum_{i=1}^n \langle e_i, R(e_i, \gamma'(0))\gamma'(0)\rangle. 
\end{align*}
Note that the Ricci curvature is nothing but the trace of the Riemann tensor $R(.,\gamma'(0))\gamma'(0)$.
On the other hand, the determinant of the Jacobian $J^{v_{\gamma_s}}_y$ of the Hamiltonian map, a quantity of our interest to bound the change of density from $x_0$ to $x_1$, can be characterized by the ratio of the volume of this parallelepiped at the beginning and the ending time $t$.
Indeed, we see later on that the log determinant of $J^{v_{\gamma_s}}_y$ can be written as a time-weighted integral of the Ricci curvature along the geodesic. 


One can extend these arguments to variations of Hamiltonian curves instead of geodesics. As a result, instead of the Riemann tensor in the Jacobi fields Equation~\eqref{eq:jacobiODE}, we end up with a slightly different operator $\Phi(t)$ which can be decomposed into a ``geometric part," the Riemann tensor, and a ``bias part," $M_x$, which comes from the derivative of the Hamiltonian bias $\mu(x)$, defined in Equation~\eqref{eq:biasdefinition}. We define this fundamental operator rigorously.
\begin{definition}[Operators $\Phi$ and $M_x$]\label{def:operatorPhi}
    At any point $x \in \mathcal M$, we define the operator $M_x$ as
    \begin{align*}
        \forall u \in T_x(\mathcal M), M_x(u) \triangleq \nabla_u \mu(x),
    \end{align*}
    where $\nabla$ is the covariant derivative on the manifold and $\mu$ is the Hamiltonian bias.
    Given the Hamiltonian curve $\gamma(t)$, we define the operator $\Phi(t)$ on the tangent space $T_{\gamma(t)}(\mathcal M)$ as
    \begin{align*}
        \Phi(t) \triangleq R(.,\gamma'(t))\gamma'(t) + M_{\gamma'(t)}.
    \end{align*}
    where $R$ is the Riemann tensor.
\end{definition}
 
Similar to Jacobi fields, for a given family of Hamiltonian curves $(\gamma_s(t))$, one can write a second order ODE for the variational vector field $\tilde{J}(t) = \frac{d}{ds}\gamma_s(t)$ along the Hamiltonian curve, which depends on operator $\Phi$ (for the proof see Appendix~\ref{sec:hmc}):

\begin{lemma}[ODE for Hamiltonian fields]\label{lem:hmcfieldODE}
    Given a family of Hamiltonian curves $\big(\gamma_s(t)\big)$, the vector field $\tilde J(t) \triangleq  \partial_s \gamma_s(t)\Big|_{s = 0}$ is characterized by the following second order ODE:
    \begin{align}
        \tilde J''(t) = \Phi(t)\tilde J(t),\label{eq:hamilvarODE}
    \end{align}
    where $\Phi(t)$ is defined in ~\ref{def:operatorPhi}. We refer to $\tilde J$ as a Hamiltonian field.
\end{lemma}
The difference between the ODE of Hamiltonian fields~\ref{eq:hamilvarODE} and that of Jacobi fields~\ref{eq:jacobiODE} comes from the fact that the primary Hamiltonian Equation~\eqref{eq:hamiltonianmanifold} includes an additional bias vector $\mu$ compared to the geodesic Equation~\eqref{eq:geodesiceq}.

Now similar to the case of variations of geodesics, for variation of Hamiltonian curves, the log determinant of the Jacobian of the Hamiltonian map $J_y^{v_{\gamma_s}}$ can be characterized by a weighted integral of the trace of $\Phi(t)$ instead of the Ricci tensor. Therefore, to study the rate of change of $\det (J_y^{v_{\gamma_s}})$ as we move from $x_0$ to $x_1$, we need to study the rate of change of $\tr(\Phi(t))$ along the variation of Hamiltonian curves $(\gamma_s(t))$, which in turn depends on the rate of change of the Ricci tensor and the trace of operator $M_x$, the two parts of the operator $\Phi(t)$. These ideas are formalized as the $(R_1, R_2, R_3)$-normality of the Hamiltonian curve in the definition below.

\begin{figure}[t]
\includegraphics[scale=0.54]{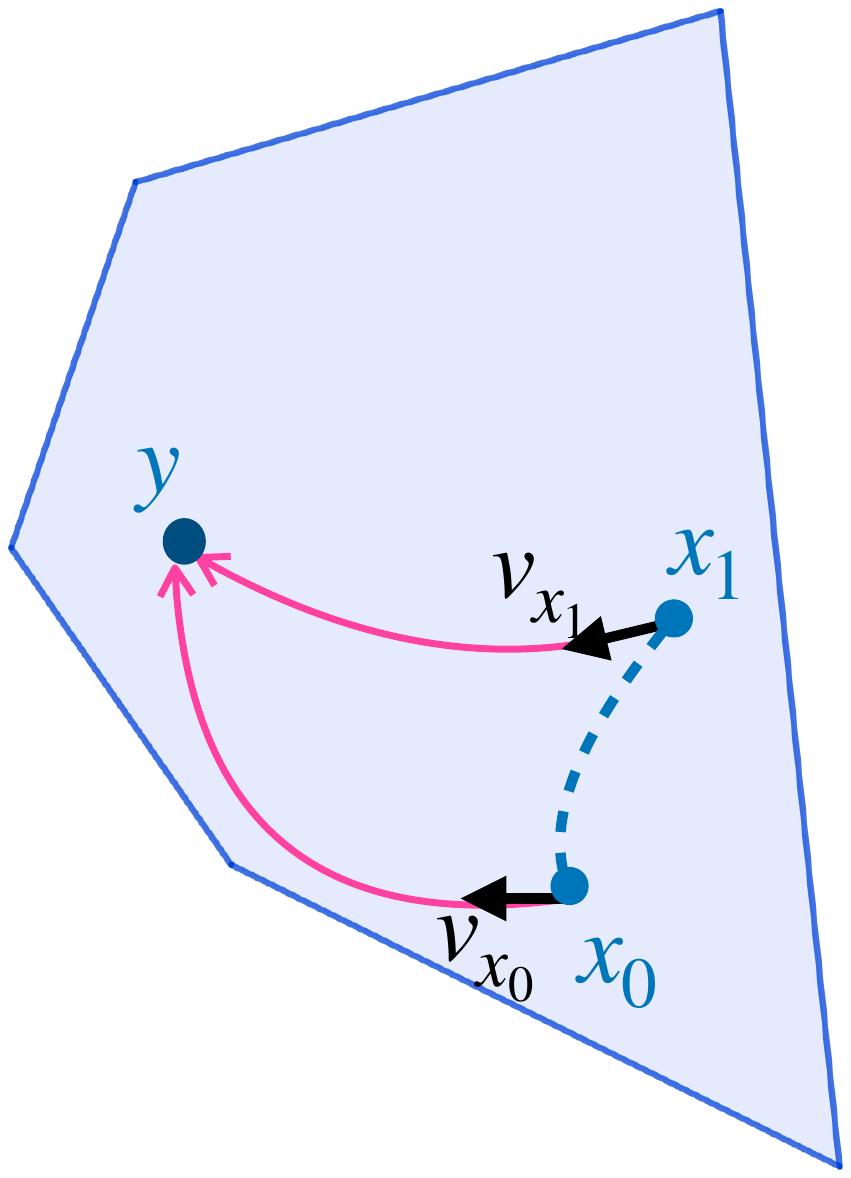}
\centering
\caption{Family of Hamiltonian curves $\gamma_s(t)$ all ending in $y$ with starting point varying from $x_0 = \gamma_0(0)$ to $x_1 = \gamma_{s'}(0)$, where $\gamma_s(0)$ is a geodesic in $s$.}
\label{fig:hamiltonianfamily}
\end{figure}

\begin{definition}\label{def:normality}
We say a Hamiltonian curve $\gamma(t)$ is $(R_1, R_2, R_3)$-normal up to time $\delta$ if for all $0 \leq t \leq \delta$ if
it satisfies the following:
\begin{itemize}
    \item Bound on the Frobenius norm of $\Phi$ (with respect to the metric $g$):
    $$\|\Phi(t)\|_F \leq R_1.$$
    \item For all times $0 \leq t \leq \delta$ and unit direction $z$ in the tangent space of $\gamma(t)$:
    \begin{align*}
        |D(tr(\Phi(t)))(z)| \leq R_2\|z\|_g.
    \end{align*}
    \item  For $\zeta(t)$ defined as the parallel transport of $\gamma'(0)$ along the curve:
    \begin{align*}
        \|\Phi(t)\zeta(t)\|_g \leq R_3.
    \end{align*}
\end{itemize}
\end{definition}

Parallel transport of a vector on the manifold is a generalization of shifting vectors in Euclidean space, using the covariant derivative of the manifold (see Appendix~\ref{sec:riemanniangeometry} for the rigorous definition.) In order to show the $(R_1, R_2, R_3)$-normal property for the family of Hamiltonian curves, we need to define a more fundamental regularity condition for the Hamiltonian curves which states that both $\|.\|_g$ and $\|.\|_{x,\infty}$ norms remain small for the tangent vector along the Hamiltonian curve. 
\begin{definition}[Nice Hamiltonian curve]\label{def:nicedef}
We say a Hamiltonian curve $\gamma(t)$ is $(\delta,c)$-nice if for $0 \leq t \leq \delta$:
\begin{align*}
    &\|\gamma'(t)\|_g \leq c\sqrt n,\\
    &\|\gamma'(t)\|_\infty \leq c.
\end{align*}
\end{definition}

In order to show the closeness of one step distributions between $x_0$ and $x_1$, we need the $(R_1, R_2, R_3)$-normality for the family of Hamiltonian curves $(\gamma_s(t))$ for all $0 \leq s \leq \delta$ as we defined in~\ref{def:normality}. Therefore, we need to show that the $(c,\delta)$-niceness property is stable for our hybrid barrier. We show this in Lemma~\ref{lem:stabilityinfnorm}, proved in Section~\ref{sec:stability}.
Our $(c,\delta)$-niceness framework is a simpler and more general framework and avoids the technical machinery of auxiliary functions on curves used in~\cite{lee2018convergence}, which needs additional parameters that need to be bounded.

\begin{lemma}[Stability of norms]\label{lem:stabilityinfnorm}
In the same setting as Theorem~\ref{thm:parameterbounds}, given a family of Hamiltonian curves $\gamma_s(t)$ for which $\gamma_0(t)$ is $(c,\delta)$-nice for 
\begin{align*}
    \delta \leq \delta' \triangleq \frac{1}{\sqrt{c^2 + \alpha \sqrt \alpha_0}n^{1/4}},
\end{align*}
then $(\gamma_s(t))$ is a $(O(c), \delta)$-nice family of Hamiltonian curves in the interval $s \in (0 , \delta)$.
\end{lemma}

A major part of our contribution is that we relate this abstract notion of $(R_1, R_2, R_3)$-normality to (a generalized notion of) metric self-concordance or Calabi-type estimates,  which (1) crucially uses a different notion of norm to bound the derivatives of the metric and (2) needs to be satisfied for higher derivatives of the metric up to third order. Our framework can potentially be reused on other manifolds and distributions. 

\begin{theorem}[Smoothness]\label{thm:parameterbounds}
Given a Hessian manifold defined by the metric $g = \nabla^2 \phi$ for our hybrid barrier (see Definition~\ref{def:hybridbarrier}) for $p < 4$, define a Hamiltonian curve $\gamma(t)$ by the ODE in Equation~\eqref{eq:hamiltonianmanifold} with target log density $f = \alpha \phi$. Assume that $\gamma$ is $(c,\delta)$-nice (see definition of niceness in~\ref{def:nicedef}), then it is also $(R_1, R_2, R_3)$-normal with parameters
\begin{align*}
&R_1 = (c^2 + \sqrt \alpha_0 \alpha)\sqrt{n},\\
&R_2 = (c^2 + \sqrt{\alpha_0}\alpha)n,\\
&R_3 = c^2(\sqrt n + n\delta) + n\delta c\alpha \sqrt{\alpha_0}.    
\end{align*}
\end{theorem}
\begin{proof}
The result follows from the key Lemmas~\ref{lem:R1bound},~\ref{lem:R2bound}, and~\ref{lem:R3bound}.
\end{proof}

 To understand the effect of self-concordance on the density of the push-forward measure, note that the more slowly the metric changes, the more slowly the geodesics will converge or diverge from one another, so we have smaller scalar and Ricci curvatures. 
 As an example, one can see that the Ricci curvature $\Ricci(\gamma'(t),\gamma'(t))$ can be written formally using the metric and its first derivative on Hessian manifolds (see Equation~\eqref{eq:riccitensor}).
As a result, the rate of change of the Ricci tensor, which corresponds to the $R_2$ parameter in Definition~\ref{thm:parameterbounds}, depends on the derivatives of the metric $g$ up to second order, and in particular can be bounded efficiently given that the metric satisfies some form of second-order self-concordance. In this regard, a question that comes up is the following: in which norm should we measure the self-concordance of the metric? 

A key to notice here is that in measuring the change of $\Ricci(\gamma'(t),\gamma'(t))$, the Ricci tensor itself involves the change of the metric in a \textbf{random direction} as we can show that $\gamma'(t)$, the tangent of the Hamiltonian curve, is distributed as a Gaussian. Now if one uses the conventional framework of self-concordance in optimization which measures the derivative of the metric in direction $v$ with respect to its local norm $\|v\|_g$, then the typical value of the quantity $\|\gamma'(t)\|_g$ is of order $\sqrt n$. This indicates a major reason we choose to measure self-concordance in the $\|.\|_{x,\infty}$ norm, which is $\tilde O(1)$ for a typical Gaussian vector $\mathcal N(0,g^{-1})$.
Importantly, we use our third-order infinity norm self-concordance in Lemma~\ref{lem:thirdorderself} in a black-box manner to show the $(R_1, R_2, R_3)$-normality of the Hamiltonian curve. 
On the other hand, even though the log barrier satisfies this type of self-concordance with respect to $\|.\|_{x,\infty}$, it does not approximate the local geometry of the polytope well, which results in poor isoperimetry and slow mixing. For this reason, we develop infinity-norm self-concordance for the $p$-Lewis weights barrier whose local ellipsoids are better approximations for the symmetrized polytope. Our approach to develop the infinity norm self-concordance estimates crucially depends on $p < 4$. Therefore, to further enhance the isoperimetry of the metric, we regularize the Lewis weights barrier with the log barrier, which results in our final hybrid barrier in Equation~\eqref{eq:hybridbarrier}.

\paragraph{Structure of the paper.} The rest of the paper is organized as follows: 
In Section~\ref{sec:prelim} we discuss the basic tools and notation that we use throughout the paper. In Section~\ref{sec:thirdorderselfconcordance}, we give our proof of second-order infinity norm self-concordance estimates for the Lewis weights barrier (we defer the proof of strong third-order self-concordance to Appendix~\ref{sec:thirdorderselfconcordance}). In Section~\ref{sec:boundingconductance}, we bound the mixing time by combining multiple components, namely the stability of the Hamiltonian curves, the isoperimetry of the stationary distribution with respect to the chosen metric, and the smoothness of the manifold with our hybrid barrier. We relate the change of density of the Markov kernel between two points to the smoothness of the manifold.   In Section~\ref{sec:smoothness}, we show how we use infinity-norm third-order self-concordance to control the smoothness of the metric. Namely, we bound the norm of an important operator $\Phi$ related to the Riemann tensor and the Hamiltonian potential, which appears in the ODE of variations of Hamiltonian curves (parameter $R_1$). To bound the determinant of the Jacobian of the RHMC map, which is a component in the pushforward density of the Gaussian distribution in the tangent space onto the manifold, we bound the rate of change of the trace of $\Phi$, which includes the Ricci tensor (parameter $R_2$) and another component originating from the Hamiltonian bias $\mu$. Finally, we bound the norm of $\Phi$ applied to the initial velocity of the Hamiltonian curve parallel transported along the curve. In Section~\ref{sec:stability}, we prove the stability of the smoothness properties of the Hamiltonian curves as we start varying the initial location and velocity of the curve. In Section~\ref{sec:isoperimetry}, we prove an isoperimetry inequality on the Riemannian manifold $\mathcal M$ equipped with metric $g$, the Hessian of our hybrid barrier. In Appendix~\ref{sec:riemanniangeometry}, we give some background on Differential Geometry. In Appendix~\ref{sec:hmc}, we describe how to derive the second order Hamiltonian ODE based on the covariant derivative on the manifold. In Appendix~\ref{sec:thirdorderselfconcordance}, we show the infinite-norm third-order self-concordance of the metric for our hybrid barrier~\eqref{eq:hybridbarrier}. Appendix~\ref{sec:derivativestability} is devoted to obtaining spectral bounds for the derivatives of our metric, which includes Lewis weights and its derivatives, which we use in our self-concordance arguments. Finally, in Appendix~\ref{sec:remainingproofs} we include missing proofs.

\section{Preliminaries}\label{sec:prelim}
To work with the metric $g$ imposed by our hybrid barrier $\phi$, it is convenient to rescale the rows of the LP matrix $\mathrm A$ by the slack variables, namely we define
\begin{align*}
    \A = \diag{\big( (a_i^\top x - b_i)^{-1}\big)_{i=1}^m} \mathrm A.
\end{align*}
 In our equations we treat hadamard product of matrices with higher priority, namely $A B\odot C$ is equivalent to $A(B\odot C)$. We refer to the $p$-Lewis weights vector of $A_x$ by $w_x$ and its diagonal matrix version by $\W \triangleq \diag{w_x}$. 
To work with a vector $v$ in the tangent space of $x$, there is an important reparameterization of $v$ defined as
\begin{align*}
    &s_{v,x} \triangleq \A v,
    \\
    &\Sxv \triangleq \diag{s_{x,v}}.\numberthis\label{eq:reparameterizedvector}
\end{align*}
Define the log barrier by $\phi_\ell$:
 \begin{align*}
 \phi_\ell(x) \triangleq -\sum_{i=1}^m \log(a_i^\top x - b_i).     
 \end{align*}
We denote the Hessian of the log barrier by $g_2 = \nabla^2 \phi_\ell(x)$. We see $g_1$ as a metric inside the polytope, such that for $v \in \mathbb R^n$ it defines a local metric $\|v\|_{g_2}^2 = v^\top g_2 v$. It is easy to check that the norm of a vector $v$ with respect to $g_2$, i.e. $v^\top g_2 v$, is given by the $\ell_2$ norm of the reparameterized vector $s_{x,v}$ defined in Equation~\eqref{eq:reparameterizedvector}.
\begin{align*}
    v^\top g_2 v = v^\top {\A}^\top \A v = \|s_{x,v}\|_2^2.
\end{align*}
For a given point $x$ inside polytope $\mathcal P$, we define the symmetrized polytope $\mathcal P \cap 2x-\mathcal P$ around $x$ as the following: we reflect $\mathcal P$ around $x$ and intersect it with the $\mathcal P$ namely $\mathcal P \cap 2x - \mathcal P$, as illustrated in Figure~\ref{fig:symmetrizedpolytope}. The approximation of the symmetrized body by the ellipsoids corresponding to the Hessian of the barrier function plays a key role in bounding the isoperimetry constant, as we describe in Section~\ref{sec:isoperimetry}.

\subsection{John Ellipsoid and Lewis weights}
Proving good isoperimetry for a specific barrier can be reduced to how well the ellipsoids corresponding to the Hessian of the barrier at each point $x$ inside the polytope approximate the symmetrized polytope around $x$. A natural way to approximate a symmetric polytope is via its John Ellipsoid, i.e. the ellipsoid of maximum volume contained in the polytope. Parametrizing the  John ellipsoid as $A_x^\top W A_x$ for a positive diagonal matrix $W$, i.e., a weighted sum of the outer product of the rows of $A_x$, the weights are characterized by the following optimization problem:
\begin{align}
&\max_{w \in \mathbb R_{\geq 0}^n} \log\det(\A^\top \mathrm{W} \A) \label{eq:jonellipse}\\ 
&s.t. \ \indic^\top w = n. \nonumber
\end{align}
where $\mathrm W = \diag{w}$ is the diagonal matrix corresponding to the vector $w$. The John ellipsoid approximates the symmetrized polytope in the sense that (1) it is inside the ellipsoid and (2) scaling it up by $\sqrt n$ will make it contain the symmetrized polytope.

On the other hand,
in order to prove smoothness of the HMC curves, we need to pick a barrier whose Hessian does not change too fast as a function of $x$. Unfortunately the John ellipsoid is not stable. In particular, the weights $\mathrm W$ which maximizes~\eqref{eq:jonellipse} are not even continuous with respect to $x$. An alternative is to use the $p$-Lewis weights to define the ellipsoid, obtained as the solution to a relaxation of the program in~\eqref{eq:jonellipse}:
\begin{align}
    w_x \triangleq & \text{argmax}_{w \in \mathbb R_{\geq 0}^n} -\log\det(\A^\top \mathrm{W}^{1-2/p}\A) + (1-2/p)\indic^\top w,\label{eq:lewisweightsdef}
\end{align}
where $\mathrm{W} = \diag{w}$. Moreover, the optimal value of the program in~\eqref{eq:lewisweightsdef} is denoted by the $p$ Lewis weights barrier at $x$ as defined next.

\begin{definition}[Lewis weights barrier]
The $p$-Lewis weights barrier can be defined as the solution of the following optimization problem:
\begin{align}
    \phi_p(x) \triangleq \max_{w\in \mathbb R^n_{>0}}-\log\det (\A^\top \mathrm{W}^{1-2/p}\A) + (1-2/p) \indic^\top w,\label{eq:lewisweightdefinition}
\end{align}
\end{definition}
Let $g_1 = \nabla^2 \phi_p$ be the metric defined by the Hessian of the $p$ Lewis weights barrier. It is known (Lemma 31 in~\cite{lee2019solving}) that the ellipsoid corresponding to $g_1$ is roughly the same as the one defined by the Lewis weights, i.e. $\A^\top \W \A$.

\begin{lemma}[Lewis weights metric]\label{lem:lewisweightbarrier}\label{lem:lsbarrierapprox}
For the Lewis weight barrier $\phi_p$ we can bound the local norm of its Hessian as
\begin{align}
     \|s_{x,v}\|_{w}^2\leq v^\top g_1 v \leq (1+p)  \|s_{x,v}\|_{w}^2,\label{eq:spectralapprox}
\end{align}
where for a vector $s_{x,v}\in \R^m$,
\begin{align*}
    \|s_{x,v}\|_{w}^2 \triangleq \sum_{i=1}^m {w_x}_i {s_{x,v}}_i^2.
\end{align*}
Equivalently 
\begin{align*}
    \A^\top \W \A \preccurlyeq g_1 \preccurlyeq (1+p)\A^\top \W \A.
\end{align*}
\end{lemma}

Next, we define another important local norm at a point $x$ inside the polytope:
\begin{align*}
    \|v\|_{x,\infty} \triangleq \|s_{x,v}\|_\infty.
\end{align*}
This norm plays a key role in our definition of strong self-concordance in Equation~\eqref{eq:infinityselfconcordance}. 
For any point $x$ inside the polytope, we define $\Px$ to be the projection matrix of $\A$ reweighted by $\W^{1-2/p}$
\begin{definition}[Projection matrix]\label{def:projecrtion}
    we define the projection matrix $\Px$, implicitly depending on $x$, as
    \begin{align*}
    \Px \triangleq P(\W^{1/2-1/p}\A) \triangleq \W^{1/2 - 1/p}\A (\A^\top \W^{1-2/p}\A)^{-1} {\A}^\top \W^{1/2 - 1/p},
\end{align*}
where $\W$ is the $p$-Lewis weights calculated at $x$. Moreover, we denote the Hadamard square $\Px^{\odot 2}$ of the projection matrix by $P^{(2)}$:
\begin{align*}
    (\Ptwo)_{ij} \triangleq (\mathbf{P_x}^{\odot 2})_{ij} = (\Px)_{ij}^2.
\end{align*}
\end{definition}
To show the estimates in Lemma~\ref{lem:thirdorderself} for the $p$-Lewis-weights barrier $\phi_p$, we need to calculate the derivatives of the Lewis weights. The following Lemma presents the form of the Jacobian of the lewis weights as a function of $x$, by taking its directional derivative in direction $v$.
\begin{lemma}[Derivative of the Lewis weights]
    For arbitrary direction $v \in \Rn$, the directional derivative $\D\W(v)$ can be calculated as 
    \begin{align*}
        \D\W(v) = -2\diag{\Lambdax \G^{-1}\W s_{x,v}},
    \end{align*}
    where we define
    \begin{align*}
        &\Lambdax \triangleq \W - \Ptwo,\\
        &\G \triangleq \W - (1-2/p)\Lambdax.\numberthis\label{eq:GLambdadefs}
    \end{align*}
    Due to the importance and repetition of the vector $\Lambda \G^{-1}\W s_v$ in our calculations later on, we give it a separate notation
    \begin{align*}
        &r_{x,v} \triangleq \G^{-1}\W s_{x,v},\\
        &\Rxv \triangleq  \diag{r_{x,v}}.\numberthis\label{eq:Rdefinition}
    \end{align*}
    Then, the derivative of $\W$ can be written as
    \begin{align*}
        \D\W(v) = -2\diag{\Lambdax r_{x,v}}.
    \end{align*}
\end{lemma}
In the above Lemma, note that $\Lambdax$, $\G$, $r_{x,v}$, and $\Rxv$ are all functions of the location variable $x$, but we drop $x$ for clarity in our calculations.
    Furthermore, when $v$ is clear from the context, we denote $\D\W(v)$ in short by $\Wxv$.
Next, we calculate the derivative of the projection matrix $\Px$ onto the column space of $\W^{1/2-1/p}\A$ which is appropriately reweighted by the Lewis weights, as defined in Definition~\ref{def:projecrtion}.
\begin{lemma}[Derivative of the projection matrix]\label{lem:projderivative}
    The derivative of the projection matrix $\Px = \mathbf{P}(\W^{1/2-1/p}\A)$ in direction $v$ is given by
    \begin{align*}
        \D \Px(v) = -\Px \Rxv - \Rxv\Px + 2\Px\Rxv\Px,
    \end{align*}
    where $\Rxv$ is defined in Equation~\eqref{eq:Rdefinition}. When $v$ is clear from the context, we refer to $\Px \Rxv \Px$ by $\Pxv$ for brevity. Moreover, controlling the spectral norm of the diagonal matrix $\Rxv = \diag{\G^{-1}\W s_{x,v}}$ by the infinity norm of $s_{x,v}$ is one of the key ideas that allows us to break the mixing time. 
\end{lemma}

To reduce notation, in the proof we also make the dependence of $A_x$ to $x$ implicit and drop the index $x$.


We denote the target probability distribution inside the polytope by $\pi(.)$. We use $g$ for the Hessian of our hybrid barrier $\phi$. We refer to the Hessian of the Lewis-p-weight before rescaling by $g_1$, and the Hessian of $n/m$ scaled log barrier by $g_2$, i.e.
\begin{align*}
    &g_1 \triangleq \nabla^2 \logdet{\A^\top \W^{1-2/p}\A},\\
    &g_2 \triangleq \frac{n}{m} \A^\top \A.\\
    &g \triangleq \alpha_0 (g_1 + g_2).
\end{align*}
Throughout the proof, we use the notation $\lesssim$ to show an inequality with ignoring the logarithmic factors. We use $D$ for Euclidean derivative and $\nabla$ and $D_t$ for covariant differentiation with respect to the metric structure on the manifold. Moreover, we use $\preccurlyeq$ to show L\"{o}wner inequalities up to universal constants.

\subsection{Markov chains}\label{sec:markovchains}

For a Markov chain with state space $\mathcal{M}$, stationary distribution $Q$ and next step distribution $p_{u}(\cdot)$ for any $u\in\mathcal{M}$, the conductance of the Markov chain is defined as
\[
\Phi\triangleq\inf_{S\subset\mathcal{M}}\frac{\int_{S}p_{u}(\mathcal{M}\setminus S)dQ(u)}{\min\left\{ Q(S),Q(\mathcal{M}\setminus S)\right\} }.
\]
The conductance of an ergodic Markov chain allows us to bound its mixing time, i.e., the rate of convergence to its stationary distribution, e.g., via the following theorem of Lovász and Simonovits.
\begin{theorem}\label{thm:LS_mixing}
Let $Q_{t}$ be the distribution of the current point after $t$ steps of a Markov chain with stationary distribution $q$ and conductance at least $\Phi$, starting from initial distribution $Q_{0}$. For any $\varepsilon>0$, 
\[
d_{TV}(Q_{t},Q)\le\varepsilon+\sqrt{\frac{1}{\varepsilon}E_{x\sim Q_{0}}\frac{dQ_{0}(x)}{dQ(x)}}\left(1-\frac{\Phi^{2}}{2}\right)^{t}.
\]
\end{theorem}

To bound the conductance, we will reduce it to geometric isoperimetry.

\begin{definition}
The isoperimetry of a metric space $\mathcal{M}$ with target distribution $\pi$ is
\[
\psi=\inf_{\delta>0}\min_{S\subseteq\mathcal{M}}\frac{\int_{\{x| \ d(S,x)\le\delta\}}\pi(x)dx-\pi(S)}{\delta\min\left\{ \pi(S),\pi(\mathcal{M}\setminus S)\right\} }
\]
where $d$ is the shortest path distance in $\mathcal{M}$. 
\end{definition}

For a proof of the following theorem, see e.g.,~\cite{vempala2005geometric}.

\begin{lemma}\label{thm:mixing_proof}
Given a metric space $\mathcal{M}$ and a time-reversible Markov chain $p$ on $\mathcal{M}$ with stationary distribution $Q$, fix any $r>0$ and suppose that for any $x,y\in\mathcal{M}$ with $d(x,z)<r$, we have $d_{TV}(p_{x},p_{y})\leq 0.9$. Then, the conductance of the Markov chain is $\Omega(r\psi)$. 
\end{lemma}

We will need a more refined notion of $s$-conductance, to be able to ignore small subsets when proving isoperimetry.
\begin{definition}
[$s$-conductance] Consider a Markov chain with a state space $\mathcal M$,
a transition distribution $\mathcal{T}_{x}$ and stationary distribution
$\pi$. For any $s\in[0,1/2)$, the $s$-\emph{conductance} of the
Markov chain is defined by 
\[
\Phi_{s} \triangleq \inf_{\pi(S)\in(s,1-s)}\frac{\int_{S}{\mathcal T}_{x}(S^{c})\pi(x)dx}{\min(\pi(S)-s,\pi(S^{c})-s)}.
\]
\end{definition}

A lower bound on the $s$-conductance
of a Markov chain leads to an upper bound on its mixing rate.
\begin{lemma}
\cite{lovasz1993random} \label{lem:tvDecrease} Let $\pi_{t}$
be the distribution of the points obtained after $t$ steps of a lazy
reversible Markov chain with the stationary distribution $\pi$.  For $0<s\leq 1/2$ and $H_{s}=\sup\{|\pi_{0}(A)-\pi(A)|:A\subset\mathcal M,\,\pi(A)\leq s\}$,
it follows that 
\[
d_{TV}(\pi_{t},\pi)\leq H_{s}+\frac{H_{s}}{s}\left(1-\frac{\Phi_{s}^{2}}{2}\right)^{t}.
\]
\end{lemma}

The following theorem (see ~\cite{kook2022condition}) illustrates how one-step coupling with the
isoperimetry leads to a lower bound on the $s$-conductance. Its proof is similar to that of Lemma 13 in \cite{lee2018convergence} and can be found in full detail in Appendix~\ref{app:sconductance}.
\begin{theorem}
\label{prop:conductance} For a Riemannian manifold $(\mathcal M,g)$,
let $\pi$ be the stationary distribution of a reversible Markov chain
on $\mathcal M$ with a transition distribution $P_{x}$. Let ${\mathcal M}'\subset{\mathcal M}$
be a subset with $\pi({\mathcal M}')\geq 1-\rho$ for some $\rho<\frac{1}{2}$.
We assume the following one-step coupling: if $d_{g}(x,x')\leq\Delta\leq1$
for $x\in{\mathcal M}'$, then $d_{TV}(\mathcal T_{x},\mathcal T_{x'})\leq0.9$. Then
for any $\rho/(\Delta \psi_{\mathcal M})\leq s<\frac{1}{2}$ and given $\psi_{\mathcal M} \Delta \leq 1/2$, the $s$-conductance is bounded
below by 
\[
\Phi_{s}=\Omega(\psi_{\mathcal M}\Delta).
\]
\end{theorem}

\section{Hybrid barrier metric and second-order self-concordance}\label{sec:secondorderselfconcordance}

The goal of this section is to prove the strong self-concordance properties of our hybrid barrier as defined in Lemma~\ref{lem:calabitype}. We start by developing some basic properties of Lewis weights, the corresponding metric, and their derivatives, which we exploit throughout the proof. For sake of clarity of the calculations, we denote the matrix $\Px \Rxv \Px$ regarding vector $v$, which will appear a number of times by $\Pxv$.
Here we show the infinity norm self-concordance for the first and second order derivative of the metric as a warm up. For the proof of our third order self-concordance, we refer the reader to section~\ref{sec:thirdorderselfconcordance}.
In this section, for sake of brevity and clarity of the proof, we do not track the constants (which depends on $\frac{1}{4/p-1}$) and all of our inequalities $\lesssim, \preccurlyeq$ are up to log factors.

The following Lemma is proved in appendix~\ref{app:metricform}.
\begin{lemma}[$p$-Lewis-weight metric]\label{lem:metricsecondform}
The p-Lewis weight metric $g_1 = \nabla^2 \log\det{\left( \A^T\W^{1-2/p}\A \right)}$ can be written in the following form
\begin{align}
    g_1(x) &= \A^\top (\W + 2\Lambdax)\A + 2(1-2/p)\A^\top \Lambdax {\G}^{-1} \Lambdax \A,\label{lem:initialform}
\end{align}
or alternatively
\begin{align*}
    g_1(x) &= \A^\top (\W + 2\Lambdax)\A + (p^2/2) (1-2/p)\A^\top \G \A \\
    &-p^2 (1-2/p) \A^\top \Ptwo \A + (p^2/2) (1-2/p) \A^\top 
\Ptwo\G^{-1}\Ptwo\A.\numberthis\label{eq:metricotherform}
\end{align*}
\end{lemma}

In the following Lemma we state a vital $\|.\|_{\infty \rightarrow \infty}$ norm bound for the matrix $\G^{-1}\W$ which enables us to obtain L\"{o}wner inequalities by pulling off the $\|.\|_{x,\infty}$ norm of $v$, the direction of the derivative. Note that condition $p < 4$ is vital for this norm bound.

\begin{lemma}[Operator infinity norm bound]
For $\gv = \G^{-1}\W \sv$, given any vector $\sv$ and $p < 4$, we have
\begin{align*}
    \|\gv\|_\infty \leq \frac{1}{4/p - 1}\|\sv\|_\infty.
\end{align*}
\end{lemma}
\begin{proof}
    The proof can be found in Appendix~\ref{sec:infnormcomparison}.
\end{proof}

Next, we state a lemma regarding the expansion of the directional derivative of the Lewis weights metric $g_1$.

\begin{lemma}[Derivative of the $p$-Lewis weights metric]\label{lem:gderivative}
Given arbitrary direction $v$, we have
\begin{align*}
    \D g_1(v) = & - 2\A^\top \Sxv \big(\W + 2\Lambdax + 2(1-2/p)\Lambdax \G^{-1}\Lambdax\big) \A \tab &&(\trr 1)\\
    & - 2\A^\top \big(\W + 2\Lambdax + 2(1-2/p)\Lambdax \G^{-1}\Lambdax\big) \Sxv \A\\
    & + 3\A^\top \Wxv \A - 4\A^T \Big(\Px \odot (-\Rxv \Px - \Px \Rxv + 2\Pxv) \Big) \A \tab &&(\trr 2)\\  
    & - 4(1-2/p)\A^\top \Big(\Px \odot (-\Rxv \Px - \Px\Rxv) \Big) \G^{-1} \Lambdax \A \tab &&(\trr 3)\\
    & - 8(1-2/p)\A^\top \big(\Px \odot \Pxv \big) \G^{-1} \Lambdax \A \tab &&(\trr 4)\\
    & - 4(1-2/p)\A^\top \Lambdax \G^{-1} \Big(\Px \odot (-\Rxv \Px - \Px \Rxv + 2\Pxv) \Big) \A\\
    & + 2(1-2/p) \A^\top \Wxv\G^{-1}\Lambdax \A \tab &&(\trr 5)\\
    & + 2(1-2/p) \A^\top \Lambdax \G^{-1}\Wxv \A\\
    &-2(1 - 2/p)(2/p)\A^\top \Lambdax \G^{-1} \Wxv\G^{-1}\Lambdax \A. \tab  &&(\trr 5)'\\
    &-2(1 - 2/p)\A^\top \Lambdax \G^{-1} 2(1-2/p)\big(\Px \odot (-\Rxv \Px - \Px\Rxv) \big)\G^{-1}\Lambdax \A. \tab &&(\trr 6)\\
    &-2(1 - 2/p)\A^\top \Lambdax \G^{-1} 4(1-2/p)\big(\Px \odot \Pxv \big)\G^{-1}\Lambdax \A. \tab &&(\trr 7)\numberthis\label{eq:clarifyterms}
\end{align*}

We have numbered the terms above by $(\trr \dots)$ to refer to them later on. 
\end{lemma}

In order to show the first, second, and third self-concordance of our metric, we need to control the terms above as well as their first and second derivatives. We give the proof for the first and second order self-concordance in this section and delay the proof of third order self-concordance to appendix~\ref{sec:thirdorderselfconcordance}. Here, we start with a lemma which illustrates the calculation of the derivative of the $(\star 4)$ term above.
Ultimately we derive spectral bounds for each of the terms in these derivatives. We do not care about constants and factors of $p$ in these calculations (note that with the choice $p = 4 - 1/\log(m)$ these factors are at most polylogarithmic). Therefore, to simplify our calculation a bit, we ignore these constants.

\begin{lemma}\label{lem:star4}
The derivative of the term $(\trr 4)$ in Equation~\eqref{eq:clarifyterms} in direction $z$ is given by (up to constants)
\begin{align*}
    \D\big(\trr 4\big)(z) &\rightarrow \A^\top \Sxz \Px\odot \Pxv \G^{-1} \Lambdax \A &&(\trrr 1)\\
    &+ \A^\top \big(-\Rxz\Px - \Px\Rxz\big) \odot \Pxv \G^{-1}\Lambdax \A  &&(\trrr 2)\\
    &+ \A^\top \Pxz \odot \Pxv \G^{-1}\Lambdax \A &&(\trrr 3)\\
    & + \A^\top \Px\odot \Pxv \G^{-1} \D\G(z)\G^{-1}\Lambdax \A && (\trrr 4)\\
    & + \A^\top \Px \odot \Pxv \G^{-1} \D\Lambdax(z) \A &&(\trrr 5)\\
    & + \A^\top \Px\odot \Pxv \G^{-1} \Lambdax \Sxz \A &&(\trrr 6)\\
    & + \A^\top \Px\odot \Big(\big(-\Rxz \Px - \Px\Rxz + \Px\Rxz\Px\big)\Rxv \Px\Big) \G^{-1}\Lambdax \A &&(\trrr 7)\\
    & + \A^\top \Px \odot (\Px \D(\Rxv)(z) \Px) \G^{-1}\Lambdax \A &&(\trrr 8),
\end{align*}
where in the last term $\D(\Rxv)(z)$ we are considering $v$ as a fixed vector (i.e. the derivative in direction $z$ does not hit $v$).
\end{lemma}
\begin{proof}
Follows from ordinary differentiation and applying Lemma~\ref{lem:projderivative}.
\end{proof}

In order to get a handle on these matrices via L\"{o}wner ordering, we derive various stability Lemmas for the derivatives of the Lewis weights and their related matrices $\Lambdax$, $\G$, etc and the stability of their derivatives. For example, we show the following third order self-concordance type property for Lewis weights themselves. The following Lemma is proved in Appendix~\ref{sec:lownerineq} in Lemma~\ref{lem:wprimesecondderivative}.
\begin{lemma}[Third derivative bound for Lewis weights]
We have
\begin{align*}
    -\|s_z\|_\infty \|s_u\|_\infty \|s_v\|_\infty \W \preccurlyeq \D^2(\Wxv)(z,u) \preccurlyeq \|s_z\|_\infty \|s_u\|_\infty \|s_v\|_\infty \W,
\end{align*}
where recall $\Wxv = \D\W(v)$.
\end{lemma}
Recall that the symbol $\preccurlyeq$ means L\"{o}wner order up to a constant factor. For more details and the proofs, we refer the reader to Appendix~\ref{sec:derivativestability}. Next, we proceed to show our first- and second-order strong self-concordance for the Lewis weight barrier. Note that strong self-concordance is easily checked for the log barrier, so the major remaining challenge is to prove it for the Lewis weights barrier. The general theme of the proof is that we pull out the infinity norm of the directional derivative vectors $v,w,u$ from the tensors that are generated as a result of differentiation. This requires us to develop estimates on various fundamental matrix quantities that we defined in section~\ref{sec:prelim}, namely $\G, \Lambdax, \Rxv$ at any point $x$ inside the polytope. Importantly, we develop these estimates with respect to the $\|.\|_{x,\infty}$ norm instead of the usual metric norm $\|.\|_g$, which crucially requires $p < 4$. This constraint on $p$ has its root in controlling the $\|.\|_{\infty \rightarrow \infty}$ norm of the matrix $\G^{-1}\W$ in Lemma~\ref{lem:ginfnorm}.

\begin{lemma}[First order infinity norm self-concordance]\label{lem:gfirstderivativebound}
For a direction $v$ we have
\begin{align*}
   -\|s_{x,v}\|_\infty g_1 \preccurlyeq \D g_1(v) \preccurlyeq \|s_{x,v}\|_\infty g_1.
\end{align*}
\end{lemma}
\begin{proof}
Direct consequence of Lemmas~\ref{lem:Glambdabound} and~\ref{lem:lem:wbound}.
\end{proof}

In the rest of this section, we bring the proof of the second order strong self-concordance of our metric.
\begin{lemma}[Second order infinity norm self-concordance]\label{lem:gsecondderivativebound}
The second derivatives of the metric $g_1$ of our hybrid barrier is bounded as
\begin{align*}
    -\|s_z\|_\infty \|s_v\|_\infty g_1 \preccurlyeq \D g_1(z,v) \preccurlyeq \|s_z\|_\infty \|s_v\|_\infty g_1.
\end{align*}
\end{lemma}
\begin{proof}
The goal is to look at the quadratic form of $\D g(z,v)$ on arbitrary vector $\qq$, i.e. $\qq^\top \D g(z,v)\qq$ and control it with $\|s_z\|_\infty \|s_v\|_\infty \|\qq\|_g^2$.  First, we consider each of the subterms as a result of differentiating $(\trr 4)$ in Lemma~\ref{lem:gderivative}, in direction $z$. This derivative is expanded in Lemma~\ref{lem:star4}. Regarding the term $(\trrr 1)$ of this expansion in Lemma~\ref{lem:star4}, we have
\begin{align*}
    (\trrr 1):\ 
    & \Big|\qq^\top \A^\top \Sxz \Px\odot \Pxv \G^{-1}\Lambdax \A\qq \Big|\\
    & \leq 
    \brackets{\qq^\top \A^\top \Sxz \W \Sxz \A \qq}^{1/2}
    \brackets{\qq^\top \A^\top \Lambda \G^{-1}\Px\odot \Pxv \W^{-1} P\odot \tilde \Px \G^{-1}\Lambda A\qq}^{1/2}\\
    & \leq\|s_z\|_\infty \|s_v\|_\infty \bracketss{\qq^\top \A^\top \W \A\qq} \bracketss{\qq^\top \A^\top \Lambdax \G^{-1}\W\G^{-1}\Lambdax \A \qq}\\
    & \lesssim \|s_z\|_\infty \|s_v\|_\infty \|\qq\|_{g_1}^2,\numberthis\label{eq:trrr1} 
\end{align*}

Next, for the $(\trrr 2)$ term in Lemma~\ref{lem:star4}:
\begin{align*}
  (\trrr 2):\ 
  &\qq^\top \A^\top (-\Rxz\Px - \Px\Rxz)\odot \Pxv \G^{-1}\Lambdax \A \qq \\
  &\leq 
  \Big|\qq^\top \A^\top \Rxz \Px\odot \Pxv \G^{-1}\Lambdax \A\qq\Big| + 
  \Big|\qq^\top \A^\top \Px \odot \Pxv \Rxz \G^{-1}\Lambdax \A \qq\Big|.
\end{align*}
The first part is similar to the handle of term $(\trrr 1)$ in Equation~\eqref{eq:trrr1}. 
For the second part:
\begin{align*}
    &\Big|\qq^\top \A^\top \Px \odot \Pxv \Rxz \G^{-1}\Lambdax \A \qq\Big| \\
    & \leq \bracketss{\qq^\top \A^\top \Px\odot \Pxv \W^{-1} \Px\odot \Pxv \A \qq}
    \bracketss{\qq^\top \A^\top \Lambdax \G^{-1}\Rxz \W \Rxz \G^{-1}\Lambdax \A \qq}\\
    & \lesssim  \|s_v\|_\infty \|s_z\|_\infty \|\qq\|_{g_1}^2.
\end{align*}

For the $(\trrr 3)$ term in Lemma~\ref{lem:star4}:
\begin{align*}
 &\qq^\top \A^\top \Pxz \odot \Pxv \G^{-1}\Lambdax \A\qq \\
 &\lesssim
    \bracketss{\qq^\top \A^\top \Pxz \odot \Pxv \W^{-1} \Pxz \odot \Pxv \A\qq} \bracketss{\qq^\top A^\top\Lambdax \G^{-1}\W\G^{-1}\Lambdax \A \qq}
    \\
    &\lesssim
    \bracketss{\qq^\top \A^\top \W^{1/2}(\W^{-1/2} \Pxz \odot \Pxv \W^{-1/2})^2\W^{1/2} \A \qq}\bracketss{\qq^\top \A^\top \W \A \qq}
    \\
    &\lesssim
    \|s_z\|_\infty \|s_v\|_\infty \|\qq\|_{g_1}^2.\numberthis\label{star4ast3}
\end{align*}
where we used Lemma~\ref{lem:hadamardlowner} and~\ref{lem:levelzerobounds}.
Next, for term $(\trrr 4)$:
\begin{align*}
&\qq^\top \A^\top \Px\odot \Pxv \G^{-1} \D\G(z)\G^{-1}\Lambdax \A \qq\\
&\leq
\bracketss{\qq^\top \A^\top \Px\odot \Pxv \W^{-1} \Px\odot \Pxv \A\qq}\brackets{\qq^\top \A^\top \Lambdax \G^{-1}\D\G(z)\G^{-1}\W \G^{-1} \D\G(z)\G^{-1}\Lambdax \A
\qq}\\
&\lesssim 
 \brackets{\qq^\top \A^\top \W^{1/2}(\W^{-1/2}\Px\odot \Pxv \W^{-1/2})^2\W^{1/2}\A\qq} \brackets{\qq^T \A^\top (\W^{1/2}\G^{-1}\D\G(z)\G^{-1}\W^{1/2})^2  \A\qq}\\
&\lesssim \|s_v\|_\infty\|s_z\|_\infty\|\qq\|_{g_1}^2.
\end{align*}

Terms $(\trrr 5)$ and $(\trrr 6)$ are similar. For term $(\trrr 7)$, for the first term $\A^\top \Px\odot (\Rxz \Px\Rxv)\Px \G^{-1}\Lambdax \A$, note that
\begin{align*}
     \A^T \Px\odot (\Rxz \Px\Rxv\Px)\G^{-1}\Lambdax \A 
    = \A^\top \Rxz \Px\odot (\Px\Rxv\Px) \G^{-1}\Lambdax \A.
\end{align*}
which can be dealt with similar to $(\trrr 1)$ term using Lemma~\ref{lem:ginfnorm}.
The second term $\A^\top \Px\odot (\Px\Rxz\Rxv\Px)\G^{-1}\Lambdax \A$ in $(\trrr 7)$ is also similar to $(\trrr 1)$. For the last term in $(\trr 7)$, note that
\begin{align*}
    \qq^\top \Px \Rxz\Px\Rxv \Px \qq & \leq
    \brackets{\qq^\top \Px \Rxz \Px \Rxz \Px\qq}\brackets{\qq^\top \Px \Rxv \Px \Rxv \P \qq}\\
    & \leq \brackets{\qq^\top \Px \Rxz \Rxz \Px\qq}\brackets{\qq^\top \Px \Rxv  \Rxv \Px \qq}\\
    & \leq \|r_z\|_\infty\|r_v\|_\infty \qq^T \Px \qq,
\end{align*}
which implies
\begin{align*}
    \Px\Rxz \Px \Rxv \Px + \Px\Rxv \Px \Rxz \Px \preccurlyeq \|r_z\|_\infty \|r_v\|_\infty \Px.
\end{align*}
As a result, 
\begin{align*}
    \qq^\top &\A^\top \Px \odot (\Px\Rxz\Px\Rxv \Px + \Px\Rxv\Px\Rxz \Px) \G^{-1}\Lambdax \A\qq
    \\
    \leq &\bracketss{\qq^\top \A^\top (\Px \odot (\Px\Rxz\Px\Rxv \Px + \Px\Rxv\Px\Rxz \Px)) \W^{-1} \\
    &(\Px \odot (\Px\Rxz\Px\Rxv \Px + \Px\Rxv\Px\Rxz \Px)) \A\qq}\\
    &\times \bracketss{\qq^\top \A^\top \Lambdax \G^{-1}\W\G^{-1}\Lambdax \A\qq} \\
    \leq &\|s_v\|_\infty\|s_z\|_\infty \bracketss{\qq^\top \A^\top \Ptwo\W^{-1}\Ptwo\A\qq}\bracketss{\qq^\top \A^\top \Lambdax \G^{-1}\W \G^{-1}\Lambdax \A\qq}\\
    \lesssim &\|s_v\|_\infty\|s_z\|_\infty \|\qq\|_{g_1}^2.
\end{align*}

The bound for term $(\trrr 8)$ in Lemma~\ref{lem:star4} follows similarly, using Lemma~\ref{lem:derivativeofr}:
\begin{align*}
    &\qq^\top \A^\top \Px \odot (\Px \D(\Rxv)(z) \Px) \G^{-1}\Lambdax \A\qq \\
    & \leq \bracketss{\qq^\top \A^\top \Px\odot \Px\D(\Rxz)\Px \W^{-1} \Px\odot \Px\D(\Rxv)(z)\Px \A\qq}\bracketss{\qq^\top \A^\top \Lambdax \G^{-1}\W\G^{-1}\Lambdax \A\qq} \\
    &\lesssim \|\D(\Rxv)(z)\|_\infty\|\qq\|_{g_1}^2 \lesssim \|s_v\|_\infty \|s_z\|_\infty \|\qq\|_{g_1}^2.
\end{align*}

Next, we move on to bound the directional derivative of term $(\trr 5)$ in Lemma~\ref{lem:gderivative}, in direction $z$. This derivative is calculated in Lemma~\ref{lem:star5terms} in the Appendix.
For subterm $(\trrr 1)$ of $(\trr 5)$ defined in Lemma~\ref{lem:star5terms}, using Lemma~\ref{lem:lem:wbound}:
\begin{align*}
    \qq^\top \A^\top \Sxz \Wxv \G^{-1}\Lambdax \A\qq
    &\leq \bracketss{\qq^\top \A^\top \Sxz \W \Sxz \A\qq}\bracketss{\qq^\top \A^\top \Lambdax \G^{-1}\Wxv \W^{-1}\Wxv\G^{-1}\Lambdax \A\qq}\\
    &\lesssim \|s_z\|_\infty \|\qq\|_{g_1}^2.
\end{align*}
For subterm $(\trrr 2)$ of $(\trr 5)$ defined in Lemma~\ref{lem:lem:wbound}, we have using Lemmas~\ref{lem:derivativeofwprime} and~\ref{lem:levelzerobounds}:
\begin{align*}
        \qq^\top \A^\top \D(\Wxv)(z)\G^{-1}\Lambdax \A \qq
    & \leq \bracketss{\qq^\top \A^\top \D(\Wxv)(z)\W^{-1}\D(\Wxv)(z)\A\qq}\bracketss{\qq^\top \A^\top \Lambdax \G^{-1}\W\G^{-1}\Lambdax \A\qq}\\
    & \lesssim \|\qq\|_{g_1}^2 \|s_z\|_\infty\|s_v\|_\infty.
\end{align*}
For subterm $(\trrr 3)$ of $(\trr 5)$, using Lemmas~\ref{lem:levelzerobounds},~\ref{lem:Glambdabound}, and~\ref{lem:lem:wbound}:
\begin{align*}
    \qq^\top \A^\top \Wxv\G^{-1}\D\G(z)\G^{-1}\Lambdax \A\qq 
    \leq &\bracketss{\qq^\top \A^\top \Wxv\G^{-1}\Wxv \A\qq}\\
    &\bracketss{\qq^\top \A^\top \Sxz \Lambdax \G^{-1}\D\G(z)\G^{-1}\D\G(z)\G^{-1}\Lambdax \A\qq}\\
    \lesssim &\|\qq\|_{g_1}^2 \|s_v\|_\infty \|s_z\|_\infty. 
\end{align*}
Subterm $(\trrr 4)$ of $(\trr 5)$ is similar to $(\trrr 3)$ and subterm $(\trrr 5)$ is similar to subterm $(\trrr 1)$.

Now considering the second formulation of the metric presented in Lemma~\ref{lem:metricsecondform}, in Equation~\eqref{eq:metricotherform}, above we handled the case where one of the directional derivatives, with respect to either $v$ or $z$, hits the $P^{(2)}$ part in the last term of the metric in Equation~\eqref{eq:metricotherform}. Hence, regarding this last term, the remaining terms in its derivative are the ones for which the derivative with respect to both of $v$ and $z$ hit either the $\A$ matrix or the $\G^{-1}$ matrix, i.e.
\begin{align*}
    &\text{remaining terms of } \qq^\top\D^2(\A^\top \Ptwo \G^{-1}\Ptwo\A)(v,z)\qq \rightarrow\\
    & \qq^T \A^\top \Ptwo\G^{-1}\D\G(z,v)\G^{-1}\Ptwo\A\qq \\
    +& \qq^\top \A^\top \Ptwo\G^{-1}\D\G(z)\G^{-1}\D\G(v)\G^{-1}\Ptwo\A\qq + \qq^\top \A^\top \Ptwo\G^{-1}\D\G(v)\G^{-1}\D\G(z)\G^{-1}\Ptwo\A\qq \\
    +& \qq^\top\A^\top \Sxv\Sxz \Ptwo \G^{-1}\Ptwo \A\qq + \qq^\top\A^\top \Ptwo \G^{-1}\Ptwo \Sxv\Sxz \A\qq \\
    +& \qq^\top\A^\top \Ptwo \G^{-1}\Ptwo \Sxv\Sxz \A\qq \\
    +& \qq^\top \A^\top \Sxv \Ptwo \G^{-1}\Ptwo \Sxz\A\qq + \qq^\top \A^\top \Sxz \Ptwo \G^{-1}\Ptwo \Sxv\A\qq\\
    +& \qq^\top \A^\top \Sxv \Ptwo G^{-1}\D \G(z)\G^{-1} \Ptwo \A\qq + \qq^\top \A^\top  \Ptwo \G^{-1}\D \G(z)\G^{-1} \Ptwo \Sxv A\qq\\
    +& \qq^\top \A^\top \Sxz \Ptwo G^{-1}\D \G(v)\G^{-1} \Ptwo \A\qq + \qq^\top \A^\top  \Ptwo \G^{-1}\D \G(v)\G^{-1} \Ptwo \Sxz \A\qq.
    \numberthis\label{eq:Ginverseterms}
\end{align*}

All of the terms in~\eqref{eq:Ginverseterms} can be bounded by $O(\|s_\ell\|_w^2 \|s_z\|_\infty\|s_v\|_\infty)$.
For terms in the first line of Equation~\eqref{eq:Ginverseterms} we use Lemmas~\ref{lem:levelzerobounds} and~\ref{lem:Gsecondderivative}. For the second line we use Lemmas~\ref{lem:levelzerobounds} and~\ref{lem:Glambdabound}, and~\ref{lem:ginfnorm}. 
The $O(\|\qq\|_{g_1}^2 \|s_z\|_\infty\|s_v\|_\infty)$ bound on the rest of the terms in Equation~\eqref{eq:Ginverseterms} follows from Lemmas~\ref{lem:levelzerobounds} and~\ref{lem:ginfnorm} as well. 
Hence, overall we have shown for the last term in Equation~\eqref{eq:metricotherform}:
\begin{align*}
    -\|s_v\|_\infty \|s_z\|_\infty \A^\top \W \A \preccurlyeq \D(\A^\top \Ptwo\G^{-1}\Ptwo\A)(v,z) \preccurlyeq \|s_v\|_\infty \|s_z\|_\infty \A^\top \W \A.
\end{align*}
On the other hand, the derivative of the initial terms $\A^\top \W \A$, $\A^\top \Lambda \A$, $\A^\top \G \A$ in Equation~\eqref{eq:metricotherform} are similarly handled using Lemmas~\ref{lem:wprimesecondderivative},~\ref{lem:Gsecondderivative}, and~\ref{lem:levelzerobounds}, and~\ref{lem:ginfnorm}. This completes the proof of the second order strong self-concordance for $g_1$.

\end{proof}

Next, we move on to the third order self-concordance. For this, the number of terms grow quite large but luckily bounding them uses a similar approach. Hence, to give the essential ideas and derivations, we omit the proofs for the similar terms and only illustrate with the directional derivative of the $(\trr 4)$ term in Equation~\eqref{lem:gderivative}, which is the most complicated to handle. We state our final result for the directional derivatives of $(\trr 4)$ in Lemma~\ref{lem:twoderivativesstar4} below (for the proof, see Appendix~\ref{sec:thirdorderselfconcordance}).

\begin{lemma}[Second derivative of $(\trr 4)$]\label{lem:twoderivativesstar4}
Let $\mathbf{B}_{x,v}$ be the symmetrized version of the $(\trr 4)$ term in Lemma~\ref{lem:gderivative}:
\begin{align*}
    \mathbf{B}_{x,v} \triangleq \A^\top \Px\odot \Pxv \G^{-1}\Lambdax \A + \A^\top \Lambdax \G^{-1}\Px \odot \Pxv \A.
\end{align*}
where recall $\Pxv = \Px\Rxv\Px$. Then, two times derivative of $B_{x,v}$ in directions $z$ and $u$ can be spectrally controlled by the metric norm as the following:
\begin{align*}
     -\|v\|_{x,\infty}\|u\|_{x,\infty} \|z\|_{x,\infty} \A^\top \W \A 
     \preccurlyeq \D(\mathbf{B}_{x,v})(u,z) 
     \preccurlyeq \|u\|_{x,\infty} \|z\|_{x,\infty}\|v\|_{x,\infty} \A^\top \W \A,
\end{align*}
\end{lemma}

Finally, it is not hard to see that the log barrier also satisfies the infinity norm strong self-concordance.
For completeness, we state this in the following Lemma, proved in Appendix~\ref{sec:logbarrierself}. 

\begin{lemma}[Infinity self-concordance of the log barrier]\label{lem:logbarrierself}
    The metric $g_2 = \nabla^2 \phi_\ell$ regarding the log barrier $\phi_\ell(x) = -\sum_{i=1}^m \log(a_i^T x - b_i)$ in the polytope satisfies infinity norm third order strong self-concordance:
    \begin{gather*}
    -\|v\|_{x,\infty}g_2 \preccurlyeq \D g_2(v) \preccurlyeq \|v\|_{x,\infty} g_2,\\ 
    -\|v\|_{x,\infty}\|z\|_{x,\infty} g_2 \preccurlyeq \D g_2(v,z) \preccurlyeq \|v\|_{x,\infty}\|z\|_{x,\infty} g_2,\\
    -\|v\|_{x,\infty}\|z\|_{x,\infty}\|u\|_{x,\infty} g_2 \preccurlyeq \D^3g_2(v,z,u) \preccurlyeq \|v\|_{x,\infty}\|z\|_{x,\infty}\|u\|_{x,\infty} g_2.
\end{gather*}
\end{lemma}

Combining Lemma~\ref{lem:logbarrierself} with the infinity norm self-concordance of the $p$ Lewis weights metric proves the infinity self-concordance of the metric regarding our hybrid barrier.
\begin{proof}[Proof of Lemmas~\ref{lem:thirdorderself} and~\ref{lem:calabitype}]
    Proof of Lemma~\ref{lem:thirdorderself} is a direct consequence of Lemmas~\ref{lem:gfirstderivativebound},~\ref{lem:gsecondderivativebound}, and~\ref{lem:thirdorderselfconcordance}, and~\ref{lem:logbarrierself}. Proof ofo Lemma~\ref{lem:calabitype} follows from Lemma~\ref{lem:thirdorderself} and noting the fact that the $\|.\|_{x,\infty}$ norm can be upper bounded by the $\|.\|_g$ norm according to Lemma~\ref{lem:infwithgnorm}.
\end{proof}

\section{Bounding conductance and mixing time}\label{sec:boundingconductance}
The goal of this section is to illustrate how we combine different pieces together to prove Theorem~\ref{thm:mixing}. To this end, we prove a general purpose mixing time on a manifold in Theorem~\ref{thm:mixingtimetheorem}. The key to show Theorem~\ref{thm:mixingtimetheorem} is Lemma~\ref{lem:onestepcoupling} which we defer its proof to later. We start by defining an important concept of a "Nice set," which links the initial velocity $v_{x_0}$ to the $(R_1,R_2,R_3)$ normality.

\begin{definition}[Nice set]
    Given $x_0 \in \mathcal M$, we say a set $Q_{x_0} \subseteq T_{x_0}(\mathcal M)$ is $(R_1, R_2, R_3, \delta)$-nice if for $v_{x_0} \sim \mathcal N(0, g(x_0)^{-1})$, we have
    \begin{enumerate}
    \item $\mathbb P(v_{x_0} \notin Q_{x_0})\leq 0.001$.
    \item for every $x_1$ with $d(x_1, x_0) \leq \delta$, the Hamiltonian family of curves between $x_0$ and $x_1$ ending at  $Ham^\delta(x_0, v_{0})$ is $(R_1, R_2, R_3)$-normal. 
    \end{enumerate}
\end{definition}
\begin{theorem}\label{thm:mixingtimetheorem}
     Suppose we want to sample from some distribution $\pi$ on the manifold $\mathcal M$, starting from distribution $\pi_0$ with $M = \sup_{x\in \mathcal M} \frac{d\pi_0(x)}{d\pi(x)}$. Suppose there exists a set $S \subseteq \mathcal M$ with $\pi(S) \geq 1 - O(\epsilon/M)$, such that for every $x_0 \in S$ there exists an $(R_1, R_2, R_3, \delta)$-nice set $Q_{x_0} \subseteq T_{x_0}(\mathcal M)$. Moreover, let $\psi$ be the isoperimetric constant of the pair $(\mathcal M, g)$. Then, for any $\delta$ satisfying $\delta^2 R_1 \leq 1$, $\delta^2 R_3 \leq 1$, $\delta^3 R_2 \leq 1$, the mixing time to reach a distribution within TV distance $\epsilon$ of $\pi$ is bounded by
    \begin{align*}
        O(\log(M)(\psi\delta)^{-2}).
    \end{align*}
\end{theorem}
\begin{proof}
    Now with this choice of $\delta$, Lemma~\ref{lem:onestepcoupling}, which given a nice set for $x_0$ shows a bound on the closeness of the one step distributions, implies
for every $x_0 \in S$ and every $x_1$ within distance $d(x_0, x_1) \leq \delta$:
\begin{align*}
    TV(\mathcal{T}_{x_0}, \mathcal{T}_{x_1}) \leq 0.01.
\end{align*}
 Using Theorem~\ref{prop:conductance}, for $\rho = \mathbb P(S^c) = O(\epsilon/M)$ we get a lower bound on the $s$-conductance for $s = O(\epsilon/M)$:
\begin{align*}
    \Phi_s \geq \Omega((\psi \delta)^{-2}).
\end{align*}

Now using Lemma~\ref{lem:tvDecrease} with the same choice of $s$,
\begin{align*}
    d_{TV}(\pi_{t},\pi)\leq H_{s}+\frac{H_{s}}{s}\left(1-\frac{\Phi_{s}^{2}}{2}\right)^{t} \leq \epsilon,
\end{align*}
where we used the fact that $H_s \leq Ms = O(\epsilon)$ (recall the definition of $M$) and the fact that we pick $t$ of the order $\log(M)(\psi \delta)^2$ as $H_s/s \leq \epsilon$. The proof is complete.

\end{proof}

What remains to show is Lemma~\ref{lem:onestepcoupling} regarding the closeness of the one step distributions of the Markov chain. which is the main content of this section. This is vital in proving Theorem~\ref{thm:mixingtimetheorem} as it is one of the main building blocks, in addition ot the isoperimetry of the target measure, to bound the conductance of the chain.

To prove Lemma~\ref{lem:onestepcoupling}, we start with some definitions. The overall plan is that we approximate the density of a Hamiltonian step as written in Equation~\eqref{eq:thedensity} as in Equation~\eqref{eq:approx} and bound its change going from $x_0$ to $x_1$ for most of the vectors $v_{x_0}$ within a nice set in the tangent space of $x_0$.

\begin{definition}
Consider a family of Hamiltonian curves $\gamma_s(t)$ for time interval $s,t \in [0,\delta]$ all ending at $y$, where $\gamma(0) = x$, and $\gamma'(0) = v_x$.
Define the local push-forward density of $v_x \sim \mathcal N(0, g^{-1})$ onto $y$ by
\begin{align}
   P^{v_x}(y) = det(J^{v_x}_{y})\frac{\sqrt{ |g(y)|}}{\sqrt{(2\pi)^n}} e^{-\|v_x\|_g^2/2},\label{eq:thedensity}
\end{align}
where $J^{v_x}_{y}$ is the inverse Jacobian of the Hamiltonian after time $\delta$, sending $v_x$ to $y$, which we denoted by $Ham^\delta$. we consider the Jacobian as an operator between the tangent spaces. The push forward density at $y$ with respect to the manifold measure is given by
\begin{align*}
    P(y) = \sum_{v_x: \ Ham^{\delta}(x,v_x) = y} P^{v_x}(y).
\end{align*}
 Note that $dg(y)$ refers to the manifold measure. Define the approximate local push-forward density of $v_x$ as
\begin{align}
    \tilde P^{v_x}(y) = \exp{\big(-\int_{t=0}^\delta \frac{t(\delta - t)}{2} tr(\Phi(t)) dt \big)} \sqrt{ |g(y)|}/\sqrt{(2\pi)^n} e^{-\|v_x\|_g^2/2}.\label{eq:approx}
\end{align}
\end{definition}

\begin{lemma}[Lemma 22 in~\cite{lee2018convergence}]
For an $R_1$-normal Hamiltonian curve, for $0 \leq \delta^2 \leq \frac{1}{R_1}$ we have
\begin{align}
    |\log(\tilde P^{v_x}(y)) - \log(P^{v_x}(y))| \lesssim (\delta^2 R_1)^2.\label{eq:densityratio}
\end{align}
\end{lemma}

\begin{lemma}[Lemma 32 in~\cite{lee2018convergence}]\label{lem:derivativeoftangent}
In the setting of Lemma~\ref{lem:familyof}, for an $(R_1, R_3)$ normal $\gamma_0$, denoting $\frac{d}{ds}\gamma_s(0)$ by $z$, 
we have
\begin{align*}
    \delta\frac{d}{ds}\|\gamma'_s(0)\|^2 \leq |\langle v_x, z\rangle| + \delta^2R_3 \|z\|.
\end{align*}
\end{lemma}

\begin{lemma}[Change of the pushforward density]\label{lem:familyof}
Consider the family of smooth Hamiltonian curves $\gamma_s(t)$ up to time $\delta$ from $x_0$ to $x_1$ pointing towards $y$, namely $\gamma_0(0) = x_0$, $\gamma_0(\delta) = y$, and $\gamma_s'(0) = v_x$ regarding a point $x = \gamma_s(0)$ along the geodesic between $x_0$ to $x_1$ whose tangent to the geodesic is $z \triangleq \frac{d}{ds}\gamma_s(0)$. Then, given that $\gamma_s(t)$ is $(R_1, R_2, R_3)$ normal for $0\leq s,t \leq \delta$ and $\delta^2 \leq \frac{1}{R_1}$, we have
\begin{align*}
    \delta \frac{d}{ds} \log(\tilde P^{v_x}(y)) \leq |\langle v_x, z \rangle| + \delta^3 R_2 + \delta^2 R_3.
\end{align*}
\end{lemma}
\begin{proof}
Simply differentiating Equation~\eqref{eq:approx}:
\begin{align*}
    \delta \big|\frac{d}{ds}\log(\tilde P^{v_x}(y))| &= \big| -\delta \frac{d}{ds}\big(\int_{t=0}^\delta \frac{t(\delta - t)}{\delta} tr(\Phi(t)) dt \big) - \delta\frac{d}{ds}\|v_x\|_g^2/2 \big|\\
    &\leq \delta \int_{t=0}^\delta \frac{t(\delta - t)}{\delta} |\frac{d}{ds}tr(\Phi(t))\big| dt + |\langle v_x, z \rangle| + \delta^2 R_3 \|z\|.
\end{align*}
where we used Lemma~\ref{lem:derivativeoftangent}.
Furthermore, using Lemma~\ref{lem:R2bound} and noting our assumption $\|z\| = \|\frac{d}{ds}\gamma_s(0)\| = 1$:
\begin{align*}
LHS  &\lesssim |\langle v_x, z \rangle| + \delta^3 R_2 + \delta^2 R_3.
\end{align*}
\end{proof}

\begin{lemma}[Change in probability of events under approximate density]\label{lem:eventbound}
Let $Q_{x_0} \subseteq T_{x_0}(\mathcal M)$ be a $(R_1,R_2,R_3,\delta)$ nice set 
in the tangent space of $x_0$ and let $x$ be an arbitrary point in the geodesic between $x_0$ and $x_1$. For vector $v_x$ in the tangent space of $x$ with $Ham^\delta(x,v_x) = y$ we can consider the family of hamiltonian curves $\gamma_s(t)$ between $x_0 = \gamma_0(0)$ and $x_1 = \gamma_\delta(0)$ with $\gamma_s(\delta) = y$ for all $0\leq s\leq \delta$.
Now let $p_{n}$ be the finite measure obtained by restricting the normal distribution in the tangent space of $x$ to vectors $v_x$ for which the corresponding $v_{x_0} = \gamma'_0(0) \in Q_{x_0}$. For a point $y \in \mathcal M$, let $\tilde P_x^{n}(y)$ be the approximate pushforward density of $p_n$ onto $\mathcal M$, defined as
\begin{align}
    \tilde P^{n}(y) = \tilde P_x^{n}(y)= \Big(\sum_{v_x: \ Ham^{\delta}(x,v_x) = y, \ v_{x_0} \in Q_{x_0}} \tilde P^{v_x}(y)\Big) dg(y),\label{eq:tildepndef}
\end{align}
where $\tilde P_x^{v_x}(y)$ is defined in~\eqref{eq:approx}. We define $\tilde P^n(.)$ to be the corresponding finite measure. Now given a fixed event $Y \subset \mathcal M$ with probability $ \tilde P^n(Y) \geq n^{-10}$, we have
\begin{align}
    \delta \Big|\frac{d}{ds}\log(\tilde P^{n}(Y))\Big| \lesssim
    1 + \delta^3 R_2 + \delta^2 R_3,\label{eq:case1}
\end{align}
 and for all $Y$:
 \begin{align*}
      \delta \Big|\frac{d}{ds}\log(\tilde P^{n}(Y))\Big| \lesssim
    \sqrt n + \delta^3 R_2 + \delta^2 R_3.
 \end{align*}
 
 Note that $\tilde P_x^{n}$ depends on $x = \gamma_s(0)$, and we are fixing the set $Q_{x_0}$ in the tangent space of $x_0$. 
\end{lemma}
\begin{proof}
Let $\tilde P_1^n$ be the density of further restricting $\tilde P^n$ to $v_x$'s for which $\langle v_x , z\rangle \lesssim 1$ where recall $z \triangleq \frac{d}{ds}\gamma_s(0)$, and $\tilde P_2^n$ be such that $\tilde P^n(y) = \tilde P_1^n(y) + \tilde P_2^n(y)$. Note that
\begin{align*}
    \Big|\frac{\frac{d}{ds}\tilde P^{n}(Y)}{\tilde P^{n}(Y)}\Big|
    & =\Big(\frac{\frac{d}{ds}\tilde P_1^{n}(Y)}{\tilde P_1^{n}(Y)}\Big)\Big(\frac{\tilde P_1^{n}(Y)}{\tilde P^{n}(Y)}\Big)
     + \Big(\frac{\frac{d}{ds}\tilde P_2^{n}(Y)}{\tilde P_2^{n}(Y)}\Big)\Big(\frac{\tilde P_2^{n}(Y)}{\tilde P^{n}(Y)}\Big)\\
    &=\text{LHS}_1 + \text{LHS}_2.\numberthis\label{eq:root}
\end{align*}
But note that for the first term
\begin{align*}
    \text{LHS}_1 & \leq
     \int_{Y}\sum_{v_x: \ Ham^{\delta}(x,v_x) = y, \ v_{x_0} \in Q_{x_0}, \ \langle v_x, z\rangle \lesssim 1} \Big| \Big(\frac{\frac{d}{ds}\tilde P^{v_x}(y)}{\tilde P^{v_x}(y)}\Big)\Big(\frac{P^{v_x}(y)}{\tilde P_1^n(y)}\Big)\Big(\frac{\tilde P_1^{n}(y)}{\tilde P_1^{n}(Y)}\Big)\Big| dg(y)\\
    & \leq \int_{Y} \Big(\frac{\tilde P_1^{n}(y)}{\tilde P_1^{(n)}(Y)}\Big) \Big(\Big|\langle v_x, z \rangle\Big| + \delta^3 R_2 + \delta^2 R_3\Big)\delta^{-1}\\
    & \lesssim (1 + \delta^3 R_2 + \delta^2 R_3)/\delta.
\end{align*}
To see why the second line holds, note that the hamiltonian curve from $x$ to $y$ is $(R_1, R_2, R_3)$ normal from our assumption for time $t\in (0,\delta)$.
The second line follows from Lemma~\ref{lem:familyof}. The third line follows simply by the choice $\langle v_x, z \rangle \leq 1$.

Similarly for the second term
\begin{align*}
    \text{LHS}_2 \leq \Big(|\langle v_x, z \rangle| + \delta^3 R_2 + \delta^2 R_3\Big)/\delta \leq \Big(\sqrt n + \delta^3 R_2 + \delta^2 R_3\Big)/\delta,
\end{align*}
where we used $|\langle v_x, z \rangle| \leq \|v_x\|_g\|z\|_g$. Combining these and putting back in~\eqref{eq:root} implies
\begin{align*}
    \delta\Big|\frac{d}{ds}\log(\tilde P^{n}(Y))\Big| \leq 1 + \delta^3 R_2 + \delta^2 R_3 + \Big(\frac{\tilde P_2^{n}(Y)}{\tilde P^{n}(Y)}\Big) \big(\sqrt n + \delta^3 R_2 + \delta^2 R_3\big).
\end{align*}
To show case~\eqref{eq:case1}, using the fact that the densities regarding $\tilde P^n$ and $P^n$ are within constant of one another~\eqref{eq:densityratio}:
\begin{align*}
     n^{-10} \lesssim P^n(Y) \lesssim \tilde P^n(Y),
\end{align*}
which follows from assumption on $Y$ while 
\begin{align*}
    \tilde P_2^n(Y) \lesssim P_2^n(Y) \leq n^{-10},
\end{align*}
which follows becuae $1 \lesssim \langle v_x, z\rangle$ is a low probability event using gaussian tail bound. This completes the proof.
\end{proof}

Using the bounds on smoothness, we will show that one-step distributions of RHMC from two nearby points will have large overlap (and hence TV distance less than $1$).
\begin{lemma}[One-step coupling for RHMC]\label{lem:onestepcoupling}
Consider two points $x_0$ and $x_1$ and suppose $Q_{x_0}$ is a $(R_1, R_2, R_3, \delta)$-nice set in the tangent space of $x_0$.
Now given step size $\delta$ such that $\delta^2 \leq \frac{1}{R_1}, \ \delta^3R_2 \leq 1, \ \delta^2R_3 \leq 1$ and close by point $x_1$ such that $d(x_0, x_1) \leq \delta$, where $d$ is the distance on the manifold, the total variation distance between $P_{x_0}$ and $P_{x_1}$ is bounded by $0.01$.
\end{lemma}
\begin{proof}
Similar to~\eqref{eq:tildepndef}, we define
\begin{align*}
    P_x^{n}(y)= \Big(\sum_{v_x: \ Ham^{\delta}(x,v_x) = y, \ v_{x_0} \in Q_{x_0}} P^{v_x}(y)\Big) dg(y).
\end{align*}
First, note that for any event $Z \subseteq \mathcal M$, we have using Lemma~\ref{lem:highprobset}
\begin{align*}
|P^n_{x_0}(Z)- P_{x_0}(Z)| \leq \mathbb P(v_x \notin Q_{x_0})\leq 0.001.
\end{align*}
Suppose $Y \subseteq \mathcal M$ be a set for which 
\begin{align*}
    P_{x_0}(Y) - P_{x_1}(Y) > 0.01.
\end{align*}
This means $P_{x_0}(Y) \geq 0.01$, and in particular from~\eqref{eq:aval}
\begin{align}
    P^n_{x_0}(Y) - P^n_{x_1}(Y) \geq P_{x_0}(Y) - P_{x_1}(Y) - \mathbb P(v_{x_0} \notin Q_{x_0}) \geq 0.005,\label{eq:tv3}
\end{align}
which also implies
\begin{align*}
    P^n_{x_0}(Y) \geq 0.005.
\end{align*}
Now from
~\eqref{eq:densityratio} we have $\tilde P^n(Y) \geq 0.001$. But now using the assumptions on $R_2$ and $R_3$ and plugging it into Equation~\eqref{eq:case1} in Lemma~\ref{lem:eventbound} we can state
\begin{align*}
    \delta \Big|\frac{d}{ds}\log(\tilde P^n_x(Y))\Big| \lesssim 1,
\end{align*}
which implies at time $s = \delta$ we have
\begin{align*}
     \log(\tilde P^n_{\gamma_0(0)}(Y)) - \log(\tilde P^n_{\gamma_\delta(0)}(Y)) \lesssim 1,
\end{align*}
or in other words
\begin{align*}
    \tilde P^n_{x_1}(Y) / \tilde P^n_{x_0}(Y) \lesssim 1.
\end{align*}
Now again applying the constant boundedness of the ratio between $\tilde P^n$ and $P^n$, we obtain
\begin{align}
  P^n_{x_0}(Y) / P^n_{x_1}(Y) \lesssim 1.\label{eq:ratioboundedness} 
\end{align}
By picking small enough constants, Equation~\eqref{eq:ratioboundedness} implies
\begin{align*}
    P^n_{x_0}(Y) - P^n_{x_1}(Y) < 0.001.
\end{align*}
This further implies from~\eqref{eq:aval}:
\begin{align*}
    P_{x_0}(Y) - P_{x_1}(Y) \leq P^n_{x_0}(Y) - P^n_{x_1}(Y) + 0.001 \leq 0.002,
\end{align*}
which contradicts Equation~\eqref{eq:tv3}. This completes the proof.
\end{proof}

Finally, Combining Theorems~\ref{thm:mixingtimetheorem} and~\ref{thm:parameterbounds} and Lemma~\ref{lem:stabilityinfnorm}, we prove the main Theorem~\ref{thm:mixing}.
\\
\begin{proof}[Proof of Theorem~\ref{thm:mixing}]\label{eq:proofofmaintheorem}
    Given a fixed parameter $c > 1$, using Lemma~\ref{lem:highprobset}, there exists a high probability set $S = S_c \subseteq \mathcal M$,
\begin{align}
\pi(S) \geq 1 - poly(m)e^{-c^2/2},\label{eq:Sprob}
\end{align}
    such that every
    $x_0 \in S$ has a corresponding nice set $Q_{x_0} \in T_p(\mathcal M)$.
    

 (Recall $\pi$ is the distribution supported on the polytope with density $e^{-\alpha \phi}$.)

    Now for the same arbitrary $c > 1$ we considered above, we wish to satisfy the conditions in Theorem~\ref{thm:mixingtimetheorem} on $\delta$, namely $\delta^2 R_1(c) \leq 1$, $\delta^2 R_3(c) \leq 1$, $\delta^3R_2(c) \leq 1$ (We have used this notation to emphasize that $R_1, R_2, R_3$ are function of $c$). But according to Theorem~\ref{thm:parameterbounds}, these parameters can be set as:
    \begin{align*}
        &R_1 = (c^2 + \alpha \sqrt \alpha_0)\sqrt n,\\
        &R_2 = (c^2 + \alpha \sqrt{\alpha_0}) n,\\
        &R_3 = c^2(\sqrt n + n \delta) + n\delta c\alpha \sqrt{\alpha_0},
    \end{align*}
    plus Lemma~\ref{lem:stabilityinfnorm} imposes the following condition $\delta$:
    \begin{align*}
        \delta \leq \delta' = \frac{1}{\sqrt{c^2 + \alpha \sqrt \alpha_0}n^{1/4}}.
    \end{align*}
    Hence, the conditions on $\delta$ translates into
    \begin{align*}
        &\delta \leq \frac{1}{n^{1/4}c},\\
        &\delta \leq \frac{1}{n^{1/3}c^{2/3}},\\
        &\delta \leq \frac{1}{n^{1/3}},\\
        &\delta \leq \frac{1}{n^{1/3}c^{1/3}(\alpha \sqrt{\alpha_0})^{1/3}},\\
        &\delta \leq \frac{1}{\alpha^{1/2} {\alpha_0}^{1/4} n^{1/4}}
    \end{align*}
    Note that a sufficient condition on $\delta$ which satisfies all of the above constraints is (assuming $c \geq 1$) 
    \begin{align*}
        \delta = \frac{1}{c}\min\{\frac{1}{n^{1/3}}, \frac{1}{n^{1/3}(\alpha \sqrt{\alpha_0})^{1/3}}, \frac{1}{\alpha^{1/2} {\alpha_0}^{1/4}n^{1/4}}\}.\numberthis\label{eq:delta}
    \end{align*}
   
    Now to satisfy the condition $P(S) \geq 1 - O(\epsilon)$ in Theorem~\ref{thm:mixingtimetheorem}, noting Equation~\eqref{eq:Sprob}, we set 
    \begin{align*}
    c = \sqrt{\log\big(poly(n)M/(\epsilon \Delta \psi)\big)} = \Theta\big(\sqrt{\log(Mm/\epsilon)}\big).    
    \end{align*}

   On the other hand, from Theorem~\ref{thm:hybrid-iso}, we see that for the choice of $p = 4 - \lambda$ converging to $4$ from below ($\lambda$ is a small constant), the square of the isoperimetry constant is $\psi^2 =\Theta(\max\{m^{-\frac{2/p}{2/p + 1}} n^{-\frac{1}{2/p + 1}}, \alpha\})$.
   Now plugging this $\psi$ and $\delta$ from~\eqref{eq:delta} into Theorem~\ref{thm:mixingtimetheorem} and noting the choice of $c$ we get the following mixing bound:
   \begin{align*}
       \min\{\alpha^{-1}, n^{2/3}m^{1/3}\}\max\{n^{2/3}, n^{2/3}(\alpha\sqrt{\alpha_0})^{2/3}, n^{1/2}\alpha \sqrt{\alpha_0}\}\log(M)\log(Mm/\epsilon).
   \end{align*}
   But it is easy to check that picking $\lambda = \Theta(1/\log(n))$ only adds a $1/poly(\log(m))$ factor to $\delta$. Note that with this choice of $\lambda$, we have $m^{-\frac{2/p}{2/p + 1}} n^{-\frac{1}{2/p + 1}} = \Theta(n^{2/3}m^{1/3})$, hence the mixing time becomes
   \begin{align*}
       \min\{\alpha^{-1}, n^{2/3}m^{1/3}\}\max\{n^{2/3}, n^{2/3}(\alpha\sqrt{\alpha_0})^{2/3}, n^{1/2}\alpha \sqrt{\alpha_0}\}poly(\log(m))\log(M)\log(Mm/\epsilon).
   \end{align*}
     But note that if $\alpha\sqrt{\alpha_0} n^{1/2} \geq n^{2/3}$ or $n^{2/3}(\alpha\sqrt{\alpha_0})^{2/3} \geq n^{2/3}$, then $\alpha^{-1} \leq n^{2/3}m^{1/3}$. Hence, the mixing time boils down to
     \begin{align*}
         &\min\{\alpha^{-1}(n^{2/3}+n^{2/3}(\alpha \sqrt{\alpha_0})^{2/3} + \alpha\sqrt{\alpha_0} n^{1/2}), n^{4/3}m^{1/3}\}poly(\log(m))\log(M)\log(Mm/\epsilon),\\
         &=\min\{\alpha^{-1}n^{2/3}+\alpha^{-1/3}n^{5/9}m^{1/9} + m^{1/6}n^{1/3}, n^{4/3}m^{1/3}\}poly(\log(m))\log(M)\log(Mm/\epsilon).
     \end{align*}

\end{proof}

\section{On the Geometry and Stability of Hessian Manifolds}\label{sec:smoothness}
In this section, we prove the smoothness of the operator $\Phi(t)$, namely we show with that a nice Hamiltonian curve is $(R_1, R_2, R_3)$ normal. Our proof does not open up the definition of the mtric $g$ and its derivatives for our hybrid barrier, instead we exploit the strong-self concordance property in Lemma~\ref{lem:thirdorderself} to show the desired smoothness bounds, hence our framework potentially can be applied in other settings.  Interestingly, in order to bound the trace of certain operators that arise from bounding the smoothness of the Hamiltonian curves on manifold, it turns out that writing them as the average of random low rank tensors will enable us to apply our strong self-concordance estimates more efficiently and provide sufficient bounds to break the mixing time.
\subsection{Bounding $R_1$}

\begin{lemma}\label{lem:R1bound}
For the parameter $R_1$ regarding the Frobenius norm bound of $\Phi(t)$, given the control over the infinity norm of $\|s_v\|_\infty \lesssim c$, $\|v\|_g \lesssim c\sqrt n$ (note that the vector $v$ is inherent in the definition of $\Phi$), then we have
\begin{align*}
    R_1 \lesssim (c^2 + \alpha\sqrt{\alpha_0})\sqrt n..
\end{align*}
\end{lemma}
\begin{proof}
Directly follows from Lemmas~\ref{lem:Roperatorbound} and~\ref{lem:Moperatorbound}.
\end{proof}

First, recall the definition of the Frobenius norm:
\begin{align*}
    \|\Phi(t)\|_F^2 = \mathbb E_{v_1, v_2 \sim \mathcal N(0,g^{-1})} \mathbb E\langle v_1, \Phi(t)v_2\rangle^2.
\end{align*}
To bound $R_1$, i.e. the Frobenius norm of $\Phi(t)$, note that
\begin{align*}
    \Phi(t) = R(.,v)v + M(.),
\end{align*}
where $R$ is the Riemann tensor and $M$ is obtained from the bias vector $\mu$. 
In particular, we have
\begin{align*}
    R(\ell,v)v &= g^{-1}\D g(v)g^{-1}\D g(v)\ell\\
    &+ g^{-1}\D g(\ell)g^{-1}\D g(v)v,\numberthis\label{eq:riemann}\\
    M(\ell) & = \nabla_{\ell}(\nabla (\alpha\phi)) + \frac{1}{2}\nabla_{\ell}(g^{-1}\tr(g^{-1}\D g)).
\end{align*}

We start from the Riemann tensor.  The proof of this bound follows directly from the infinity norm second-order self-concordance of $g$. 
\begin{lemma}[Frobenius norm of random Riemann tensor]\label{lem:Rfrobeniusbound}
Assuming $\|s_v\|_\infty \lesssim c, \ \|v\|_g \lesssim c\sqrt n$, we have
\begin{align*}
    \|R(.,v)v\|_F \leq c^2\sqrt n.
\end{align*}
\end{lemma}
\begin{proof}
For the first term of $R(.,v)v$ as written in~\eqref{eq:riemann}:
\begin{align*}
    \mathbb E_{v_1, v_2 \sim \mathcal N(0, g^{-1})} (v_1^T \D g(v)g^{-1}\D g(v)v_2)^2
    &=\mathbb E v_1^T \D g(v)g^{-1}\D g(v)v_2v_2^T \D g(v)g^{-1}\D g(v)v_1\\
    &= \mathbb E v_1^T \D g(v)g^{-1}\D g(v)g^{-1} \D g(v)g^{-1}\D g(v)v_1\\
    & \lesssim {\|s_v\|^4_\infty} \mathbb E v_1^T g v_1\leq \|s_v\|_\infty^4 n\lesssim c^4 n.
\end{align*}
For the second term of the Riemann tensor:
\begin{align*}
   \mathbb E_{v_1, v_2 \sim \mathcal N(0, g^{-1})} (v_1^T \D g(v_2)g^{-1}\D g(v)v)^2
   &=\mathbb E v^T\D g(v)g^{-1}\D g(v_2)v_1v_1^T \D g(v_2)g^{-1}\D g(v)v\\
   &=\mathbb E v^T\D g(v)g^{-1}\D g(v_2)g^{-1} \D g(v_2)g^{-1}\D g(v)v\\
   & \leq \mathbb E \|s_v\|_\infty^2\|s_{v_2}\|_\infty^2 v^T g v\\
   & \lesssim \|s_v\|_\infty^2 v^T g v \lesssim c^4 n.
\end{align*}
\end{proof}
Lemma~\ref{lem:Rfrobeniusbound} states $\sqrt n$ as an upper bound on the Frobenius norm of $R(\ell, v)v$ given that the curve is nice. 

Next, we prove a lemma regarding the expansion of the operator $M$, applying the covariant derivative.
\begin{lemma}[Subterms for operator $M$]\label{lem:mterms}
We have the following expansion for the subterms of operator $M$:
\begin{align*}
    &\langle \nabla_{v_1}(\nabla (\alpha\phi)), v_2\rangle 
    = v_2^\top \D g(\nabla (\alpha\phi))v_1 + v_2^\top \D^2(\alpha\phi) v_1,\\
    & \langle \nabla_{v_1}(g^{-1}tr(g^{-1}Dg)), v_2\rangle = 
    v_2^\top \D g(\xi)v_1 + v_2^\top\D(g\xi)v_1,\numberthis\label{eq:secondpartt}
\end{align*}
where
\begin{align*}
    \xi \triangleq g^{-1}\tr(g^{-1}Dg).
\end{align*}
Moreover,
\begin{align}
    v_2^\top \D(g\xi)v_1 = -\tr(g^{-1}\D g(v_1)g^{-1}\D g(v_2))
    + \tr(g^{-1}\D^2g(v_1. v_2)).\label{eq:Mlastterm}
\end{align}
\end{lemma}
\begin{proof}
By differentiating the first term:
\begin{align*}
    \langle \nabla_{v_1}(\nabla (\alpha\phi)) , v_2\rangle = \langle -g^{-1}\D g(v_1)g^{-1}\D(\alpha\phi) + g^{-1}\D^2(\alpha\phi)[v_1] + g^{-1}\D g(\nabla (\alpha\phi))v_1, v_2\rangle.
\end{align*}
But noting that $\nabla (\alpha\phi) = g^{-1}\D(\alpha\phi)$, the first and third terms are the same and we get the result. For the second term:
\begin{align}
    \langle \nabla_{v_1}(g^{-1}\tr(g^{-1}\D g)), v_2\rangle = 
    v_2^\top \D g(\xi)v_1 + v_2^\top \D(g\xi)v_1.
\end{align}
Finally, for the second argument of the Lemma
\begin{align*}
   v_2^\top \D(g\xi)v_1 & = v_2^\top \D(g\xi)(v_1) \\
   &= -v_2^\top tr(g^{-1}\D g(v_1)g^{-1}\D g) + v_2^\top \tr(g^{-1}\D^2g(v_1,.))\\
    & = -\tr(g^{-1}\D g(v_1)g^{-1}\D g(v_2))
    + tr(g^{-1}\D^2g(v_1, v_2)).
\end{align*}
\end{proof}

Next, we bound the Frobenius norm of the $M$ part in the following lemma, again only using infinity norm second-order self-concordance of $g$ to bound each of the four terms.
\begin{lemma}[Frobenius norm of operator $M$]\label{lem:Mfrobeniusnorm}
We have
\begin{align*}
    \|M(x)\|_F \lesssim \alpha \sqrt{\alpha_0 n}.
\end{align*}
\end{lemma}
\begin{proof}
To bound the Frobenius norm of the first part of the first term of operator $M$ stated in Lemma~\ref{lem:mterms}:
\begin{align*}
    \mathbb E_{v_1, v_2 \sim \mathcal N(0, g^{-1})}(v_1^\top \D g(\nabla \phi) v_2)^2
    &= \mathbb E v_1^\top \D g(\nabla \phi)v_2v_2^\top \D g(\nabla \phi) v_1 \\
    &= \mathbb E v_1^\top \D g(\nabla \phi)g^{-1}\D g(\nabla \phi)v_1\\
    &= \mathbb E {\nabla \phi}^\top Dg(v_1)g^{-1}Dg(v_1)\nabla \phi\\
    &\leq \mathbb E \|s_{v_1}\|_\infty^2 \|\nabla \phi\|_g^2 \lesssim n \alpha^2 \alpha_0, 
\end{align*}
where in the second line we are rewriting $v_1^\top Dg(\nabla \phi)$ as $\nabla\phi^\top Dg(v_1)$ which is true due to the symmetry of the derivatives of the metric on Hessian manifolds, i.e. $\partial_k g_{ij} = \partial_i g_{jk} = \partial_j g_{ik}$. Furthermore, we used Lemma~\ref{lem:philemma} in the last line.
For the second part of first term of $M$, note that $D^2\phi = g$, so the Frobenius norm is at most $n$ automatically. Next, for the first part of the second term of $M$, again based on Lemma~\ref{lem:mterms}
\begin{align*}
    \mathbb E_{v_1, v_2 \sim \mathcal N(0, g^{-1})} (v_1^TDg(\xi) v_2)^2
    &= \mathbb E v_2^\top Dg(\xi) v_1v_1^\top Dg(\xi)v_2\\
    &= \mathbb E \xi^\top Dg(v_2)g^{-1}Dg(v_2)\xi \\
    & \leq \mathbb E \|s_{v_2}\|_\infty^2 \xi^\top g \xi \leq n,
\end{align*}
where in the last line we used Lemma~\ref{lem:boundonxi}.
For the second part of the second term of $M$, from Lemma~\ref{lem:mterms}:
\begin{align*}
    \mathbb E_{v_1, v_2} (v_1^\top \D(g\zeta) v_2)^2 \lesssim 
    \mathbb E \tr^2(g^{-1}\D g(v_1)g^{-1}\D g(v_2))
    + \mathbb E \tr^2(g^{-1}\D^2g(v_1. v_2))
\end{align*}
for the first part
\begin{align*}
\mathbb E \tr^2(g^{-1}\D g(v_1)g^{-1}\D g(v_2)) &= \mathbb E_{v_1, v_2} ( \mathbb E_{v \sim \mathcal N(0,g^{-1})}v^\top \D g(v_1)g^{-1}\D g(v_2)v)^2\\
&\leq \mathbb E_{v_1, v_2, v} (v^\top \D g(v_1)g^{-1}\D g(v_2)v)^2\\
&= \mathbb E_{v_1, v_2, v} v^\top \D g(v_1)g^{-1}\D g(v)v_2v_2^\top \D g(v)g^{-1}\D g(v_1)v\\
&= \mathbb E v^\top \D g(v_1)g^{-1}\D g(v)g^{-1} \D g(v)g^{-1}\D g(v_1)v\\
&\leq \mathbb E \|s_v\|_\infty^2 \|s_{v_1}\|_\infty^2 \|v\|_g^2 \lesssim n. 
\end{align*}
For the second part:
\begin{align*}
    \mathbb E_{v_1, v_2} tr^2(g^{-1}\D^2g(v_1, v_2))
    &\leq \mathbb E_{v_1, v_2, v} (v^\top \D^2g(v_1, v_2) v)^2\\
    &= \mathbb E_{v_1, v_2, v} (v^\top \D^2g(v_1, v) v_2)^2\\
    &=\mathbb E_{v_1, v_2, v} v^\top \D^2g(v_1, v) v_2 {v_2}^\top \D^2g(v_1, v) v\\
    &=\mathbb E_{v_1, v} v^T \D^2g(v_1, v) g^{-1} \D^2g(v_1, v) v\\
    &\leq \mathbb E_{v_1, v} \|s_{v_1}\|_\infty^2 \|s_v\|_\infty^2 \|v\|_g^2 \lesssim n. 
\end{align*}
\end{proof}

Combining Lemmas~\ref{lem:Mfrobeniusnorm} and~\ref{lem:Rfrobeniusbound} concludes 
\begin{align*}
    R_1 \leq (c^2 + \alpha\sqrt{\alpha_0})\sqrt n.
\end{align*}

\subsection{Bounding $R_2$}\label{sec:R2}
Here we state the bound on $R_2$.
\begin{lemma}\label{lem:R2bound}
    For point $x = \gamma_s(t)$ on a $(c,\delta)$-nice Hamiltonian curve with $v= \gamma'_s(t)$, namely that $\|s_{\gamma'_s}\|_\infty \leq c$ and $\|\gamma'_s\|_g \leq c \sqrt n$ along the curve up to time $t = \delta$, suppose now we move on the unit direction $z$ parameterized by $s$. Then, the change in the trace of the operator $\Phi$ can be bounded as
    \begin{align*}
        R_2 = |\frac{d}{ds}tr(\Phi(t))| \leq n(c^2+\sqrt{\alpha_0}\alpha).
    \end{align*}
\end{lemma}
\begin{proof}
Directly from Lemmas~\ref{lem:Mchangebound} and~\ref{lem:riccichange}.
\end{proof}
In sections~\ref{sec:Mopchangebound} and~\ref{sec:riccichangebound}, we bound the change in the $M$ part and the Ricci part of $\Phi$ respectively.

\subsubsection{Bounding the change in Operator $M_x$}\label{sec:Mopchangebound}
Given a distribution $e^{-\phi(x)}$ that we want to sample from, we study the properties of the derivatives of the corresponding operator $M$ which is defined as 
\begin{align}
    M_x(v_1,v_2) = \langle \nabla_{v_1} \mu(x), v_2 \rangle,\label{eq:Mdefinition}
\end{align}
where 
\begin{align*}
    \mu(x) = \nabla_{g} (\alpha\phi)(x) + \frac{1}{2}g^{-1}tr(g^{-1}Dg) = 
    g^{-1}\D(\alpha\phi) + \frac{1}{2}g^{-1}tr(g^{-1}\D g),
\end{align*}

Recall from Lemma~\ref{lem:mterms}:
\begin{align}
    \text{LHS} = \langle \nabla_{v_1}(\nabla (\alpha\phi)) + \frac{1}{2}\nabla_{v_1}(g^{-1}tr(g^{-1}\D g)), v_2\rangle = \langle A_1(v_1), v_2\rangle + \langle A_2(v_1), v_2 \rangle.\label{eq:mform}
\end{align}

where we defined matrices $A_1(v_1)$ and $A_2(v_1)$. Here we introduce the main lemma of this section which bounds the derivative of the trace of $M$:
\begin{lemma}[Bound on the change of operator $M$]\label{lem:Mchangebound}
For operator $M$ defined in~\eqref{eq:Mdefinition} for any unit direction $z$ we have
\begin{align*}
    |\D(\tr{M(x)})(z)| \lesssim (1 + \sqrt{\alpha_0} \alpha)n\|z\|_g. 
\end{align*}
\end{lemma}
\begin{proof}
    To prove Lemma~\ref{lem:Mchangebound}, we bound the derivative of $tr(A_1)$ and $tr(A_2)$ in direction $z$ separately in Lemmas~\ref{lem:firstpartM} and~\ref{lem:secondpartM}. As a result, the proof of Lemma~\ref{lem:Mchangebound} directly follows from Lemmas~\ref{lem:firstpartM} and~\ref{lem:secondpartM}.
\end{proof}
We start from $tr(A_1)$ in the following Lemma.
\begin{lemma}[Trace of $A_1$]\label{lem:firstpartM}
    Regarding the operator $A_1(v_1) = \nabla_{v_1}(\nabla (\alpha\phi))$, we have
    \begin{align*}
       \D(\tr(A_1))(z) \lesssim \alpha \sqrt{n\alpha_0} \|z\|_g.
    \end{align*}
\end{lemma}

\begin{proof}
Note that from Lemma~\ref{lem:mterms}:
\begin{align*}
    \D(\tr(A_1))(z) &= \langle \D(\nabla (\alpha\phi))[v_1] + \frac{1}{2}g^{-1}\D g(\nabla (\alpha\phi))v_1, v_2\rangle\\
    & = v_2^\top \D g(\nabla (\alpha\phi))v_1 + v_2^\top \D^2\phi v_1.\numberthis\label{eq:Mfirstterm}
\end{align*}

For the second part, note that $D^2(\alpha\phi) = \alpha g$. Hence
\begin{align*}
    \D(\tr(g^{-1}\D^2\phi))(z) = 0.
\end{align*}

So we only need to handle the derivative of the first part. First, we bound the $g$-norm of the vector $\nabla \phi$ in the following helper lemma.
\begin{lemma}\label{lem:philemma}
 For the gradient of the potential $\phi$ we have
\begin{align*}
    \|\D (\alpha\phi)\|_{g^{-1}} \leq \alpha \sqrt{n\alpha_0}.
\end{align*}
\end{lemma}
\begin{proof}
We decompose the potential as $(\alpha\phi) = \phi_1 +  \phi_2$ for 
\begin{align*}
&\phi_1 = \alpha \alpha_0\log\det{A^T W^{1-2/p}A},\\
&\phi_2 = \alpha \alpha_0\frac{n}{m}\sum_i \log(s_i),
\end{align*}
 where $\W$ are the $p$-Lewis weights.

Now using Lemma~\ref{lem:lsbarriergradient}, we have
\begin{align*}
    \D\phi_1(x) = \alpha_0 \alpha A^\top w.
\end{align*}
\begin{align*}
    \|\D \phi_1\|_{g^{-1}} & = \alpha\alpha_0\sqrt{w_x^\top \A{g}^{-1}\A^\top w_x}\\
    & \leq \alpha\sqrt{\alpha_0} \sqrt{{w_x^{1/2}}^\top \Pxhat w_x^{1/2}} \leq \alpha\sqrt{\alpha_0} \|w_x^{1/2}\|_2 \leq \alpha\sqrt{\alpha_0} \sqrt n.
\end{align*}
where $\Pxhat = \mathbf{P}(\W^{1/2}\A) = \mathbf{P}(\W^{1/2}\A)$ is the projection matrix regarding the reweighted matrix $\W^{1/2}\A$ by the Lewis weights $w_x$. Note that we are using Lemma~\ref{lem:lewisweightbarrier} to conclude that $\A^\top \W \A \preccurlyeq g$.
For the log barrier part, similarly:
\begin{align*}
    \|\D\phi_2\|_{g^{-1}} &= \alpha \alpha_0 \sqrt{(\frac{n}{m})^2 1^\top \A g^{-1}\A^\top 1}\\
    & \leq \alpha\alpha_0 \sqrt{(\frac{n}{m})^2 1^\top \A g_2^{-1}\A^\top 1} \leq \alpha \sqrt{n\alpha_0},
\end{align*}
which completes the proof. Now we handle the first term of the $M$ operator, namely the first term in~\eqref{eq:mform} using the helper Lemmas.
\end{proof}

Now we got back to bound the first term in~\eqref{eq:Mfirstterm}, which we can expand as
\begin{align}
   \D(\tr(g^{-1}\D g(\nabla (\alpha\phi))))[z]
   & = \tr(g^{-1}\D g(z, \nabla (\alpha\phi))) + \tr(g^{-1}\D g(\D(\nabla (\alpha\phi))[z])) - \tr(g^{-1}\D g(z)g^{-1}\D g(\nabla (\alpha\phi))).\label{eq:subpart}
\end{align}

For the first term in~\eqref{eq:subpart}, according to Lemma~\ref{lem:philemma}:
\begin{align*}
    &\tr(g^{-1}\D g(z, \nabla(\alpha\phi)))\\
    & =  \mathbb E_{v' \sim \mathcal N(0,g^{-1})} {v'}^\top \D g(z, \nabla(\alpha\phi))v' \\
    & =
     \mathbb E_{v' \sim \mathcal N(0,g^{-1})} {v'}^T \D g(z, v')\nabla(\alpha\phi)\\
    & \leq \mathbb E_{v' \sim \mathcal N(0,g^{-1})} \|s_z\|_\infty \|s_{v'}\|_\infty \sqrt{{v'}^\top g v'} \sqrt{{\nabla(\alpha\phi)}^\top g \nabla(\alpha\phi)}\\ 
    & \leq \alpha\sqrt{\alpha_0} \sqrt n \sqrt n  \|s_z\|_\infty\\
     & \leq \alpha\sqrt{\alpha_0} n  \|z\|_g,\numberthis\label{eq:sameupperbound}
\end{align*}
where we used Lemma~\ref{lem:momentbound1} to bound $\mathbb E_{v'} \|s_{v'}\|_\infty \sqrt{{v'}^T g v'}$ and used Lemma~\ref{lem:infwithgnorm}. For the second term in~\eqref{eq:subpart}, we follow a similar reasoning:
\begin{align*}
    tr(g^{-1}\D g(\D(\nabla (\alpha\phi))(z))) 
    &=\mathbb E_{v'\sim \mathcal N(0,g^{-1})} {v'}^\top \D g(\D(\nabla (\alpha \phi))(z))v'\\
    &=\mathbb E_{v'\sim \mathcal N(0,g^{-1})} {v'}^\top \D g(v')\D(\nabla (\alpha \phi))(z)\\
    &\leq \mathbb E_{v'}\|v'\|_\infty \sqrt{{v'}^\top g v'} \sqrt{{\D(\nabla (\alpha\phi))(z)}^\top g \D(\nabla (\alpha\phi))(z)}.\numberthis\label{eq:factoredargument}
\end{align*}
Therefore, bounding $\tr(g^{-1}\D g(\D(\nabla (\alpha\phi))(z)))$ boils down to bounding $\|\D(\nabla (\alpha \phi))(z)\|_g$. Focusing on the subterm of ${\D(\nabla (\alpha\phi))(z)}^\top g \D(\nabla (\alpha\phi))(z)$ regarding $\phi_1$, namely $\|\D(\nabla \phi_1)(z)\|_g$
\begin{align*}
    \D (\nabla \phi_1)(z)^\top g \D (\nabla \phi_1)(z) & \leq \D^2(\phi_1)(z)^\top g^{-1}  \D^2(\phi_1)(z) + \D\phi_1^\top g^{-1} \D g(z)g^{-1}\D g(z)g^{-1} \D\phi_1\\
    &\leq \alpha^2 w_x^\top \Sxz\A g^{-1} \A^\top \Sxz w_x + \alpha^2 {w'_{x}}^\top \A g^{-1}\A^\top w'_x  + \|s_z\|_\infty^2 \D\phi_1^\top g^{-1}\D\phi_1\\
    & \leq \alpha_0\alpha^2\|\Sxz w_x^{1/2}\|_2^2 + \alpha_0\alpha^2 \|w'_x/w_x^{1/2}\|_2^2 + \alpha_0 \alpha^2 \|s_{x,z}\|_\infty^2 n\\
    & \leq \alpha_0\alpha^2 \|s_{x,z}\|_\infty^2 n\\
     &\leq \alpha_0\alpha^2 \|z\|_g^2 n,\numberthis\label{eq:Dphifirst}
\end{align*}
where we used Lemma~\ref{lem:philemma} and Lemma~\ref{lem:infwithgnorm}. Similarly for $\phi_2$:
\begin{align*}
    \D (\nabla \phi_2)(z)^\top g \D (\nabla \phi_2)(z) & \leq \D^2(\phi_2)(z)^\top g^{-1}  \D^2(\phi_2)(z) + \D\phi_2^\top g^{-1} \D g(z)g^{-1}\D g(z)g^{-1} \D\phi_2\\
    &\leq \alpha_0\alpha^2 (\frac{n}{m})^2 1^\top \Sxz\A g_2^{-1} \A^\top\Sxz 1 +  \alpha^2 \|s_{x,z}\|_\infty^2 \D\phi_1^\top g^{-1}\D\phi_1\\
    & \leq \alpha_0 \alpha^2 \frac{n}{m}\|s_{x,z}\|_2^2 + \alpha_0\alpha^2 \|s_{x,z}\|_\infty^2 n\\
    & \leq \alpha_0 \alpha^2 \|s_{x,z}\|_\infty^2 n,\\
    &\leq \alpha_0 \alpha^2 n\|z\|_g^2.\numberthis\label{eq:Dphisecond}
\end{align*}
 where we used Lemma~\ref{lem:infwithgnorm}. Combining the above with the inequality
\begin{align*}
     \sqrt{\D(\nabla (\alpha\phi))(z) g \D(\nabla (\alpha\phi))(z)} \leq \sqrt{\D(\nabla \phi_1)(z) g \D(\nabla \phi_1)(z)} + \sqrt{\D(\nabla \phi_2)(z) g \D(\nabla \phi_2)(z)}
\end{align*}
and plugging back into Equation~\eqref{eq:factoredargument} implies the following bound on the second term in Equation~\eqref{eq:subpart}, we have for the second term in Equation~\eqref{eq:subpart}:
\begin{align}
   \tr(g^{-1}\D g(\D(\nabla (\alpha\phi))(z))) \leq n\sqrt \alpha_0\alpha \|z\|_g.\label{eq:secondtermbound}
\end{align}

For the third term in~\eqref{eq:subpart}, we reduce it to the first group of terms. Note that
\begin{align*}
    \tr(g^{-1}\D g(z)g^{-1}\D g(\nabla (\alpha\phi)))
    & =  \mathbb E_{v' \sim \mathcal N(0,g^{-1})} {v'}^T \D g(z)g^{-1}\D g( v')\nabla (\alpha\phi)\\
    & = \mathbb E_{v'}\sqrt{{v'}^\top \D g(z)g^{-1}\D g(z)v'} \sqrt{\nabla (\alpha\phi)^\top \D g(v')g^{-1}\D g(v') \nabla (\alpha\phi)}\\
    & \leq \mathbb E_{v'}\|s_{x,v'}\|_\infty\|s_{x,z}\|_\infty \|v'\|_g \sqrt{\nabla (\alpha\phi)^\top g^{-1} \nabla (\alpha\phi)}\\
    &\leq \alpha\sqrt{\alpha_0} n  \|z\|_g,\numberthis\label{eq:thirdtermupperbound}
\end{align*}
which is the same upper bound obtained in Equation~\eqref{eq:sameupperbound} and~\eqref{eq:secondtermbound}. Hence, combining Equations~\eqref{eq:sameupperbound},~\eqref{eq:secondtermbound}, and~\eqref{eq:thirdtermupperbound} we conclude
\begin{align*}
       \D(\tr(A_1))(z) \lesssim \alpha \sqrt{n\alpha_0} \|z\|_g.
\end{align*}
\end{proof}

Next, we focus on the second term in~\eqref{eq:mform} and bound the derivative of the trace of the operator $A_2(v_1) = \nabla_{v_1}(g^{-1}\tr(g^{-1}\D g))$. 

\begin{lemma}[Trace of $A_2$]\label{lem:secondpartM}
    For operator $A_2(v_1) = \nabla_{v_1}(g^{-1}\tr(g^{-1}\D g))$ as defined in Equation~\eqref{eq:mform} we have
    \begin{align*}
        |\D(\tr(A_2))(z)| \leq n\|z\|_g.
    \end{align*}
\end{lemma}

Let
\begin{align*}
    \xi = g^{-1}\tr(g^{-1}\D g).
\end{align*}

\begin{proof}
From Lemma~\ref{lem:mterms}, we have
\begin{align}
    \langle \nabla_{v_1}(g^{-1}tr(g^{-1}\D g)), v_2\rangle & = 
   v_2^T \D g(\xi)v_1 + v_2^T \D(g\xi)v_1,\numberthis\label{eq:nonphipart}
\end{align}
We bound the derivatives of the two terms in Equation~\eqref{eq:nonphipart} separately in Lemmas~\ref{lem:thirdpartM} and~\ref{lem:forthpartM}. Hence, the proof of Lemma~\ref{lem:secondpartM} directly follows from these Lemmas.
\end{proof}
We start from bounding the derivative of the first term in Equation~\eqref{eq:nonphipart}, i.e. we wish to bound $|D(tr(g^{-1}Dg(\xi)))(z)|$.

\begin{lemma}\label{lem:thirdpartM}
    Regarding the first quadratic form in Equation~\eqref{eq:nonphipart}, we can bound its trace as
    \begin{align*}
        |\D(tr(g^{-1}\D g(\xi)))(z)| \lesssim n\|z\|_g.
    \end{align*}
\end{lemma}
\begin{proof}
To this end, we repeat a similar arguemnt as we did in Equation~\eqref{eq:subpart} for bounding 
$$\D(\tr(g^{-1}\D g(\nabla(\alpha \phi))))[z].$$ In particular, our argument regarding  $\nabla (\alpha\phi)$ in Equations~\eqref{eq:Dphifirst} and~\eqref{eq:Dphisecond} only cares about the bound on $\|\nabla (\alpha\phi)\|_{g}$ and $\|\D(\nabla (\alpha\phi))(z)\|_{g}$. We show a similar bound for $\xi$. As a warmup, we start by bounding the norm $\|\xi\|_g$, then we move on to bounding $\|\D(\xi)(z)\|_g$.

\begin{lemma}\label{lem:boundonxi}
We have 
\begin{align*}
    \|\xi\|_g \leq \sqrt n.
\end{align*}
\end{lemma}
\begin{proof}
We have
\begin{align*}
    \|\xi\|_g^2 & = {\tr(g^{-1}\D g)}^\top g^{-1}\tr(g^{-1}\D g)\\
    & = \mathbb E_{v,v'\sim \mathcal N(0,g^{-1})} {v}^\top\D g(v)g^{-1}\D g(v')v'\\
    & \leq \mathbb E_{v\sim \mathcal N(0,g^{-1})} {v}^\top \D g(v)g^{-1}\D g(v)v\\
    &\leq \mathbb E_v \|v\|_\infty^2 v^\top g v\lesssim n,\numberthis\label{eq:xinormbound} 
\end{align*}
where $\tr(g^{-1}\D g)$ is a vector with its $i$th entry equal to $tr(g^{-1}{D_ig})$. The first inequality above is due to Cauchy-Schwarz, and the second one is due to Lemma~\ref{lem:momentbound1}.
\end{proof}

Furthermore, we have the following bound on $\|\D(\xi)(z)\|_g^2$:
\begin{lemma}\label{lem:zetabound}
For the derivative of $\xi$ in direction $z$ we have
\begin{align*}
\|\D(\xi)(z)\|_g^2 \lesssim n.
\end{align*}
\end{lemma}
\begin{proof}
Note that
\begin{align*}
    \|\D(\xi)(z)\|_g^2 &= {\tr(g^{-1}\D g)}^\top g^{-1}\D g(z)g^{-1}\D g(z)g^{-1}\tr(g^{-1}\D g)  &&\text{LHS}_1\\
    &+ \tr(g^{-1}\D g(z)g^{-1}\D g)^\top g^{-1}tr(g^{-1}\D g(z)g^{-1}\D g) && \text{LHS}_2\\
    &+ {\tr(g^{-1}\D^2g(z,.))}^\top g^{-1}\tr(g^{-1}\D^2g(z,.)). &&\text{LHS}_3
\end{align*}
For the first term above, 
\begin{align*}
{tr(g^{-1}Dg)}^Tg^{-1}Dg(z)g^{-1}Dg(z)g^{-1}tr(g^{-1}Dg)
&\leq {\tr(g^{-1}\D g)}^\top g^{-1/2}(g^{-1/2}\D g(z)g^{-1/2})^2 g^{-1/2}\tr(g^{-1}\D g)\\
&\leq \|z\|_g^2 {\tr(g^{-1}\D g)}^\top g^{-1}\tr(g^{-1}\D g).
\end{align*}

following our argument in~\eqref{eq:xinormbound}:
\begin{align*}
    \text{LHS}_1 \leq \|s_z\|_\infty^2 n.
\end{align*}
For the second term, we write the second $g^{-1}$ within the tracec as an expectation $\mathbb E_{v' \sim \mathcal N(0,I)} v'v'^T$, i.e.
\begin{align*}
    \tr(g^{-1}\D g(z)g^{-1}\D g) &= \tr(\D g g^{-1}\D g(z)g^{-1})\\
    &= \mathbb E_{v'} \tr(\D g g^{-1}\D g(z)v'v'^\top)\\
    &= \mathbb E_{v'} \D g(v')g^{-1}\D g(z)v'.
\end{align*}
Therefore, using independent normal vectors $v,v' \sim \mathcal N(0,g^{-1})$, we can rewrite the second term as
\begin{align*}
    \text{LHS}_2 &= \mathbb E_{v,v'}v^\top\D g(z)g^{-1}\D g(v)g^{-1} \D g(v')g^{-1}\D g(z)v'\\
    &= \mathbb E_{v,v'} z^\top \D g(v)g^{-1}\D g(v)g^{-1} \D g(v')g^{-1}\D g(v')z\\
    &\leq \mathbb E_v z^\top \D g(v)g^{-1}\D g(v)g^{-1} \D g(v)g^{-1}\D g(v)z\\
    &\leq \mathbb E_v \|z\|_g^2 \|s_v\|_\infty^4 \lesssim \|z\|_g^2.
\end{align*}
where the first inequality follows from Cauchy-Schwarz and the second one follows from Lemma~\ref{lem:exponentiallownerinequality} and the fact that $\D g(v) \lesssim \|s_v\|_\infty g$. For the third term similarly
\begin{align*}
    \text{LHS}_3 &= \mathbb E_{v,v' \sim \mathcal N(0,g^{-1})} v^\top\D g(z,v)g^{-1}\D g(z,v')v'\\
    &= \mathbb E_{v,v'} z^\top\D g(v,v)g^{-1}\D g(v',v')z\\
    &\leq \mathbb E_{v} z^\top\D g(v,v)g^{-1}\D g(v,v)z\\
    & \leq \mathbb E_v\|z\|_g^2 \|s_v\|_\infty^4 \lesssim \|z\|_g^2.
\end{align*}
Combining all three bounds similar to our argument for $\nabla (\alpha\phi)$ we conclude
\begin{align}
    \|\D(\xi)(z)\|_g \lesssim \|s_z\|_\infty \sqrt n \leq \|z\|_g \sqrt n.\label{eq:dxinormbound}
\end{align}
\end{proof}

According to Lemma~\ref{lem:zetabound}, similar to our bound for $\nabla (\alpha\phi)$ by substituting $w$ with $D\xi(z)$ in Lemma~\ref{lem:factoredlemma} we get
\begin{align}
    |\tr(g^{-1}\D g(\D\xi(z)))|\leq \sqrt n \|\D\xi(z)\|_g \lesssim \|z\|_g n.\label{eq:avalii}
\end{align}
Moreover, according to Lemma~\ref{lem:factoredlemma2} and Lemma~\ref{lem:boundonxi}:
\begin{align}
    |\tr(g^{-1}\D g(\xi, z))| \leq \sqrt n \|\xi\|_g \|z\|_g \leq n\|z\|_g.\label{eq:dovomii}
\end{align}
Further, using Lemma~\ref{lem:factoredlemma} combined with Lemma~\ref{lem:boundonxi}:
\begin{align}
    |\tr(g^{-1}\D g(z)g^{-1}\D g(\xi))| \leq \sqrt n \|z\|_g \|\xi\|_g \leq n \|z\|_g.\label{eq:sevomii}
\end{align}
Hence, combining Equations~\eqref{eq:avalii},~\eqref{eq:dovomii}, and~\eqref{eq:sevomii},
\begin{align*}
    |\D(\tr(g^{-1}\D g(\xi)))(z)| &\lesssim |\tr(g^{-1}\D g(\D\xi(z)))| + |\tr(g^{-1}\D g(\xi,z))| + |\tr(g^{-1}\D g(z)g^{-1}\D g(\xi))|\\
    &\lesssim n\|z\|_g,
\end{align*}
which completes the bound for the trace of the first part $Dg(\xi)$ of the operator in Equation~\eqref{eq:nonphipart}.
\end{proof}

Finally, we move on to bound derivative of the trace of the second operator in Equation~\eqref{eq:nonphipart}, namely $D(tr(g^{-1}D(g\xi)))(z)$.

\begin{lemma}\label{lem:forthpartM}
   We can bound the derivative of the trace of the second operator in Equation~\eqref{eq:nonphipart} as
   \begin{align*}
       |\D(\tr(g^{-1}\D(g\xi)))(z)| \lesssim n\|z\|_g.
   \end{align*}
\end{lemma}
\begin{proof}
Recall from Lemma~\ref{lem:mterms}:
\begin{align}
    v_2^\top \D(g\xi)v_1 = -\tr(g^{-1}\D g(v_1)g^{-1}\D g(v_2))
    + \tr(g^{-1}\D^2g(v_1. v_2))\\
    = -v_1^\top B_1 v_2 + v_1^\top B_2 v_2.\label{eq:recall}
\end{align}

Now we wish to calculate the derivative of the trace of this operator, namely 
\begin{align}
\D(\tr(g^{-1}\D(g\xi)))(z).\label{eq:tmp2}
\end{align}
We separate the case when the derivation w.r.t $z$ is taken with respect to the outer $g^{-1}$ in~\eqref{eq:tmp2}.
First, we calculate the derivative with respect to the outer $g^{-1}$ regarding the term $\tr(g^{-1}B_1)$:
\begin{align}
    |\tr(\D(g^{-1})(z) B_1)| = |\tr(g^{-1}\D g(z)g^{-1}B_1)|.\label{eq:tmp3first}
\end{align}
Note that
\begin{align*}
    v_1^\top B_1 v_2 = \tr(g^{-1}\D g(v_1)g^{-1}\D g(v_2)).
\end{align*}
Note that this 2-form is symmetric and PSD since
\begin{align*}
    \tr(g^{-1}\D g(v_1)g^{-1}\D g(v_1)) = \tr((g^{-1/2}\D g(v_1)g^{-1/2})^2) \geq 0.
\end{align*}
Moreover, note that
\begin{align*}
    g^{-1}\D g(z)g^{-1} \leq \|s_z\|_\infty g^{-1}.
\end{align*}
Hence, Equation~\eqref{eq:tmp3first} can further be upper bounded as
\begin{align*}
    \|s_z\|_\infty tr(g^{-1} B_1) = \|s_z\|_\infty  \mathbb E_{v' \sim \mathcal N(0, g^{-1})} \tr(g^{-1}\D g(v')g^{-1}\D g(v')).
\end{align*}
But we have already bounded the operator norm of $\tr(g^{-1}\D g(v')g^{-1}\D g(v'))$ in Lemma~\ref{lem:operatornormbound1} by $\tilde O(\|s_{v'}\|_\infty^2)$, which implies its trace can be at most $n\tilde O(\|s_{v'}\|^2)$. Taking expectation, we have
\begin{align*}
     \mathbb E_{v' \sim \mathcal N(0, g^{-1})} \tr(g^{-1}\D g(v')g^{-1}\D g(v')) \lesssim n.
\end{align*}
Hence, we conclude
\begin{align}
    |\tr(\D(g^{-1})(z)B_1)| \lesssim n\|s_z\|_\infty.\label{eq:conclude1}
\end{align}
 On the other hand, note that for the second term in Equation~\eqref{eq:recall}, there is a symmetry between the inner and outer $g^{-1}$:
 \begin{align*}
     \tr(g^{-1}B_2) = \tr(g^{-1}\D^2g[g^{-1}]).
 \end{align*}
 Hence, it is sufficient to bound when taking derivative with respect $z$ hit one of them, namely the inner $g^{-1}$.
 \\
Therefore, we move on to taking derivative with respect to the $D(g\xi)$ part of $tr(g^{-1}D(g\xi))$. For this, we can again use the trick of writing $g^{-1}$ as $\mathbb E_{v\sim \mathcal N(0,g^{-1})}$:
\begin{align*}
    \tr(g^{-1}\D(g\xi)) = \mathbb E_{v} v^\top \D(g\xi)v.
\end{align*}
But from Equation~\eqref{eq:recall}, we have
\begin{align*}
   \tr(g^{-1}\D(g\xi)) = -\tr(g^{-1}\D g(v)g^{-1}\D g(v))
    + \tr(g^{-1}\D^2g(v, v)). 
\end{align*}
Now taking derivative with respect to $z$:
\begin{align}
    |\D(\tr(g^{-1}\D(g\xi)))(z)| & \leq |\mathbb  E_v \D( \tr(g^{-1}\D g(v)g^{-1}\D g(v)))(z)| + |\mathbb E_v \D(\tr(g^{-1}\D g(v,v)))(z)|\\
    & \text{LHS}_1 + \text{LHS}_2.\label{eq:subpart2}
\end{align}
But for the first term in~\eqref{eq:subpart2}, we can write:
\begin{align*}
    \text{LHS}_1 & \leq 2\mathbb E_v|\tr(g^{-1}\D g(z)g^{-1}\D g(v)g^{-1}\D g(v))| + 2\mathbb E_v|\tr(g^{-1}\D^2g(v,z)g^{-1}\D g(v))|\\
    &\lesssim \mathbb E_v\|z\|_\infty \tr((g^{-1/2}\D g(v)g^{-1/2})^2) + \mathbb E_v\|g^{-1/2}\D g(v,z)g^{-1/2}\|_1\|g^{-1/2}\D g(v)g^{-1/2}\|_{op}\\
    & \lesssim \mathbb E_v\|s_v\|_\infty^2 \|s_z\|_\infty n\\
    &\leq \mathbb E_v \|s_v\|_\infty^2 \|z\|_g n\\
    &\lesssim \|z\|_g n.\numberthis\label{eq:conclude2}
\end{align*}
For the second term in~\eqref{eq:subpart2}:
\begin{align*}
    \text{LSH}_2 & \leq |\tr(g^{-1}\D g(z)g^{-1}\D g(v,v))|
    + |\tr(g^{-1}\D g(v,v,z))| \\
    & \leq \|g^{-1/2}\D g(z)g^{-1/2}\|_{op}\|g^{-1/2}\D g(v,v)g^{-1/2}\|_1 + \|s_v\|_\infty^2 \|s_z\|_\infty \tr(g^{-1}g) \lesssim n \|s_z\|_\infty \|s_v\|_\infty^2\\
    &\leq \mathbb E_v n \|z\|_g\|s_v\|_\infty^2\\
    &\lesssim n\|z\|_g.\numberthis\label{eq:conclude3}
\end{align*}
where we used the third order self-concordance property of $g$ with respect to the infinity norm, as shown in section~\ref{sec:thirdorderselfconcordance}, and also Lemma~\ref{lem:infwithgnorm}. Combining Equations~\eqref{eq:conclude1},~\eqref{eq:conclude2}, and~\eqref{eq:conclude3} completes the porof of Lemma~\ref{lem:forthpartM}.
\end{proof}

\subsubsection{Bounding the change in the Ricci Tensor}\label{sec:riccichangebound}

First, we state the main result of this section, which is a bound on the change of the Ricci tensor.
\begin{lemma}[Bound on the change of Ricci tensor]\label{lem:riccichange}
Given the assumptions of Lemma~\ref{lem:R2bound}, we have
\begin{align*}
    \big|\frac{d}{ds} \Ricci(v, v)\big| \lesssim nc^2.
\end{align*}
\end{lemma}
Note that in the above, $v$ is implicitly a function of $s$ as well.
\begin{proof}
According to Lemma~\ref{lem:riccitensor} has two terms. We start analyzing the first term:
\paragraph{ $A_1:= -\frac{1}{4}\tr(g^{-1}\D g(v_1)g^{-1}\D g(v_2))$ term}

Taking derivative of this subterm of Ricci tensor in direction $z$:
\begin{align*}
\D A_1(z) = -\frac{1}{4}\tr(g^{-1}\D g(v,z)g^{-1}\D g(v)) + \frac{1}{4} \tr(g^{-1}\D g(v)g^{-1}\D g(z)g^{-1}\D g(v)).    
\end{align*}

Now we use Lemmas~\ref{lem:gsecondderivativebound} and~\ref{lem:gfirstderivativebound} to bound these terms:
\begin{align*}
    \tr(g^{-1/2}\D g(v,z)g^{-1/2} g^{-1/2}\D g(v)g^{-1/2}) & \leq 
    \|g^{-1/2}\D g(v,z)g^{-1/2}\|_F \|g^{-1/2}\D g(v)g^{-1/2}\|_F\\
    & \leq \|s_v\|_\infty^2 \|s_w\|_\infty\|g^{-1/2}gg^{-1/2}\|_F^2\\
    & \leq n \|s_v\|_\infty^2\|s_z\|_\infty\\
    & \lesssim n \|s_v\|_\infty^2 \|z\|_g \lesssim n c^2 \|z\|_g.
\end{align*}
Similarly
\begin{align*}
    & \tr(g^{-1/2}\D g(z)g^{-1/2}g^{-1/2}\D g(v)g^{-1/2}g^{-1/2}\D g(v)g^{-1/2}) \\
    &\leq \|g^{-1/2}\D g(v)g^{-1/2}\|_{op}\|g^{-1/2}\D g(z)g^{-1/2}\|_F \|g^{-1/2}\D g(v)g^{-1/2}\|_F \leq n \|s_v\|_\infty^2 \|z\|_g.
\end{align*}

\paragraph{Terms in the derivative of $A_1$ that involves the derivative of $v$}

Differentiating $v$ with respect to $z$, we get
\begin{align*}
   \tr(g^{-1}\D g(Dv(z))g^{-1}\D g(v)) &= \mathbb E_{v' \sim \mathcal N(0,g^{-1})} v'^\top \D g(\D v(z))g^{-1}\D g(v)v'\\
   & = \D v(z)^T\top \D g(v')g^{-1}\D g(v')v\\
   & \leq \sqrt{\D v(z)^\top\D g(v')g^{-1}\D g(v')\D v(z)}\sqrt{v^\top\D g(v')g^{-1}\D g(v')v}\\
   & \leq \|v'\|_\infty^2 \|\D v(z)\|_g \|v\|_g \leq n^{5/6} c + n^{1/2}c^2,
\end{align*}
where we used Lemma~\ref{lem:vderivative} to bound $\|\D v(z)\|_g$.
\end{proof}

\paragraph{Second part of the Ricci Tensor.}
We should take derivative of $v^\top Dg(g^{-1}tr(g^{-1}Dg))v$ in direction $z$, which is the second term in the Ricci tensor according to Lemma~\ref{lem:riccitensor}. As a warm up, we first bound the value of this term before taking derivative:
\paragraph{Before taking derivative w.r.t $z$}

Note that the second part of the Ricci tensor is
\begin{align*}
    v^\top\D g(v)g^{-1}\tr(g^{-1}\D g) & = \mathbb E_{v'} v^\top \D g(v)g^{-1} \D g(v')v'\\
    & \leq \sqrt{v^\top \D g(v)g^{-1}\D g(v)v}\sqrt{v'^\top \D g(v')g^{-1}\D g(v')v'}.
\end{align*}
Hence, we only need to bound one of the RHS terms with high probability. We have
\begin{align*}
    v^\top \D g(v)g^{-1}\D g(v)v \leq \|s_v\|_\infty^2 v^\top g v \leq n.
\end{align*}
Now to bound the derivative of this part of the Ricci tensor, first we pretend that $v$ is fixed. Then
\begin{align*}
    &\D(v^\top \D g(v)g^{-1}\D g(v)v)(z) \\
    &=
    v^\top \D g(v,z)g^{-1}\D g(v)v + v^\top \D g(v)g^{-1}\D g(z)g^{-1}\D g(v)v + v^\top \D g(v)g^{-1}\D g(v,z)v,
\end{align*}
which we further bound as 
\begin{align*}
    \D(v^\top \D g(v)g^{-1}\D g(v)v)(z) \lesssim \|s_v\|_\infty^2 \|z\|_\infty v^\top g v \lesssim n \|z\|_\infty c^2.
\end{align*}

Next, we take derivative in direction $z$ from the second term of the Ricci tensor.
\paragraph{Taking derivative in direction $z$.}
First, we differentiate the inner $g^{-1}$ term in $v^TDg(v)g^{-1}tr(g^{-1}Dg)$:
\begin{align*}
    \D(v^\top \D g(v)g^{-1}\tr(g^{-1}\D g))(z) &\rightarrow v^\top\D g(v)\tr(g^{-1}\D g(z)g^{-1}\D g)\\
    &= \mathbb E_{v'} v^\top\D g(v)g^{-1}\tr(v'^\top \D g g^{-1}\D g(z)v')\\
    & = \mathbb E_{v'} v^\top\D g(v)g^{-1}\D g(v')g^{-1}\D g(v')z\\
    & \lesssim \mathbb E \|s_{v'}\|_\infty^2  \|s_{v}\|_\infty \|s_{v}\|_g \|z\|_g \lesssim c^2 \sqrt n \|z\|_g.
\end{align*}
For the remaining derivatives we can substitute the inner $g^{-1}$ by $\mathbb E_{v \sim \mathcal N(0, g^{-1})} v'{v'}^T$. Now for the remaining derivatives which does not involve differentiating $v$:
\begin{align*}
    \mathbb E_{v'} |\D(v^\top\D g(v)g^{-1}\D g(v')v')(z)|
    &\leq \mathbb E_{v'}|v^\top\D g(v)g^{-1}\D g(v',z)v'| + \mathbb E_{v'} |v^\top\D g(v)g^{-1}\D g(z)g^{-1}\D g(v')v'|\\
    & \leq c^2 \|s_z\|_\infty n \leq nc^2 \|z\|_g. 
\end{align*}

Finally we have to check when $z$ differentiates $v$:
\begin{align*}
    \D(v^\top\D g(v)g^{-1}\tr(g^{-1}\D g))(z) &\rightarrow 
    \D v(z)^\top\D g(v)g^{-1}\tr(g^{-1}\D g)\\
    & = \mathbb E_{v'} \D v(z)^\top\D g(v)g^{-1}\D g(v')v'\\
    &\leq \mathbb E_{v'}\sqrt{\D v(z)^\top\D g(v)g^{-1}\D g(v) \D v(z)}\sqrt{v'^\top\D g(v')g^{-1}\D g(v')v'}\\
    &\leq \mathbb E_{v'}\|\D v(z)\|_g \|s_v\|_\infty \|s_{v'}\|_\infty \|v'\|_g \lesssim n^{1/2}(n^{1/3} + c)c,
\end{align*}
where we used Lemma~\ref{lem:vderivative} to bound $\|\D v(z)\|_g$.

\subsection{Bounding $R_3$}
Here we bound the parameter $R_3$ which is defined as the maximum possible value of the norm of $\Phi(t)\zeta(t)$, where $\zeta(t)$ is the parallel transport of the initial velocity. The idea is to bound the infinity norm of $\zeta(t)$ along the Hamiltonian curve, then show a more efficient bound compared to the naive operator norm of $\Phi(t)$ which works with both of the norms $\|s_{\zeta(t)}\|_\infty$ and $\|\zeta(t)\|_g$.

Recall the definition of the parameter $R_3$:
\begin{align*}
    \|\Phi(t) \zeta(t)\|_g \leq R_3
\end{align*}
where $\zeta(t)$ is the parallel transport of $\gamma'(0)$ along the Hamiltonian curve $\gamma(t)$. 
\begin{lemma}[Bound on $R_3$]\label{lem:R3bound}
Given that $\gamma$ is $(c,\delta)$-nice, we have
\begin{align*}
    R_3 \leq c^2(\sqrt n + n\delta) + n\delta c \alpha \sqrt{\alpha_0},
\end{align*}
up to time $\delta$.
\end{lemma}
\begin{proof}
From the definition of niceness, we have a $c$ upper bound on the infinity norm $\|s_{\gamma'}\|_\infty$. Using that, we can apply Lemma~\ref{lem:parallelinfnorm} to obtain 
\begin{align*}
    \|s_\zeta\|_\infty \leq \delta c \sqrt n.
\end{align*}
Finally combining this with Lemmas~\ref{lem:Roperatorbound} and~\ref{lem:Moperatorbound}:
\begin{align*}
   \|\Phi(t)\zeta\|_g \leq c^2\|\zeta\|_g + \|\zeta\|_\infty (c + \alpha\sqrt{\alpha_0})\sqrt n \leq c^2n^{1/2} + c\sqrt n \delta (c + \alpha\sqrt{\alpha_0}) \sqrt n = c^2(\sqrt n + n\delta) + n\delta c \alpha \sqrt{\alpha_0}.
\end{align*}
\end{proof}

Here we show a norm bound for $\Phi(t)$ which we used to bound $R_3$. To this end, we show bounds on the Riemann tensor $R(,v)v$ and operator $M$ separately in Lemmas~\ref{lem:Roperatorbound} and~\ref{lem:Moperatorbound}.
\begin{lemma}[Operator norm of random Riemann tensor]\label{lem:Roperatorbound}
Assuming $\|s_v\|_\infty \lesssim c, \ \|v\|_g \lesssim \sqrt n$, we have
\begin{align*}
    \|R(\ell,v)v\|_g \leq c^2\|\ell\|_g + c\sqrt n \|s_\ell\|_\infty \leq (c^2 + c\sqrt n) \|\ell\|_g.
\end{align*}
\end{lemma}
\begin{proof}
Similar to Lemma~\ref{lem:Rfrobeniusbound}, using the form of Riemann expansion in Equation~\eqref{eq:riemann}:
\begin{align*}
    \|R(\ell, v)v\|_g &\leq (\ell^\top \D g(v)g^{-1}\D g(v)g^{-1}\D g(v)g^{-1}\D g(v)\ell)^{1/2} \\
    & + (v^\top\D g(v)g^{-1}\D g(\ell)g^{-1}\D g(\ell)g^{-1}\D g(v)v)^{1/2}\\
    & \leq \|v\|_\infty^2 (\ell^\top g\ell)^{1/2} + \|v\|_\infty\|\ell\|_\infty\|v\|_g \\
    & \leq c^2\|\ell\|_g + c\sqrt n \|\ell\|_\infty. 
\end{align*}
\end{proof}

Next, we state a similar mix norm bound for operator $M$.

\begin{lemma}[Operator norm of $M$]\label{lem:Moperatorbound}
we have
\begin{align*}
    \|M(x)\ell\|_g \leq \|\ell\|_g + (1 + \alpha\sqrt{\alpha_0}) \sqrt{n}  \|s_\ell\|_\infty.
\end{align*} 
\end{lemma}
\begin{proof}
Recall from Lemma~\eqref{lem:mterms}:
\begin{align*}
    \langle M(x)v_1, v_2\rangle = \langle \nabla_{v_1}(\nabla \phi) + \frac{1}{2}\nabla_{v_1}(g^{-1}\tr(g^{-1}\D g)), v_2\rangle.
\end{align*}
Starting from the first part of the term $\langle \nabla_{v_1}(\nabla \phi), v_2\rangle$:
\begin{align*}
    \|g^{-1}\D g(\nabla \phi)\ell\|_g & = 
    \tr^{1/2}(\ell^\top \D g(\nabla \phi)g^{-1}\D g(\nabla \phi) \ell)\\
    & = \tr^{1/2}((\nabla \phi)^\top \D g(\ell)g^{-1}\D g(\ell)\nabla \phi)\\
    & \leq \|s_\ell\|_\infty \|\nabla \phi\|_g \leq \alpha \sqrt{n\alpha_0}\|s_\ell\|_{\infty}.
\end{align*}
Note that for the second part, $\D^2 \phi = g$, hence the corresponding operator is the identity and has operator norm one.
\\
Next, we move on to the second term of $M$ in~\eqref{eq:secondpartt}. For the first part of it from Equation~\eqref{eq:nonphipart}, we have:
\begin{align*}
    \|g^{-1}\D g(\xi)\ell\|_g &= \sqrt{\ell^\top \D g(\xi)g^{-1}\D g(\xi)\ell}\\
    & = \sqrt{\xi^\top \D g(\ell)g^{-1}Dg(\ell)\xi}\\
    & \leq \|s_\ell\|_\infty\sqrt{\xi^\top g \xi} = \|s_\ell\|_\infty \sqrt n.
\end{align*}
where we used Lemma~\ref{lem:zetabound}.
For the second part, note that from Equation~\eqref{eq:Mlastterm}:
\begin{align}
    v_2^\top \D(g\xi) v_1 = \tr(g^{-1}\D g(v_1)g^{-1}\D g(v_2)) + \tr(g^{-1}\D g(v_1, v_2)).\label{eq:tmp}
\end{align}
Starting from the first part, now we rewrite this term in a better way as

\begin{align*}
\tr(g^{-1}\D g(v_1)g^{-1}\D g(v_2)) = \mathbb E_v \tr(vv^\top\D g(v_1)g^{-1}\D g(v_2))
= \mathbb E_v v^\top \D g(v_1)g^{-1}\D g(v_2)v = \mathbb E_v v_1^\top \D g(v)g^{-1}\D g(v)v_2.
\end{align*}
Now due to Lemma~\ref{lem:operatornormbound1} the norm of the corresponding operator is one:
\begin{align}
    \mathbb E\|g^{-1}\D g(v)g^{-1}\D g(v)\ell\|_g 
     \leq \mathbb E_v \|s_v\|_\infty^2 \|\ell\|_g \lesssim \|\ell\|_g.\label{eq:opnormxi}
\end{align}
For the second part in~\eqref{eq:tmp}, we write it as
\begin{align*}
    \tr(g^{-1}\D g(v_1, v_2)) = \mathbb E_v tr(vv^\top \D g(v_1, v_2))
    = \mathbb E v_1\D g(v,v)v_2.
\end{align*}
Hence, the operator norm is bounded as
\begin{align*}
    \mathbb E\|g^{-1}\D g(v,v)\ell\|_g \leq \mathbb E \|s_v\|_\infty^2\|\ell\|_g \lesssim\|\ell\|_g.
\end{align*}
\end{proof}


Next, we show a bound on the derivative of the infinity norm of the parallel transported vector $\zeta$ given that we know the infinity norm of $\gamma'$ is constant (randomness + stability).
\begin{lemma}[Infinity norm of the parallel transport]\label{lem:parallelinfnorm}
Given $\delta \leq \frac{1}{c}$ and a $(c, \delta)$-nice Hamiltonian curve $\gamma$, we have for $t \leq \delta$:
\begin{align*}
    \|s_{\zeta(t)}\|_\infty \leq \delta c \sqrt n,
\end{align*}
where $\zeta$ is the parallel transport of $\gamma'(0)$ along the curve.
\end{lemma}
\begin{proof}
As $\zeta$ is the parallel transport vector, from opening up the covariant derivative being zero:
\begin{align*}
    \frac{d}{dt}(A\zeta) &= A\zeta' - (A\zeta)\odot(A\gamma')\\
    &= -\frac{1}{2}Ag^{-1}\D g(\gamma')\zeta - (A\zeta)\odot (A\gamma'),
\end{align*}
which implies using Lemma~\ref{lem:infwithgnorm}:
\begin{align*}
   \|\frac{d}{dt}(A\zeta)\|_\infty & \lesssim \|Ag^{-1}\D g(\gamma')\zeta\|_\infty + \|s_{\gamma'}\|_\infty \|s_{\zeta}\|_\infty\\
   & \lesssim \|g^{-1}\D g(\gamma')\zeta\|_g + \|s_{\gamma'}\|_\infty \|s_{\zeta}\|_\infty\\
   & \lesssim c\|\zeta\|_g + c\|s_{\zeta}\|_\infty \leq c\sqrt n + c\|s_\zeta\|_\infty,
\end{align*}

where we used $\|s_{\gamma'}\|_\infty \lesssim c$ from the definition of niceness and the fact that parallel transport preserves the norm of $\zeta$ and $\|\zeta(0)\|_g = \|\gamma'(0)\|_g \leq \sqrt n$.
This ODE implies to avoid blow up we should pick $\delta \lesssim \frac{1}{c}$. Under this condition, we further get
\begin{align*}
    \|s_\zeta\|_\infty \lesssim \delta c \sqrt n,
\end{align*}
which completes the proof.
\end{proof}

In the next section, we show the stability of the infinity norm and the manifold norm of $\gamma'$ along the curve $c_t(s)$ for $s = 0$ to time $\frac{1}{n^{1/3}}$, where $c_t(s) = \gamma_s(t)$ is defined for a fixed time $t$. 

\section{Stability of Hamiltonian curves}\label{sec:stability}
In this section, we show that the niceness property holds for Hamiltonian curves with high probability, and is stable in a family of Hamiltonian curves.

\subsection{Stability of the niceness property}
Here we show that niceness property of Hamiltonian curves is stable. 
\begin{lemma}[Stability of norms]
For a family of Hamiltonian curves $\gamma_s(t)$, given that $\gamma_0$ is $(c,\delta)$-nice, then $\gamma_s(t)$ is  also $(O(c), \delta)$-nice for all $0 \leq s \leq \delta$. In other words, given that for all $0 \leq t \leq \delta$ we have $\|s_{\gamma'_0(t)}\|_\infty \leq c$ and $\|\gamma'_0(t)\|_g \leq \sqrt n$, then for all $0 \leq t \leq \delta$ and $0\leq s\leq \delta$ under the condition 
\begin{align*}
    \delta^2 \lesssim \frac{1}{(c^2 + \alpha 
    \sqrt \alpha_0)\sqrt n}.
\end{align*}  
we have:
\begin{align*}
    &\|s_{\gamma'_s(t)}\|_\infty \leq 2c,\\
    &\|\gamma'_s(t)\|_g \leq 2\sqrt n.
\end{align*}
\end{lemma}
\begin{proof}
Suppose we denote the time until which we run the Hamiltonian curve by $\delta$, i.e. $0 \leq t \leq \delta$. 
Suppose the argument is not true, and consider the set $S$ to be the times $0\leq s\leq 1/n^{1/3}$ for which $f(s) = sup_{0 \leq t \leq \delta}\|s_{\gamma'(t,s)}\|_\infty < 2c$. Since $f(s)$ is continuous, the set $S$ is open. Hence, if we consider the infimum $s_0$ of times $s$ for which $f(s) \geq 1$, then the infimum is attained, i.e. $f(s_0) = 2c$, while $f(s) < 2c$ for every time $s < s_0$. 
Exactly the same way we can define the first time $s_1$ for which defining the function $f_2(s) = \sup_{0\leq t\leq \delta}\|\gamma'_s(t)\|_g$ we have $f_2(s_1) = 2\sqrt n$ while $f_2(s) < 2\sqrt n$ for $s < s_1$.
\\
First assume the case where $s_0 \leq s_1$.
Now again from the continuity of $f$ and the fact that $[0,\delta]$ is a compact set, its supremum is attained in some time $t_0$. This means
\begin{align}
    &\|s_{\gamma'_s(t_0)}\|_\infty < 2c,\\
    &\|\gamma'_s(t_0)\|_g < 2\sqrt n,\label{eq:inductivenormbounds}
\end{align}
for all $s < s_0$, while $\|s_{\gamma'_{s_0}(t_0)}\|_\infty = 2c$. 
But now using this infinity norm bound for times $s \leq s_0$ (for the fixed time $t_0$), we can obtain an Frobenius norm bound for $\Phi(t_0)(s)$ from Lemma in~\ref{lem:R1bound} as
\begin{align*}
    \|\Phi(t_0)(s)\|_F \lesssim R_1 = (c^2 + \alpha\sqrt{\alpha_0})\sqrt n.
\end{align*}
. Now we can apply Lemma 23 in~\cite{lee2018convergence} because condition $\delta^2 R_1 \lesssim 1$ is satisfied, so we get
\begin{align*}
    \|\nabla_{c'(s)}\gamma'(t_0)\|_g \leq 1/\delta,
\end{align*}
for every $s < s_0$, where we are using the fact that $\|c'(s)\|_g = 1$. 
But note that for $s < s_0$ we can write
\begin{align*}
    \|\frac{d}{ds}(A\gamma')\|_\infty & \leq \|A\nabla_{c'}\gamma' - Ag^{-1}\D g(c')\gamma'\|_\infty + \|s_{c'}\|_\infty \|s_{\gamma'}\|_\infty\\
   & \leq \|\nabla_{c'}\gamma'\|_g + \|g^{-1}Dg(c')\gamma'\|_g + \|s_{c'}\|_\infty \|s_{\gamma'}\|_\infty\\
   & = \|\nabla_{c'}\gamma'\|_g + \|g^{-1}\D g(\gamma')c'\|_g + \|s_{c'}\|_\infty \|s_{\gamma'}\|_\infty\\
   & \leq \|\nabla_{c'}\gamma'\|_g + \|s_\gamma'\|_\infty\|c'\|_g + \|c'\|_g \|s_{\gamma'}\|_\infty\\
   & \leq 1/\delta + c + c \lesssim \frac{1}{\delta} + c,
\end{align*}
where the first line follows from opening the definition of covariant derivative.
Finally, this ODE implies that $\|s_{\gamma'_s(t_0)}\|_\infty \lesssim s  (\frac{1}{\delta} + c) < 2c$ for all times $s < s_0$ (with the correrct choice of constants), which from continuity holds also for time $s_0$. But this contradicts $|s_{\gamma'_{s_0}(t_0)}\|_\infty = 2c$, which completes the proof for the case $s_0 \leq s_1$. Note that we the use of this condition in the above proof is that the $g$-norm condition does not fail until time $s_0$. 
\\
Next, we consider the latter case $s_1 < s_0$. Similar to the above argument, until time $s\leq s_1$ we have the Frobenius bound on $\Phi(t)$ from Lemma~\ref{lem:R1bound}, and again from Lemma 23 in~\cite{lee2018convergence} as $\delta^2 \leq \frac{1}{\sqrt n (c^2+\alpha \sqrt \alpha_0)} = \frac{1}{R_1}$, we have 
$$\|\nabla_{c'(s)} \gamma'(t_0)\| \leq 1/\delta,$$
for $s \leq s_1$.
Now we write an ODE to control the norm of $\|\gamma'_{s_1}(t_0)\|_g$ where $t_0$ is defined in the same way as the previous case, and get a contradiction:
\begin{align*}
    \frac{d}{ds}\|\gamma'\|_g^2 = 2\langle \nabla_{\frac{d}{ds}\gamma(t_0,s)} \gamma', \gamma'\rangle \leq 2\|\gamma'\|\|\nabla_{\frac{d}{ds}\gamma(t_0,s)} \gamma'\| \leq \frac{2}{\delta}\|\gamma'\|_g,
\end{align*}
which implies
\begin{align*}
    \frac{d}{ds}\|\gamma'\|_g \leq \frac{2}{\delta}.
\end{align*}
Therefore, at time $s = \delta/4$ the change in $\|\gamma'\|_g$ from its initial value is at most $1/2 < \sqrt n/2$, which means the value of $\|\gamma'\|_g$ should have remained below $2\sqrt n$. The contradiction completes the proof for the second case.
\end{proof}

Next, we show a helper lemma regarding the derivative of $\gamma'_s(t)$ in direction $\frac{d}{ds}\gamma_s(t)$:
\begin{lemma}
On a $(c, \delta)$-nice Hamiltonian curve with $\delta \leq \frac{1}{n^{1/4}c}$, We have:
\begin{align*}
    \|\frac{d}{ds}\gamma'(t)(s)\| \leq 1/\delta. 
\end{align*}
\end{lemma}
\begin{proof}
Note that from Lemma~\ref{lem:stabilityinfnorm} we have $\|s_{\gamma'_s(t_0)}\|_\infty \leq c$. Hence, from Lemma~\ref{lem:R1bound}, we can apply Lemma 23 in~\cite{lee2018convergence} to obtain
\begin{align}
    \|\nabla_{\frac{d}{ds}\gamma}\gamma'(t)\|_g \lesssim \frac{1}{\delta}.\label{eq:tmp3}
\end{align}
But now from Lemma~\ref{lem:covarianttonormal}, setting $v = \gamma'(t,s)$ and $z = \frac{d}{ds}\gamma(t,s)$:
\begin{align*}
    \|\frac{d}{ds}\gamma'(t,s)\|_{g} \leq \|\gamma'\|_\infty\|\frac{d}{ds}\gamma\|_{g} + \|\nabla_{\frac{d}{ds}\gamma}\gamma'(t)\|_{g}.
\end{align*}
From Lemma~\ref{lem:stabilityinfnorm}, we have $\|\gamma'\|_\infty \leq c$ and note that from our assumption on the $s$ parameterization, $\|\frac{d}{ds}\gamma\|_g = 1$, which combined with Equation~\eqref{eq:tmp3} finishes the proof.  
\end{proof}

\subsection{High probability bound on norms along the Hamiltonian curve}
First, we show a norm bound for the $g$ norm along the Hamiltonian curve, given a bound at initial time.

Recall the ODE related to the RHMC for curve $\gamma$ is
\begin{align*}
    D_t^2 \gamma(t) = \mu(\gamma).
\end{align*}
Opening this up
\begin{align}
    \gamma''(t) + \frac{1}{2}g^{-1}Dg(\gamma')\gamma' = \mu.\label{eq:hmcode}
\end{align}
First, we show a non-random bound on the norm $\|\gamma'\|_g$ given a bound at time zero.

\begin{lemma}[Boundedness of manifold norm along the Hamiltonian curve]\label{lem:gammaprimenorm}
Suppose $\|\gamma'(0)\|_g\leq \sqrt n$. 
Then for time $t \leq 1$ we have
\begin{align*}
    \|\gamma'(t)\|_g \leq \sqrt n.
\end{align*}
\end{lemma}
\begin{proof}
Note that
\begin{align*}
    \|g^{-1}Dg(\gamma')\gamma'\|_g \leq \|s_{\gamma'}\|_\infty \|\gamma'\|_g,
\end{align*}
hence, taking covariant derivative
\begin{align*}
    \frac{d}{dt}\|\gamma'(t)\|_g^2 &= 2\langle \nabla_{\gamma'}\gamma', \gamma'\rangle\\
    & \leq \|\gamma'\|\|\mu\|\leq (1+\alpha\sqrt{\alpha_0})\sqrt n \|\gamma'\|,
\end{align*}
where we used Lemma~\ref{lem:biasnorm} to bound $\|\mu\|$. This implies
\begin{align*}
    \frac{d}{dt}\|\gamma'(t)\|\leq (1+\alpha\sqrt{\alpha_0})\sqrt n .
\end{align*}
Solving this ODE, 
\begin{align}
    \|\gamma'(t)\|\leq (1+t)(1+\alpha\sqrt{\alpha_0})\sqrt n.\label{eq:gnormstability}
\end{align}
\end{proof}

\begin{lemma}[Stability bound on the infinity norm along the curve]\label{lem:infnormstabilityalongt}
For a hamiltonian curve with $\|\gamma'(0)\|_g \leq \sqrt n$, suppose for a fixed time $t_1$ we know $\|s_{\gamma'(t_1)}\|_\infty \lesssim c$. Then for all times $t \in (t_1 - \frac{1}{(1 + \alpha \sqrt \alpha_0)\sqrt n} , t_1 + \frac{1}{(1 + \alpha \sqrt{\alpha_0})\sqrt n})$ we have
\begin{align*}
    \|s_{\gamma'}(t)\|_\infty \lesssim c.
\end{align*}
\end{lemma}
\begin{proof}
Consider the Hamiltonian ODE below:
\begin{align*}
    \gamma''(t) + \frac{1}{2}g^{-1}\D g(\gamma')\gamma' = \mu.
\end{align*}
which implies
\begin{align*}
    \frac{d}{dt}(A\gamma') = -\frac{1}{2} Ag^{-1}\D g(\gamma')\gamma' + A\mu - (S_{\gamma'})^{\odot 2}.
\end{align*}
Hence, using Lemma~\ref{lem:infwithgnorm}
\begin{align*}
    \|\frac{d}{dt}(A\gamma')\|_\infty &\leq
    \|g^{-1}\D g(\gamma')\gamma'\|_g + \|s_{\gamma'}\|_\infty^2 + \|A\mu\|_\infty\\
    & \leq \|s_{\gamma'}\|_\infty \|s_{\gamma'}\|_g +  \|s_{\gamma'}\|_\infty^2 + \|\mu\|_g.
\end{align*}
But using Lemma~\ref{lem:gammaprimenorm} having upper bound on the $g$-norm of $\gamma'$ at time zero implies a bound on the whole curve. Combining with Lemma~\ref{lem:biasnorm}:
\begin{align*}
    \|\frac{d}{dt}(A\gamma')\|_\infty \leq \sqrt n \|s_{\gamma'}\|_\infty + \|s_{\gamma'}\|_\infty^2 + (1+\alpha \sqrt{\alpha_0})\sqrt n.
\end{align*}
This ODE implies that if at a given point the infinity norm of $\|s_{\gamma'}\|_\infty$ is bounded by $c$, then for times within $t \pm \frac{1}{c(1 + \alpha \sqrt{\alpha_0}) \sqrt n}$ we have an $O(c)$ bound on the infinity norm, which completes the proof. 
\end{proof}

\begin{lemma}[Stability bound on the  $g$-norm along the curve]\label{lem:gnormstabilityalongt}
For a Hamiltonian curve with $\|\gamma'(0)\|_g \leq \sqrt n$, suppose for a fixed time $t_1$ we know $\|\gamma'(t_1)\|_g \lesssim c$. Then for all times $t \in (t_1 - \frac{1}{(1 + \alpha \sqrt \alpha_0)\sqrt n} , t_1 + \frac{1}{(1 + \alpha \sqrt{\alpha_0})\sqrt n})$ we have
\begin{align*}
    \|\gamma'(t)\|_g \lesssim c.
\end{align*}
\end{lemma}
\begin{proof}
    Directly from Lemma~\ref{lem:gammaprimenorm}.
\end{proof}

\begin{lemma}\label{lem:infbounduniform}
Suppose we pick $x$ random from $e^{-\alpha \phi(x)}$ then run a Hamiltonian curve starting from $x$ with initial vector $\gamma'(0)$ picked according to $\mathcal N(0, g^{-1})$. Then, for any time $t_1 \in (0,1)$, with probability at least $1 - poly(m)ce^{-\Theta(c^2)}$ we have 
\begin{align*}
    &\|s_{\gamma'(t_1)}\|_\infty \lesssim c,\\
    &\|\gamma'(t_1)\|_g \lesssim c\sqrt n.
\end{align*}
\end{lemma}
\begin{proof}
From the property of the Hamiltonian curve, we know the joint density of $(\gamma(t), \gamma'(t))$ is $e^{-\alpha \phi(x)} \times \mathcal N(0, g^{-1}(x)) dxdv$. Focusing on the probability of $v_t = \gamma'(t)$, we see that for each $i$, $a_i^T v_t$ is a Gaussian distributed variable with variance 
\begin{align*}
    a_i^\top g^{-1} a_i \leq 1,
\end{align*}
where the inequality follows from Lemma~\ref{ginversebound}. Hence, from Gaussian tail bound, for a fixed time $t$:
\begin{align*}
    \mathbb P(\|s_{v_t}\|_\infty \geq c) \lesssim me^{-c^2/2},
\end{align*}
where note that $\|s_{v_t}\|_\infty$ is just the maximum of Gaussian random variables
and we applied a union bound over the entries of $s_{v_t}$. 
Moreover, note that $\|v_t\|_g$ is a subGaussian random variable with mean $O(\sqrt n)$ and subGaussian parameter $O(1)$. Hence
\begin{align*}
    \mathbb P(\|v_t\|_g \geq c) \lesssim e^{-\Theta(c^2)}.
\end{align*}

Next, consider a cover $\mathcal C = \{t_i\}_{i=1}^{c(1+\alpha\sqrt{\alpha_0})\sqrt n}$ of equally distant times of the Hamiltonian curve from $t = 0$ to $t = 1$. Apply the above argument for all the times in this cover with a union bound on top. This implies with probability at least $1 - \text{poly}(m)ce^{-\Theta(c^2)}$, we have $\|s_{v_t}\|_\infty \lesssim c$ for all $t \in \mathcal C$ and $\|v_t\|_g \lesssim c\sqrt n$, where we used the fact that $\alpha\sqrt{\alpha_0} = poly(m)$. Now combining this with Lemmas~\ref{lem:infnormstabilityalongt} and~\ref{lem:gnormstabilityalongt} completes the proof. 
\end{proof}

Next, we bring a Lemma which shows the existence of Nice sets, used in the Proof of Theorem~\ref{thm:mixing}.
\begin{lemma}\label{lem:highprobset}[Existence of Nice set]
There is a high probability region $ S \subset \mathcal M$ such that $\pi(S) \geq 1 - O(poly(m)e^{-c/2})$ (where recall $\pi(.)$ is the probability distribution of density $e^{-\phi}$ inside the polytope) and for every $x \in S$, there is a high probability region $Q_x$ in the tangent space of $x$, namely $\mathbb P(v_x \in Q_x) \geq 0.999$ 
such that for all $v_x \in Q_x$, the Hamiltonian curve starting from $x$ with initial vector $v_x$ is $(c, 1)$-nice, namely for all $0 \leq t \leq 1$:
\begin{align*}
    &\|\gamma'(t)\|_g \lesssim c\sqrt n,\\
    &\|s_{\gamma'(t)}\|_\infty \lesssim c.
\end{align*}
\end{lemma}
\begin{proof}
    For every point $x\in \mathcal M$, define $Q_x$ to be the set of vectors in its tangent space such that the resulting curve is $c$-nice up to time $1$. Define region $S$ to be the the set of points $x$ on $\mathcal M$ such that $p_{v_x}(Q_x) \geq 1 - 0.0005$, where $p_{v_x}$ denotes the density of $\mathcal N(0, g^{-1})$ in the tangent space of $x$ (The constant $1 - 0.0005$ is motivated by the definition of nice sets). Now if it was the case that $\mathbb P(S^c) \geq \text{poly}(m)ce^{-\Theta(c^2)}$, then under the joint distribution on $(x, v)$, there is a region with probability at least $poly(m)ce^{-\Theta(c^2)}$ such that the Hamiltonian curve starting from $x$ with initial vector $v$ is not $c$-nice. But this contradicts Lemma~\ref{lem:infbounduniform}. 
\end{proof}

\section{Isoperimetry}\label{sec:isoperimetry}
In this section, we show an the isoperimetry constant corresponding to our barrier, stated in Theorem~\ref{thm:hybrid-iso}.

\begin{proof}[Proof of Theorem~\ref{thm:hybrid-iso}.]
From Lemma~\ref{lem:ellipsebysup} and the definition of $g$:
\begin{align*}
    \|v\|^2_g \leq  p\|v\|_{g''} (\frac{m}{n})^{\frac{2/p}{2/p + 1}}  \leq pn (\frac{m}{n})^{\frac{2/p}{2/p + 1}}\|s_{x,v}\|^2_\infty.
\end{align*}
This means that if we scale the ellipsoid $\{v| \ v^\top g v \leq 1\}$ by $\sqrt{pn} (\frac{m}{n})^{\frac{1/p}{2/p + 1}}$ then it includes the symmetrized polytope around $x$, whose unit ball is exactly $\{v|\ \|s_{x,v}\|_\infty \leq 1\}$, i.e.
\begin{align}
    \{v|\ \|s_{x,v}\|_\infty \leq 1\} \subseteq \{\sqrt{pn} (\frac{m}{n})^{\frac{1/p}{2/p + 1}} v| \ v^\top g v \leq 1\}.\label{eq:avali}
\end{align}
On the other hand, from Lemma~\ref{lem:infwithgnorm} we have
\begin{align*}
    \|s_{x,v}\|_\infty \leq \|v\|_g,
\end{align*}
which implies that the unit ball of the norm, or the Dikin ellipsoid, is contained in the symmetrized poltope around $x$, i.e.
\begin{align}
    \{v| \ v^\top g v \leq 1\} \subseteq \{v|\ \|s_{x,v}\|_\infty \leq 1\}.\label{eq:dovomi}
\end{align}
Combining the relations~\eqref{eq:avali} and~\eqref{eq:dovomi} implies that the symmetric self-concordance parameter $\bar \nu$ defined in~\cite{laddha2020convergence} is at most $\bar \nu \leq pn(\frac{m}{n})^{\frac{2/p}{2/p+1}}$, which in turn implies that the distribution $e^{-\alpha \phi}$ has isoperimetry with constant at least $\frac{1}{\sqrt \nu} \geq \frac{1}{\sqrt{pn}}(\frac{n}{m})^{\frac{1/p}{2/p+1}}$ with respect to metric $g$ as desired.

Furthermore, using the Brascamp-Lieb inequality, we know $e^{-\alpha\phi}$ has isoperimetry at least $\sqrt{\alpha}$ on a manifold whose metric is the Hessian of $\phi$~\cite{bakry2014analysis}. Combining these two facts completes the proof.
\end{proof}


We denote the $i$th row of the matrix $\A$ by $a_i$. Note that if we have a bound on the quantity $a_i^\top g^{-1} a_i$ for our metric $g$ enables us to control the infinity norm of $s_{x,v}$ via the following simple Cauchy Schwarz on the $i$th entry of $s_{x,v}$:
\begin{align*}
    |{s_{x,v}}_i| = |a_i^\top v| \leq \|v\|_g \sqrt{a_i^\top g^{-1}a_i}.
\end{align*}
However, while we have the following relation 
\begin{align}
a_i^\top (\A^\top\W^{1-2/p})\A)^{-1}a_i = {{w_x}_i}^{2/p} \leq 1,\label{eq:counterpart}
\end{align}
only considering the $g_2$ subpart of our metric $g$, the quantity $a_i^T {g_2}^{-1} a_i$ might be orders of magnitude larger than its counterpart $a_i^\top (\A\W^{1-2/p})\A)^{-1}a_i$ in Equation~\eqref{eq:counterpart}. This is because recall as we state in~\ref{lem:lsbarrierapprox}
\begin{align*}
   \A^T \W \A \preccurlyeq g_2 \preccurlyeq p\A^T \W \A,
\end{align*}
but we do not have such spectral bounds between matrices $\A^T W\A$ and $\A^T W^{1-2/p}\A$. In fact, authors in~\cite{lee2019solving} show $\A^T W\A$ and $\A^T W^{1-2/p}\A$ are up to log factors spectrally the same, as long as $p$ is polylogarithmically large, but here we are not able to work with such large $p$'s since our infinity norm estimates break for $p \geq 4$. Nonetheless, we show that adding the log barrier and appropriately rescaling the metric $g$ indeed enables us to bound $a_i^T g^{-1}a_i$. To prove a bound on $a_i^T g^{-1}a_i$, we start by comparing the matrix $g' \triangleq \A^T \W \A + \frac{n}{m}\A^\top \A$, which is proportional to the Hessian of the hybrid barrier before scaling by $\alpha_0$, with the matrix $\A^\top \W^{1-2/p}\A$, which then enables us to analyze the quantity $a_i^\top g^{-1} a_i$ via the closed form Equation~\eqref{eq:counterpart}. In the next Lemma, we compare these two matrices.
\begin{lemma}[L\"{o}wner comparison with different weighted matrices]\label{lem:matrixcomparison}
For the PSD matrix $g' = \A^\top\W\A + \frac{n}{m}\A^\top \A$ we have
\begin{align*}
    \A^\top \W^{1-2/p}\A \preccurlyeq (\frac{m}{n})^{2/p}g'.
\end{align*}
\end{lemma}
\begin{proof}
Suppose for a given coefficient $\beta$ we wish to have
\begin{align}
    {{w_x}_i}^{1-2/p} \leq \beta({{w_x}_i} + \frac{n}{m}).\label{eq:initialwbound}
\end{align}
The first thing we notice is that if $w_i^{1-2/p} \leq \beta \frac{n}{m}$, then the inequality is already satisfied. Hence, w.l.o.g we assume
\begin{align}
    {{w_x}_i} \geq (\beta \frac{n}{m})^{\frac{1}{1-2/p}}.\label{eq:initialwinequality}
\end{align}
in this regime of $w_i$ to pick a $\beta$ which satisfies Equation~\eqref{eq:initialwbound}, we need to have
\begin{align*}
    \beta {{w_x}_i}^{2/p} \geq 1.
\end{align*}
But using Equation~\eqref{eq:initialwinequality}, it is sufficient to have
\begin{align*}
    (\beta \frac{n}{m})^{\frac{2/p}{1-2/p}} \beta \geq 1,
\end{align*}
so we need to pick $\beta$ as large as
\begin{align*}
    \beta = (\frac{m}{n})^{2/p},
\end{align*}
which completes the proof. 
\end{proof}

\begin{lemma}[Taming the hybrid metric]\label{ginversebound}
For the metric of our hybrid barrier before scaling up by $\alpha_0$, i.e. for $g''$ defined as
\begin{align}
g''(x) \triangleq \nabla^2 \phi_p(x) + \frac{n}{m}\nabla^2 \phi_{\ell} = g_1 + \frac{n}{m}\A^\top\A,\label{eq:g''def}
\end{align}
 we have for every $i$:
\begin{align*}
    a_i^\top g''^{-1} a_i \leq (\frac{m}{ n})^{\frac{2/p}{1+2/p}}.
\end{align*}
In particular, for the metric $g(x)$ of the hybrid barrier we have
\begin{align*}
    a_i^\top g^{-1} a_i \leq 1.
\end{align*}
\end{lemma}
\begin{proof}
Note that using Lemma~\ref{lem:lsbarrierapprox}, we have
\begin{align*}
    g' = \A^\top \W \A+ \frac{n}{m}\A^\top\A \preccurlyeq g_1+ \frac{n}{m}\A^\top\A = g''.
\end{align*}
Hence, using Lemma~\ref{lem:matrixcomparison}:
\begin{align}
   a_i^\top g''^{-1} a_i \leq a_i^\top g'^{-1} a_i\leq (\frac{m}{ n})^{2/p} a_i^\top (\A^\top \W^{1-2/p}A)^{-1}a_i \leq (\frac{m}{n})^{2/p} {{w_x}_i}^{2/p}.\label{eq:aval}
\end{align}
On the other hand, 
\begin{align}
   a_i^\top g''^{-1} a_i \leq a_i^\top g'^{-1} a_i \leq a_i^\top (\A^\top \W A)^{-1} a_i = 
  {{w_x}_i}^{-1} {{w_x}_i}^{1/2}a_i^\top (A^\top \W \A)^{-1} a_i {{w_x}_i}^{1/2} \leq {w_i}^{-1}.\label{eq:dovom}
\end{align}
Balancing Equations~\eqref{eq:aval} and~\eqref{eq:dovom} implies
\begin{align}
    a_i^\top g''^{-1} a_i \leq (\frac{m}{ n})^{\frac{2/p}{1+ 2/p}}.\label{eq:infbyellipse}
\end{align}
Finally, noting the fact that 
\begin{align*}
    g(x) = \nabla^2 \phi(x) = (\frac{m}{n})^{\frac{2/p}{1+2/p}}\nabla^2 (\phi_p(x) + \frac{n}{m}\phi_\ell(x)) = (\frac{m}{n})^{\frac{2/p}{1+2/p}} g''(x),
\end{align*}
the proof is complete.
\end{proof}

Finally, using our estimate on $a_i^\top g''^{-1} a_i$ in Lemma~\ref{ginversebound}, we bound the $g''$ norm of an arbitrary vector $v$:
\begin{lemma}[Bounding the ellipsoid norm by the infinity norm]\label{lem:ellipsebysup}
We can bound the metric norm $g$ by the infinity norm $\|.\|_{x,\infty}$ as
\begin{align*}
    \|v\|_{g''} \leq 
    np \|s_{x,v}\|_{x,\infty}.
\end{align*}
\end{lemma}
\begin{proof}
Using Lemma~\ref{lem:lsbarrierapprox}, we have
\begin{align*}
    v^\top g_1v \leq p\sum_{i} {w_x}_i {s_v}_i^2
    \leq p\sum_{i} {w_x}_i \|s_v\|_\infty^2 = np \|s_v\|_\infty^2,
\end{align*}
and
\begin{align*}
    v^\top g_2 v = \frac{n}{m}\sum_i {s_v}_i^2 \leq n\|s_v\|_\infty^2.
\end{align*}
Noting the definition of $g''$ in Equation~\eqref{eq:g''def} completes the proof.
\end{proof}

\begin{lemma}[Bounding infinity norm by the ellipsoidal norm]\label{lem:infwithgnorm}
Given an arbitrary vector $z \in \mathbb R^n$, we have
\begin{align*}
    &\|s_{x,z}\|_\infty \leq \|z\|_g,\\
    &\|r_{x,z}\|_\infty \leq \frac{1}{4/p - 1} \|z\|_g.
\end{align*}
\end{lemma}
\begin{proof}
For all $i$ we have using Lemma~\ref{ginversebound}:
\begin{align*}
    |a_i^\top z| \leq \sqrt{a_i^\top g^{-1} a_i} \sqrt{z^\top g z} \leq 
    (\frac{n}{m})^{\frac{1/p}{1+2/p}}\sqrt{a_i^\top {g''}^{-1} a_i} \sqrt{z^\top g z}
    \leq \|z\|_g.
\end{align*}
The second inequality follows from the fact that $\|r_{x,z}\|_\infty \leq \frac{1}{4/p - 1}\|s_{x,z}\|_\infty$ from Lemma~\ref{lem:ginfnorm}.
\end{proof}

\begin{lemma}[Infinity norm of random vectors]\label{lem:infnormbound}
For the metric $g$ of our hybrid barrier, given random vector $v \sim \mathcal N(0,g^{-1})$, with high probability we have
\begin{align*}
    &\|s_{x,v}\|_{\infty} \leq 1,\\
    &\|r_{x,v}\|_\infty \leq \frac{1}{4/p - 1}.
\end{align*}
\end{lemma}
\begin{proof}
Note that $g$ is just a scaled version of $g''$:
\begin{align*}
    g = (\frac{m}{n})^{\frac{2/p}{1+2/p}} g'',
\end{align*}
Now computing the variance of the $i$th entry of $s_{x,v}$, we observe using Lemma~\ref{ginversebound}
\begin{align*}
    a_i^\top g^{-1} a_i \leq 1.
\end{align*}
The bound on $\|r_{x,v}\|_\infty$ directly follows from the fact that $\|r_{x,v}\|_\infty  \leq \frac{1}{4/p - 1} \|s_{x,v}\|_\infty$ using Lemma~\ref{lem:ginfnorm}.
\end{proof}

\newpage
\bibliographystyle{plain}
\bibliography{main}

 \newpage
\appendix
\section{Riemannian Geometry}\label{sec:riemanniangeometry}
\subsection{Basic Manifold Definitions}
In this section, we go through some basic definitions in differential geometry that are essential to know in our proofs. A manifold is defined abstractly as a topological space which locally resembles $\mathbb R^n$. 

\begin{definition}
 A manifold $\mathcal M$ is a topological space such that for each point $p \in \mathcal M$, there exists an open set $U$ around $p$ such that $U$ is a homeomorphism to an open set of  $\mathbb R^n$. 
    \paragraph{Tangent Space.}For any point $p \in \mathcal M$, one can define the notion of tangent space for $p$, $T_p(\mathcal M)$, as the equivalence class of the set of curves $\gamma$ starting from $p$ ($\gamma(0)= p$), where we define two such curves $\gamma_0$ and $\gamma_1$ to be equivalent if for any function $f$ on the manifold:
    \begin{align*}
        \frac{d}{dt}f(\gamma_0(t))\big|_{t=0} = \frac{d}{dt}f(\gamma_1(t))\big|_{t=0}.
    \end{align*}
\end{definition}
    On can define a linear structure on $T_p(\mathcal M)$, hence it is a vector space. Now given a positive definite quadratic form $g(p)$ on the vector space $T_p(\mathcal M)$, one can equip the manifold $\mathcal M$ with metric $g$. While the definition of a general manifold is abstract, putting a metric on it allows us to measure length, areas, volumes, etc. on the manifold, and do calculus similar to Euclidean space.
Next, we define some basic notions regarding manifolds.

\paragraph{Differential.}
For a map $f: {\mathcal M} \rightarrow \mathcal N$ between two manifolds, the differential $df_p$ at some point $p \in \mathcal M$ is a linear map from $T_p(\mathcal M)$ to $T_{f(p)}(\mathcal N)$ with the property that for any curve $\gamma(t)$ on $\mathcal M$ with $\gamma(0) = p$, we have
\begin{align}
    df(\frac{d}{dt}\gamma(0)) = \frac{d}{dt}f(\gamma)(0).\label{eq:prop}
\end{align}
. As a special case, for a function $f$ over the manifold, the differential $df$ at some point $p \in \mathcal M$ is a linear functional over $T_p(\mathcal M)$, i.e. an element of $T^*_p(\mathcal M)$. Writing~\eqref{eq:prop} for curve $\gamma_i$ with $\frac{d}{dt}\gamma_i(0) = \partial x_i$, testing property~\eqref{eq:prop}, we see
\begin{align*}
    df(\partial x_i) = \frac{d}{dt} f(\gamma_i(t))\Big|_{t = 0} 
    = \frac{\partial f}{\partial x_i}(\gamma_i(0)).
\end{align*}
We can write $df = \sum_i \frac{\partial f}{\partial x_i}dx_i$.

\paragraph{Vector field.} 
A vector field $V$ is a smooth choice of a vector $V(p) \in T_p(\mathcal M)$ in the tangent space for all $p \in \mathcal M$. 

\paragraph{Metric and inner product.} 
A metric is a tensor  
on the manifold $\mathcal M$ which is simply a smooth choice of a symmetric bilinear map over $\mathcal M$. Alternatively, the metric or dot product $\langle ,\rangle$ can be seen as a bilinear map over the space of vector fields with the tensorization property, i.e. for vector fields $V, W, Z$ and scalar functions $\alpha, \beta$ over $\mathcal M$:
\begin{align}
    &\langle V+W, Z\rangle = \langle V , Z\rangle + \langle W, Z\rangle,\\
    &\langle \alpha V, \beta W\rangle = \alpha \beta \langle V, W\rangle.
\end{align}

\subsection{Manifold Derivatives, Geodesics, Parallel Transport}
\subsubsection{Covariant derivative}
Given two vector fields $V$ and $W$, the covariant derivative, also called the Levi-Civita connection $\nabla_{V}W$ is a bilinear operator with the following properties:
\begin{align*}
    \nabla_{\alpha_1 V_1 + \alpha_2V_2}W = \alpha_1\nabla_{V_1}W + \alpha_2\nabla_{V_2}W,\\
    \nabla_{V}(W_1 + W_2) = \nabla_V(W_1) + \nabla_V(W_2),\\
    \nabla_V(\alpha W_1) = \alpha \nabla_V(W_1) + V(\alpha)W_1
    \end{align*}
    where $V(\alpha)$ is the action of vector field $V$ on scalar function $\alpha$.
Importantly, the property that differentiates the covariant derivative from other kinds of derivaties over manifold is that the covariant derivative of the metric is zero, i.e., 
$\nabla_V g = 0$ for any vector field $V$. In other words, we have the following intuitive rule:
\begin{align*}
    \nabla_V \langle W_1, W_2 \rangle 
    = \langle \nabla_V W_1, W_2\rangle + \langle  W_1, \nabla_V W_2\rangle.
\end{align*}
Moreover, the covariant derivative has the property of being torsion free, meaning that for vector fields $W_1, W_2$:
\begin{align*}
    \nabla_{W_1}W_2 - \nabla_{W_2}W_1 = [W_1, W_2],
\end{align*}
where $[W_1, W_2]$ is the Lie bracket of $W_1, W_2$ defined as the unique vector field that satisfies
\begin{align*}
    [W_1, W_2]f = W_1(W_2(f)) - W_2(W_1(f))
\end{align*}
for every smooth function $f$.

In a local chart with variable $x$, if one represent $V = \sum V^i \partial x_i$, where $\partial x_i$ are the basis vector fields, and $W = \sum W^i \partial x_i$, the covariant derivative is given by
\begin{align*}
    \nabla_V W & = \sum_i V^i\nabla_i W =
    \sum_i V^i \sum_j \nabla_i (W^j \partial x_j) \\
    & = \sum_i V^i \sum_j \partial_i(W^j) \partial x_j) + \sum_i V^i \sum_j W^j \nabla_i \partial x_j)\\
    & =  \sum_j V(W^j) \partial x_j) + \sum_i \sum_j V^i W^j \sum_k \Gamma_{ij}^k \partial x_k)
    = \\
    & = \sum_k \big(V(W^k) + \sum_i \sum_j V^i W^j \Gamma_{ij}^k \big)\partial x_k.\label{eq:covariantexpansion}
\end{align*}
The Christoffel symbols $\Gamma_{ij}^k$ are the representations of the Levi-Cevita derivatives of the basis $\{\partial x_i\}$:
\begin{align*}
    \nabla_{\partial x_j} \partial x_i = \sum_{k} \Gamma_{ij}^k \partial x_k
\end{align*}
and are given by the following formula:
\begin{align*}
    \Gamma_{ij}^k = \frac{1}{2}\sum_{m} g^{km}(\partial_j g_{mi} + \partial_i g_{mj} - \partial_m g_{ij}).
\end{align*}
Above, $g^{ij}$ refers to the $(i,j)$ entry of the inverse of the metric. 
In the following Lemma, we calculate the Christoffel symbols on a Hessian manifold and $g = D^2\phi$ is the Hessian of a convex function.
\begin{lemma}\label{lem:cristoffelsymbols}
    On a Hessian manifold with metric $g$ we have
    \begin{align*}
        \Gamma_{ij}^k = \frac{1}{2} \sum_m g^{km}Dg_{mij}.
    \end{align*}
\end{lemma}
\begin{proof}
Since the manifold is Hessian, we have
\begin{align*}
    \partial_j g_{mi} = \partial_i g_{jm} = \partial_m g_{ij} = Dg_{ijm},
\end{align*}
where $Dg_{ijm}$ is just the notation that we use for Hessian manifolds.

\end{proof}

\subsubsection{Parallel Transport}
The notion of parallel transport of a vector $V$ along a curve $\gamma$ can be generalized from Euclidean space to a manifold. On a manifold, parallel transport is a vector field restricted to $\gamma$ such that $\nabla_{\gamma'}(V) = 0$. By this definition, for two parallel transport vector fields $V(t), W(t)$ we have that their dot product $\langle V(t), W(t)\rangle$ is preserved, i.e., $\frac{d}{dt} \langle V(t), W(t)\rangle = 0$. 

\subsubsection{Geodesic}
A geodesic is a curve $\gamma$ on $\mathcal M$ is a ``locally shortest path", i.e., the tangent to the curve is parallel transported along the curve: $\nabla_{\dot \gamma}\dot \gamma = 0$ ($\dot \gamma$ denotes the time derivative of the curve $\gamma$.) Writing this in a chart, one can see it is a second order nonlinear ODE 
which locally has a unique solution given initial location and speed. 

\begin{align}
    \frac{d^2\gamma_k}{dt^2}(t) = -\frac{1}{2}\sum_{i,j} \frac{d\gamma_i}{dt} \frac{d \gamma_j}{dt} \Gamma_{ij}^k, \ \ \forall k.
\end{align}

\subsubsection{Riemann Tensor}
The Riemann tensor is particular tensor on the manifold which arise from the covariant derivative. In particular, it is a linear mapping 
from $T_p(\mathcal M) \times T_p(\mathcal M) \times T_p(\mathcal M) \rightarrow T_p(\mathcal M)$ defined as
\begin{align*}
    R(X,Y)Z = \nabla_X \nabla_Y Z - \nabla_Y \nabla_X Z - \nabla_{[X,Y]}Z.
\end{align*}
The Riemann tensor can be calculated in a chart given the following formula:
\begin{align}
    R^i_{jkl} & = \frac{\partial \Gamma_{l j}^i}{\partial x_k} - \frac{\partial \Gamma_{kj}^i}{\partial x_l} + \sum_p (\Gamma_{kp}^i \Gamma_{l j}^p - \Gamma_{lp}^i \Gamma_{kj}^p).\label{eq:riemanntensordef}
\end{align}
In the following Lemma, we calculate the Riemann tensor on a Hessian manifold:
\begin{lemma}\label{lem:riemanntensor}
The Riemann tensor is given by
\begin{align*}
    R^i_{jkl}  = -\frac{1}{4} g^{i, \ell}\D g_{\ell,k,p} g^{p, \ell_2} \D g_{\ell_2,l,j} + \frac{1}{4} g^{i, \ell}\D g_{\ell,l,p} g^{p,\ell_2} \D g_{\ell_2, k,j}.
\end{align*}
\end{lemma}
\begin{proof}
    We consider the terms in Equation~\eqref{eq:riemanntensordef} one by one. For the first term
    \begin{align*}
        \frac{\partial \Gamma^i_{lj}}{\partial x_k} &= \partial_{x_k}(\frac{1}{2}\sum_{\ell_2} g^{i\ell_2}\D g_{\ell_2lj})\\
        &= \frac{1}{2} \sum_{\ell_2} \partial_{x_k}(g^{i\ell_2})\D g_{\ell_2lj} + \frac{1}{2} \sum_{\ell_2} g^{i\ell_2}\partial_{x_k}(\D g_{\ell_2 lj})\\
        &= -\frac{1}{2} \sum_{\ell_2} g^{i\ell}\D g_{\ell k p}g^{p\ell_2}\D g_{\ell_2 lj} + \frac{1}{2}\sum_{\ell_2} g^{i\ell_2}D^2 g_{k\ell_2 lj}.
    \end{align*}
    Similarly 
    \begin{align*}
        \frac{\partial \Gamma^i_{kj}}{\partial x_l} & =  -\frac{1}{2} \sum_{\ell_2} g^{i\ell}\D g_{\ell l p}g^{p\ell_2}\D g_{\ell_2 kj} + \frac{1}{2}\sum_{\ell_2} g^{i\ell_2}D^2 g_{l\ell_2 kj}.
    \end{align*}
    Hence
    \begin{align}
        \frac{\partial \Gamma^i_{lj}}{\partial x_k} - \frac{\partial \Gamma^i_{kj}}{\partial x_l} = 
        -\frac{1}{2} \sum_{\ell_2} g^{i\ell}\D g_{\ell k p}g^{p\ell_2}\D g_{\ell_2 lj} + \frac{1}{2} \sum_{\ell_2} g^{i\ell}\D g_{\ell l p}g^{p\ell_2}\D g_{\ell_2 kj}.\label{eq:komaki1} 
    \end{align}
    For the third and forth terms
    \begin{align}
     \sum_p \Gamma_{kp}^i \Gamma_{l j}^p =   \frac{1}{4} \sum_{\ell_2} g^{i\ell}\D g_{\ell k p}g^{p\ell_2}\D g_{\ell_2 lj},\\
     \sum_p \Gamma_{lp}^i \Gamma_{k j}^p =   \frac{1}{4} \sum_{\ell_2} g^{i\ell}\D g_{\ell l p}g^{p\ell_2}\D g_{\ell_2 kj}.\label{eq:komaki2}
    \end{align}
    Combining Equations~\eqref{eq:komaki1} and~\eqref{eq:komaki2} and plugging into~\eqref{eq:riemanntensordef} completes the proof.
\end{proof}

\subsubsection{Ricci tensor}
The Ricci tensor is just the trace of the Riemann tensor with respect to the second and third components or first and forth components, i.e. the trace of the operator $R(.,X)Y$:
\begin{align*}
    \Ricci(.,X)Y = \tr(R(.,X)Y).
\end{align*}
Equivalently, if $\{e_i\}$ is an orthogonal basis in the tangent space, we have
\begin{align}\label{Ricci2}
    \Ricci(X,Y) = \sum_i \langle Y, R(X,e_i)e_i\rangle.
\end{align}
\begin{lemma}[Form of the Ricci tensor on Hessian manifolds]\label{lem:riccitensor2}
On a Hessian manifold, the Ricci tensor is given by
\begin{align*}
    \Ricci(v_1,v_2) = -\frac{1}{4}\tr(g^{-1}\D g(v_1)g^{-1}\D g(v_2)) + \frac{1}{4}v_1^\top\D g(g^{-1}\tr(g^{-1}\D g))v_2.
\end{align*}
\end{lemma}
\begin{proof}
Using the form of Riemann tensor in~\eqref{eq:riemanntensordef} and the definition of Ricci tensor in~\eqref{Ricci2}
\begin{align*}
    \Ricci(\partial_j,\partial_k) &= \sum_{i=l=1}^n \big(\frac{\partial \Gamma_{l j}^i}{\partial x_k} - \frac{\partial \Gamma_{kj}^i}{\partial x_l} + \sum_p (\Gamma_{kp}^i \Gamma_{l j}^p - \Gamma_{lp}^i \Gamma_{kj}^p) \big)\\
    & = \sum_{i=l=1}^n -\frac{1}{4} g^{i, \ell}\D g_{\ell,k,p} g^{p, \ell_2} \D g_{\ell_2,l,j} + \frac{1}{4} g^{i, \ell}\D g_{\ell,l,p} g^{p,\ell_2} \D g_{\ell_2, k,j}\\
    & = -\frac{1}{4}\tr(g^{-1}\D g_k g^{-1}\D g_j) + \frac{1}{4}e_j^\top \D g(g^{-1}\tr(g^{-1}\D g))e_k.
\end{align*}
Therefore, for arbitrary vector $v_1$ and $v_2$
\begin{align*}
    \Ricci(v_1, v_2) &= \sum_{j,k} {v_1}_j {v_2}_k \big(-\frac{1}{4}Tr(g^{-1}\D g_k g^{-1}Dg_j) + \frac{1}{4}\D g(g^{-1}\tr(g^{-1}\D g)) \big)\\
    &= -\frac{1}{4}\tr(g^{-1}\D g(v_1)g^{-1}Dg(v_2)) + \frac{1}{4}v_1^\top \D g(g^{-1}\tr(g^{-1}\D g))v_2.
\end{align*}
\end{proof}

\subsubsection{Exponential Map}
The exponential $\exp_p(v)$ at point $p$ is a map from $T_p(\mathcal M)$ to $\mathcal M$, defined as the point obtained on a geodesic starting from $p$ with initial speed $v$, after time $1$. We use $\gamma_t(x)$ to denote the point after going on a geodesic starting from $x$ with initial velocity $\nabla F$, after time $t$. 

\begin{lemma}[Commuting derivatives]\label{lem:commutingderivatives}
    Given a family of curves $\gamma_s(t)$ for $s\in [0,s']$ and $t\in [0,t']$, we have
    \begin{align*}
        \D_s \partial_t \gamma_s(t) = \D_t \partial_s\gamma_s(t).
    \end{align*}
\end{lemma}
\begin{proof}
    Let $\partial_s$ and $\partial_t$ be the standard vector fields in the two dimensional $\R^2$ space $(t,s)$. Then, we know
    \begin{align*}
        \D_s \partial_t \gamma_s(t) - \D_t \partial_s \gamma_s(t) &= [\partial_s \gamma_s(t), \partial_t \gamma_s(t)]\\
        &=[\partial_t, \partial_s] = 0.
    \end{align*}
    where $[.,.]$ is the Lie bracket.
\end{proof}

\subsection{Hessian manifolds}
In this work we are working with a specific class of manifold whose metric is impoesd by the Hessian of our hybrid barrier. A nice property of Hessian manifolds is that the terms in the Riemann tensor which depends on the second derivative of the metric cancels out, and we end up just with the first derivative and the metric itself. Specifically, for a Hessian manifold recall from Lemmas~\ref{lem:cristoffelsymbols},~\ref{lem:riemanntensor}, and~\ref{lem:riccitensor} we have the following equations for Cristoffel symbols, the Riemann tensor, and the Ricci tensor:
\begin{align*}
    & \Gamma_{ij}^k = \frac{1}{2}(g^{-1}\D_kg)_{ij},\\
    & \Rr^i_{jkl}  = \frac{\partial \Gamma_{l j}^i}{\partial x_k} - \frac{\partial \Gamma_{kj}^i}{\partial x_l} + \sum_p (\Gamma_{kp}^i \Gamma_{l j}^p - \Gamma_{lp}^i \Gamma_{kj}^p)\\
    & \ \ \ \ \ \ \ = \sum_{\ell, \ell_2}-\frac{1}{4} g^{i, \ell}\D g_{\ell,k,p} g^{p, \ell_2} \D g_{\ell_2,l,j} + \frac{1}{4} g^{i, \ell}\D g_{\ell,l,p} g^{p,\ell_2} \D g_{\ell_2, k,j},
    \\ 
    & \Ricci(\partial_k,\partial_j)  = \sum_{\ell, \ell_2} -\frac{1}{4} g^{i, \ell}\D g_{\ell,k,p} g^{p, \ell_2} \D g_{\ell_2,i,j} + \frac{1}{4}  g^{i,\ell}\D g_{\ell,i,p} g^{p,\ell_2} \D g_{\ell_2, k,j}.
\end{align*}
As we mentioned, the change of the determinant of the Jacobian matrices $J^{v_{\gamma_s}}_y$ regarding the Hamiltonian family $(\gamma_s(t))$ between $x_0$ and $x_1$ is related to the rate of change of the Ricci tensor on the manifold. In Lemma~\ref{lem:riccitensor} below, we concretely calculate the Ricci tensor for a Hessian manifold in the Euclidean chart, based on the metric $g$ and its derivatives.

\begin{lemma}[Form of Ricci tensor on Hessian manifolds]\label{lem:riccitensor}
On a Hessian manifold, the Ricci tensor is given by
\begin{align}
    \Ricci(v_1,v_2) = -\frac{1}{4}tr(g^{-1}\D g(v_1)g^{-1}\D g(v_2)) + \frac{1}{4}v_1^T\D g(g^{-1}tr(g^{-1}\D g))v_2.\label{eq:riccitensor}
\end{align}
\end{lemma}

we use the formula of Ricci tensor on manifold in section~\ref{sec:R2} and bound its derivative to bound the rate of change of the pushforward density of RHMC going from $x_0$ to $x_1$ in section~\ref{sec:riccichangebound}. Note that we only need to have a multiplicative control over the change of density of a sampled Gaussian vector on the destination point on the manifold, as we move from $x_0$ to $x_1$. 

\newpage

\section{Hamiltonian Curves and Fields on Manifold}\label{sec:hmc}
Here we recall the formulation of the Hamiltonian curve based on covariant differentiation. Starting from the definition of the hamiltonian ODE for the potential $H(x,v) = f(x) + \frac{1}{2}\log((2\pi)^n \det g(x)) + \frac{1}{2}v^T g(x)^{-1} v$.
\begin{align*}
    &\frac{dx}{dt} = g(x)^{-1} v,\\
    &\frac{dv}{dt} = -\nabla f(x) -\frac{1}{2}\tr[g(x)^{-1}\D g(x)] + \frac{1}{2}\frac{dx}{dt}^\top \D g(x) \frac{dx}{dt}.
\end{align*}
Taking derivative with respect to $t$ from the first Equation and then using the second equation, we get
\begin{align*}
    \frac{d^2x}{dt^2} & = -g(x)^{-1}\D g(x)[\frac{dx}{dt}]g(x)^{-1}v + g(x)^{-1}\frac{dv}{dt}\\
    &= -g(x)^{-1}\D g(x)[\frac{dx}{dt}]\frac{dx}{dt} - g(x)^{-1}\nabla f(x) - \frac{1}{2}g^{-1}\tr[g^{-1}\D g(x)] + \frac{1}{2}g(x)^{-1}{\frac{dx}{dt}}^\top \D g(x) \frac{dx}{dt},\\
    &= -\frac{1}{2}g(x)^{-1}\D g(x)[\frac{dx}{dt}]\frac{dx}{dt} - g(x)^{-1}\nabla f(x) - \frac{1}{2}g^{-1}\tr[g^{-1}\D g(x)],
\end{align*}
which implies
\begin{align}
    \frac{d^2x}{dt^2} + \frac{1}{2}g(x)^{-1}\D g(x)[\frac{dx}{dt}]\frac{dx}{dt} = - g(x)^{-1}\nabla f(x) - \frac{1}{2}g^{-1}\tr[g^{-1}\D g(x)].\label{eq:cristoffelcomesout}
\end{align}
But the left hand side of Equation~\eqref{eq:cristoffelcomesout} is the definition of Christoffel symbols as in Lemma~\eqref{lem:cristoffelsymbols}. To see this, note that
\begin{align*}
    \left( g(x)^{-1}\D g(x)[\frac{dx}{dt}]\frac{dx} {dt} \right)_k = 
    \sum_{i,j} (\sum_m g^{km}\D g_{mij}) \frac{dx_i}{dt} \frac{dx_j}{dt},
\end{align*}
where $\left( g(x)^{-1}\D g(x)[\frac{dx}{dt}]\frac{dx} {dt} \right)_k$ is the $k$th entry of $g(x)^{-1}\D g(x)[\frac{dx}{dt}]\frac{dx} {dt}$.
Moreover
\begin{align*}
    (\frac{d^2 x}{dt^2})_k &= \frac{d}{dt}(\frac{dx_k}{dt}) = (\sum_{i=1}^n\frac{dx_i}{dt}\partial x_i)(\frac{dx_k}{dt}).
\end{align*}

Hence, from the definition of Cristoffel symbols and its expansion in Equation~\eqref{eq:covariantexpansion} we see

\begin{align*}
    \D_t \frac{dx}{dt} = -g(x)^{-1}\nabla f - \frac{1}{2}g(x)^{-1}\tr[g(x)^{-1}\D g(x)],
\end{align*}
where $D_t\frac{dx}{dt} = \nabla_{\frac{dx}{dt}} \frac{dx}{dt}$ is covariant differentiation and we look at $\frac{dx}{dt} = \sum_{i=1}^n \frac{dx_i}{dt}\partial x_i$ as a vector in the tangent space of $x$.
We define the right hand side of the above equation as the bias of Hamiltonian Monte Carlo:
\begin{align*}
    \mu(x) \triangleq -g(x)^{-1}\nabla f - \frac{1}{2}g(x)^{-1}\tr[g(x)^{-1}\D g(x)].
\end{align*}

\begin{proof}[Proof of Lemma~\ref{lem:hmcfieldODE}]
We start from the ODE of HMC:
\begin{align*}
    \gamma_s''(t) = \mu(\gamma_s(t)).
\end{align*}
Taking covariant derivative in direction $s$:
\begin{align*}
    \D_s \mu(\gamma_s(t))= \D_s \gamma_s''(t) &= \D_s \D_t \gamma_s''(t)\\
    &= \nabla_{\partial_s\gamma_s(t)}\nabla_{\partial_t\gamma_s(t)} \gamma_s'(t).
\end{align*}
Now we apply the definition of Riemann tensor. Namely for arbitrary vector fields $X, Y, Z$, we have
\begin{align*}
\nabla_X\nabla_Y Z - \nabla_Y \nabla_X Z = R(X,Y)Z + \nabla_{[X,Y]}Z. 
\end{align*}
Setting $X = \partial_s\gamma_s(t)$ and $Y = \partial_t\gamma_s(t)$, we first observe that $[\partial_s\gamma_s(t), \partial_t\gamma_s(t)]$ because they are just the application of the differential of $\gamma$ to the standard vectors $\partial_s$ and $\partial_t$ in $\R^2$. Applying this above
\begin{align}
    \D_s \mu(\gamma_s(t)) = \nabla_{\partial_t\gamma_s(t)} \nabla_{\partial_s\gamma_s(t)} \gamma_s'(t) - R(\partial_s \gamma_s(t), \partial_t \gamma_s(t)) \gamma_s'(t).\label{eq:tmp0}
\end{align}
But note that because $\partial_t \gamma_s(t)$ and $\partial_s \gamma_s(t)$ are the image of the differential of $\gamma_s(t)$ applied to $\partial_t$ and $\partial_t$, we have
\begin{align}
    \nabla_{\partial_s\gamma_s(t)} \gamma_s'(t) =   \nabla_{\partial_t\gamma_s(t)} \partial_s\gamma_s(t) = J'(t).\label{eq:tmp1}
\end{align}
Applying Equation~\eqref{eq:tmp1} to Equation~\eqref{eq:tmp0}:
\begin{align*}
     \D_s \mu(\gamma_s(t)) = J''(t) - R(J(t), \gamma'_s(t))\gamma'_s(t).
\end{align*}
Noting the definition of the operator $M$ completes the proof.
\end{proof}

\newpage
\section{Third order strong self-concordance of the metric}\label{sec:thirdorderselfconcordance}
The goal of this section is to prove the following lemma.
\begin{lemma}[Infinity norm Self-concordance for Lewis-p-weight barrier]\label{lem:thirdorderselfconcordance}
The Lewis-p-weights barrier, defined in~\eqref{eq:lewisbarrier}, is third-order strongly self-concordant with respect to the local norm $\|.\|_{x,\infty}$, i.e., at any point $x$ on the Hessian manifold with metric $g(x)$ given by the Hessian $\nabla^2 \phi_1$ of the Lewis-p-weights barrier $\phi_1$, we have
\begin{align*}
    &-\|v\|_{x,\infty}\|z\|_{x,\infty}\|u\|_{x,\infty} g \preccurlyeq \D^3g(v,z,u) \preccurlyeq \|v\|_{x,\infty}\|z\|_{x,\infty}\|u\|_{x,\infty} g.
\end{align*}
\end{lemma}

 Now we first handle the derivatives in directions $z$ and $u$ of the $(\trr 4)$ term in Lemma~\ref{lem:gderivative}. We state the final result regarding the $(\star 4)$ term in the Lemma~\ref{lem:twoderivativesstar4}, which we prove below. 

\begin{proof}[Proof of Lemma~\ref{lem:twoderivativesstar4}]
The general style of the proof below is that $(\trrr)$ terms are referring to the subterms obtained from differentiating the $(\trr 4)$ term by $z$, which are stated in Lemma~\ref{lem:star4}. Note that the $(\trr 4)$ term itself is a subterm of the derivative of $g_1$ in direction $v$ which is stated in Lemma~\ref{lem:gderivative}.
\subsection*{$(\trr 4)$ terms}
The first subterm of the $(\trr 4)$ term that we consider is the $(\trrr 3)$ term as defined in Equation~\ref{lem:star4}.
\subsubsection*{$(\trrr 3)$ term}
\begin{align*}
   & \D(\A^\top \Pxz \odot \Pxv \G^{-1}\Lambdax \A)(u)\\
   & = \A^\top \Sxu \Pxz \odot \Pxv \G^{-1}\Lambdax \A &&(1)\\
   & + A^\top (\Rxu\Px\Rxz\Px + \Px\Rxu\Rxz\Px + \Px\Rxu\Px\Rxz\Px + \Px\Rxz\Px\Rxu)\odot (\Px\Rxv\Px) \G^{-1}\Lambdax \A  &&(2)\\
   & - \A^\top \Pxz \odot \Pxv \G^{-1} \D\G(u)\G^{-1}\Lambdax \A  &&(3)\\
   & + \A^\top \Pxz \odot \Pxv \G^{-1}\D\Lambdax(u)\A &&(4)\\
   & + \A^\top\Pxz \odot \Pxv \G^{-1}\Lambdax \Sxu \A. &&(5)\\
   & + \A^\top \tilde (\Px\D(r_{x,z})(u)\Px) \odot \Pxv \G^{-1}\Lambdax \A + \A^\top \Pxz \odot \Px D(r_v)(u)\Px \G^{-1}\Lambdax \A &&(7).
\end{align*}
For the first part (1), using Lemma~\ref{lem:hadamardlowner}:
\begin{align*}
    \qq^\top (1)\ell & \leq 
    \bracketss{\qq^\top \Sxu \Pxz \odot \Pxv \W^{-1} \Pxz\odot \Pxv \Sxu \qq} \bracketss{\qq^\top \A^\top \Lambdax \G^{-1}\W\G^{-1}\Lambdax \A\qq}\\
    &\leq \|s_{x,v}\|_\infty \|s_{x,z}\|_\infty \bracketss{\qq^\top \A^\top \Sxu \W \Sxu\A\qq} \bracketss{\qq^\top \A^\top \W \A\qq}\\
    &\leq \|s_{x,v}\|_\infty \|s_{x,z}\|_\infty \|s_{x,u}\|_\infty \qq^\top \A^\top \W \A \qq.
\end{align*}

For the second part (2), note that
\begin{align*}
    \qq^\top \A^\top \mathfrak D \G^{-1}\Lambdax \A\qq & \leq
    \bracketss{\qq^\top \A^\top \mathfrak D\W^{-1}\mathfrak D\A\qq}\bracketss{\qq^\top \A^\top \Lambdax \G^{-1}\W\G^{-1}\Lambdax \A\qq}\\
    & \lesssim \bracketss{\qq^\top \A^\top \W^{1/2}(\W^{-1/2}\mathfrak D\W^{-1/2})^2\W^{1/2}\A\qq}\bracketss{\qq^\top \A^\top \W \A\qq},
\end{align*}
where we are denoting the big chunk in the middle by $\mathfrak D$ for simplicity.
But combining Lemma~\ref{lem:lownerhelper} and~\ref{lem:hadamardlowner}
\begin{align*}
    \mathfrak D \leq \|\rxu\|_\infty \|\rxz\|_\infty \|\rxv\|_\infty \W,
\end{align*}
which implies
\begin{align*}
    \W^{-1/2}\mathfrak D\W^{-1/2} \leq \|\rxu\|_\infty \|\rxz\|_\infty \|\rxv\|_\infty I.
\end{align*}
Overall, we conclude
\begin{align*}
    \qq^\top(2)\qq \leq \|\rxu\|_\infty \|\rxv\|_\infty\|\rxz\|_\infty \qq^\top \W \qq.
\end{align*}
For (3):
\begin{align*}
    \qq^T (3) \qq &\leq \brackets{\qq^\top \A^\top \Pxz\odot \Pxv \G^{-1} \Pxz\odot \Pxv\A\qq}
    \bracketss{\qq^\top A^\top \Lambdax \G^{-1}\D\G(u)\G^{-1}\D\G(u)\G^{-1}\Lambdax \A\qq}\\
    &\leq \|\rxu\|_\infty \|\rxv\|_\infty\|\rxz\|_\infty \qq^\top \A^\top \W \A \qq.
\end{align*}
(4) and (5) are similar.
Term (7) is also similar to Equation $(\trr 4)(\trrr 3)$ after applying Lemma~\ref{lem:derivativeofr}. Next, we move on to $(\trrr 2)$ term.
\subsubsection*{$(\trrr 2)$ term}
\begin{align*}
    & \D(\A^\top \Big(\Rxz\Px + \Px\Rxz\Big) \odot \Pxv \G^{-1}\Lambdax \A)(u)\\
    & = \A^\top \Sxu (\Rxz\Px + \Px\Rxz)\odot \Pxv \G^{-1}\Lambdax \A  &&[1]\\
    & + \A^\top (\D(\Rxz)(u)\Px + \Px\D(\Rxz)(u)) \odot \Pxv \G^{-1}\Lambdax \A &&[2]\\
    & + \A^\top (\Rxz \Px + \Px\Rxz)\odot (\Rxu\Px\Rxv\Px \\
    &+ \Px\Rxu\Px\Rxv\Px+ \Px\Rxu\Rxv\Px+\Px\Rxv\Px\Rxu) \G^{-1}\Lambdax \A &&[3], [4], [5], [6]\\
    & + \A^\top \Big(\Rxz\Px + \Px\Rxz\Big) \odot (\Px\D(\Rxv)(u)\Px) \G^{-1}\Lambdax \A &&[7]\\
    & + \A^\top (\Rxz\Px+\Px\Rxz)\odot \Pxv \G^{-1}\D\G(u)\G^{-1}\Lambdax \A &&[8]\\
    & + \A^\top \Big(\Rxz\Px + \Px\Rxz\Big) \odot \Pxv \G^{-1}\D\Lambdax(u) \A &&[9]\\
    & + \A^\top \Big(\Rxz\Px + \Px\Rxu\Big) \odot \Pxv \G^{-1}\Lambdax \Sxu \A &&[10].
\end{align*}
Note that if $z$ differentiate any of the $\Rxu, \Rxv, \W,$ or $\W'$, then handling those terms is similar to Equation $(\trr 4) (\trrr 2)$.
\\
term [1] is similar to $(\trr 4) (\trrr 1)$ and$(\trr 4) (\trrr 2)$.
\\
term [2] is similar to Equation$(\trr 4) (\trrr 1)$ and $(\trr 4) (\trrr 2)$ after using Lemma~\eqref{lem:derivativeofr}.
\\
term [3] the first part is similar to Equation$(\trr 4) (\trrr 1)$. For the second part
\begin{align*}
    &\qq^\top \A^\top \Px\Rxz \odot (\Rxu\Px\Rxv\Px) \G^{-1}\Lambdax \A \qq\\ 
    &= \qq^\top \A^\top \Rxu \Px\odot (\Px\Rxv\Px) \Rxz \G^{-1}\Lambdax \A\qq,
\end{align*}
which similar to $(\trr 4) (\trrr 2)$ can be upper bounded by 
\begin{align}
&\|\Rxu s_{x,\qq}\|_{w_x}\|s_{x,\qq}\|_{w_x} \|\sxz\|_\infty \|\sxv\|_\infty \leq \|\rxu\|_\infty \|\sxz\|_\infty \|\sxv\|_\infty \|s_{x,\qq}\|_{w_x}^2 \\
&= \|\rxu\|_\infty \|\sxz\|_\infty \|\sxv\|_\infty \qq^\top \A^\top \W \A\qq.\label{eq:reductiontrick}    
\end{align}
as desired.
\\
term [4] the first part is similar to Equation$(\trr 4) (\trrr 7)$ combined with the trick in~\eqref{eq:reductiontrick}. For the second part:
\begin{align*}
    &\qq^\top \A^\top \Px\odot (\Px\Rxu\Px\Rxv\Px)\Rxz \G^{-1}\Lambdax \A\qq
    \\
    &\leq \bracketss{\qq^\top \A^\top \Px\odot (\Px\Rxu\Px\Rxv\Px)\W^{-1}\Px\odot(\Px\Rxu\Px\Rxv\Px)\A\qq}\\
    & \times \bracketss{\qq^\top \A^\top \Lambdax \G^{-1}\Rxz\W\Rxz\G^{-1}\Lambdax \A\qq}\\
    &\leq \|\rxu\|_\infty \|\rxv\|_\infty \|\rxz\|_\infty \|\qq\|_g^2.\numberthis\label{eq:likethis}
\end{align*}
term [5] the first part is similar to Equation $(\trr 4) (\trrr 1)$ and the second part is similar to~\eqref{eq:likethis}.
\\
term [6] part 1 is similar to [3] part 2, and part 2 is similar to term $(\trr 4) (\trrr 2)$ part 2.  
\\
term [7] is similar to$(\trr 4) (\trrr 2)$.
\\
term [8], the first part is similar to$(\trr 4) (\trrr 4)$ using the trick in~\eqref{eq:reductiontrick}. For term [8] second part
\begin{align*}
    \qq^\top [8] \qq & \leq
    \bracketss{\qq^\top A^\top (\Px\odot \Pxv) \W^{-1} (\Px\odot \Pxv) \A\qq}\\
    &\times \bracketss{\qq^\top \A^\top \Lambdax \G^{-1}\D \G(u)\G^{-1}\Sxz\W\Sxz \G^{-1}\D\G(u)\G^{-1}\Lambdax \A\qq} \\
    & \leq \|r_{x,v}\|_\infty\|r_{x,z}\|_\infty\|r_{x,u}\|_\infty \|\qq\|_g^2.
\end{align*}
term [9] is similar to what we did for [9]. term [10] first part similar to~$(\star 4) (\ast 6)$ using the trick in~\eqref{eq:reductiontrick}. for the second part:
\begin{align*}
    \qq^\top [10] \qq & \leq
    \bracketss{\qq^\top \A^\top (\Px\odot \Pxv) \W^{-1} (\Px\odot \Pxv) \A\qq}\bracketss{\qq^\top \A^\top\Sxu \Lambdax \G^{-1}\Sxz\W\Sxz\G^{-1}\Lambdax \Sxu\A\qq} \\
    & \leq \|r_{x,v}\|_\infty\|r_{x,z}\|_\infty\|r_{x,u}\|_\infty \|\qq\|_g^2.
\end{align*}

\subsubsection*{$(\trrr 4)$ term}
\begin{align*}
    &\D(\A^\top \Px\odot \Pxv \G^{-1} \D\G(z)\G^{-1}\Lambdax \A)(u) = \\
    &\A^\top \Sxu \Px\odot \Pxv \G^{-1} \D\G(z)\G^{-1}\Lambdax \A &&[1]\\
    &+\A^\top \Px^{u}\odot \Pxv \G^{-1} \D\G(z)\G^{-1}\Lambdax \A &&[2]\\
    &+\A^\top \Px\odot \big(\Px\D(\Rxv)(u)\Px\big) \G^{-1} \D\G(z)\G^{-1}\Lambdax \A &&[3]\\
    &+\A^\top \Px\odot \big(\Rxu\Px\Rxv\Px + \Px\Rxu\Rxv\Px+\Px\Rxv\Px\Rxu + \Px\Rxv\Px\Rxu\Px \\
    &+ \Px\Rxu\Px\Rxv\Px\big) \G^{-1} \D\G(z)\G^{-1}\Lambdax \A &&[4]\\
    &+\A^\top \Px\odot \Pxv \G^{-1}(\D\G(u)\G^{-1} \D\G(z) + \D\G(z)\G^{-1}\D\G(u))\G^{-1}\Lambdax \A&&[5]\\
    &+\A^\top \Px\odot \Pxv \G^{-1} \D\G(z,u)\G^{-1}\Lambdax \A &&[6]\\
    &+\A^\top \Px\odot \Pxv \G^{-1}\D\G(z)\G^{-1}\D\Lambdax(u)\A &&[7]\\
    &+\A^\top \Px\odot \Pxv \G^{-1}\D\G(z)\G^{-1}\Lambdax \Sxu \A. &&[8]
\end{align*}
term [1] is similar to $(\trr 4) (\trrr 2) [8]$. 
\\
term [2] is similar to $(\trr 4) (\trrr 3) [3]$.
\\
term [3] is handled by Lemma~\ref{lem:derivativeofr}. 
\\
term [4] first part is similar to $(\trr 4) (\trrr 2) [8]$ part 1. term [4] part 2 is similar to $(\trr 4) (\trrr 4)$. term [4] part 3 is similar to $(\trr 4) (\trrr 2) [8]$ part 2. For term [4] parts 4 and 5:
\begin{align*}
    \qq^\top (.) \qq \leq &\bracketss{\qq^\top \A^\top \Px\odot (\Px\Rxv\Px\Rxu\Px + \Px\Rxu\Px\Rxv\Px)\W^{-1}\Px\\
    &\odot (\Px\Rxv\Px\Rxu\Px + \Px\Rxu\Px\Rxv\Px) \A\qq} \\
    &\times \bracketss{\qq^\top \A^\top \Lambdax \G^{-1}\D\G(z)\G^{-1}\W\G^{-1}\D\G(z)\G^{-1}\Lambdax \A\qq}\\
    & \leq \|r_{x,v}\|_\infty\|r_{x,z}\|_\infty\|r_{x,u}\|_\infty \|\qq\|_w^2.
\end{align*}
\\
term [5] is similar to $(\trr 4) (\trrr 4)$.
\\
term [6]:
\begin{align*}
    \qq^\top (.)\qq  \leq &\bracketss{\qq^\top \A^\top \Px\odot \Pxv\G^{-1}\W\G^{-1}\Px\odot \Pxv \A \qq}\\
    &\times \bracketss{\qq^\top \A^\top \Lambdax \G^{-1}\D\G(z,u)\G^{-1}\W\G^{-1}\D\G(z,u)\G^{-1}\Lambdax \A\qq}\\
    & \leq \|r_{x,v}\|_\infty\|r_{x,z}\|_\infty\|r_{x,u}\|_\infty \|\qq\|_w^2.
\end{align*}
term [7]: similar to [6].
\\
term [8]:
\begin{align*}
    \qq^\top (.)\qq & \leq \bracketss{\qq^\top \A^\top \Px\odot \Pxv\G^{-1}\W\G^{-1}\Px\odot \Pxv \A \qq}\\
    &\times \bracketss{\qq^\top \A^\top\Sxu \Lambdax \G^{-1}\D\G(z)\G^{-1}\W\G^{-1}\D\G(z)\G^{-1}\Lambdax \Sxu \A\qq}\\
    & \leq \|r_{x,v}\|_\infty\|r_{x,z}\|_\infty\|s_{x,u}\|_\infty \|\qq\|_w^2.
\end{align*}

\subsubsection*{$(\trrr 5)$ term}
\begin{align*}
    &\D(\A^\top \Px \odot \Pxv \G^{-1} \D\Lambdax(z) \A)(u) =\\
    &\A^\top \Sxu \Px \odot \Pxv \G^{-1} \D\Lambdax(z) \A\\
    & +\A^\top \Px \odot \Pxv \G^{-1} \D\Lambdax(z) \A\\
    &\dots
\end{align*}
These terms are similar to $(\trr 4) (\trrr 4)$.

\subsubsection*{$(\trrr 6)$ term}
\begin{align*}
    \D(\A^\top \Px\odot \Pxv \G^{-1} \Lambdax \Sxz \A)(u)  
\end{align*}
these terms are similar to $(\trr 4) (\trrr 4)$.

\subsubsection*{$(\trrr 7)$ term}
\begin{align*}
    &\D(\A^\top \Px\odot \big(\Rxu\Px\Rxv\Px + \Px\Rxu\Rxv\Px+\Px\Rxv\Px\Rxu \\
    &+ \Px\Rxv\Px\Rxu\Px + \Px\Rxu\Px\Rxv\Px\big) \G^{-1}\Lambdax \A)(u)=\\
    &\A^\top \Sxu \Px\odot \big(\Rxu\Px\Rxv\Px + \Px\Rxu\Rxv\Px+\Px\Rxv\Px\Rxu \\
    &+ \Px\Rxv\Px\Rxu\Px + \Px\Rxu\Px\Rxv\Px\big) \G^{-1}\Lambdax \A &&[1]\\
    &+\A^\top \Px\Sxu\Px \odot (\Rxz\Px\Rxv\Px + \Px\Rxz\Px\Rxv\Px+ \Px\Rxz\Rxv\Px+\Px\Rxv\Px\Rxz) \G^{-1}\Lambdax \A &&[2]\\
    &+\A^\top \Px\odot \big(\sum \Px\Sxv\Px\Sxz\Px\Sxu\Px + \sum \Sxv\Px\Sxz\Px\Sxu \\
    &+ \sum \Px\Sxv\Px\Sxu\Sxz + \sum \Sxu \Sxv \Px \Sxz\Px + \sum \Sxu \Px\Sxv\Sxz\Px \\
    &+ \sum \Px\Sxv\Sxz\Px\Sxu \big)\G^{-1}\Lambdax \A &&[3:1], \dots\\
    & - \A^\top \Px\odot \big(\Rxu\Px\Rxv\Px + \Px\Rxu\Rxv\Px+\Px\Rxv\Px\Rxu \\
    &+ \Px\Rxv\Px\Rxu\Px + \Px\Rxu\Px\Rxv\Px\big) \G^{-1}\D\G(u)\G^{-1}\Lambdax \A  &&[4:1]\dots\\
    & + \A^\top \Px\odot \big(\text{same as above}\big) \G^{-1}\D\Lambdax(u)\A &&[5]\\
    & + \A^\top \Px\odot \big(\text{same as above}\big) \G^{-1}\Lambdax \Sxu \A &&[6].
\end{align*}
\\
where for simplicity, we have used the $\sum$ notation indicating all possible symmetric combinations of that term with respect to $v$, $w$, and $u$.
\\
term [1]: considering the quadratic form $\qq^\top(.)\qq$ on this term, note that on the left we get $s_{x,\qq}^\top \Sxu \dots s_{x,\qq}$. Now we can just reduce this term to $(\trr 4) (\trrr 7)$ to conclude
\begin{align*}
   s_{x,\qq}^\top \Sxu \dots s_{x,\qq} \leq \|s_{x,v}\|_\infty \|s_{x,v}\|_\infty \|s_{x,\qq} \Sxu\|_w \|s_{x,\qq}\|_w \leq \|s_{x,\qq}\|_w^2\|s_{x,u}\|_\infty\|s_{x,v}\|_\infty\|s_{x,z}\|_\infty. 
\end{align*}
term [2]: similar to $(\trr 4)(\trrr 2)$.
\\
term [3:1]: Noting the fact that 
$$
- \|s_{x,u}\|_\infty \|s_{x,v}\|_\infty \|s_{x,z}\|_\infty \Px \leq \Px\Sxv\Px\Sxz\Px\Sxu\Px \leq \|s_{x,u}\|_\infty \|s_{x,v}\|_\infty \|s_{x,z}\|_\infty \Px,
$$
and using Lemma~\ref{lem:hadamardlowner} this term is similar to~\eqref{star4ast3}.
\\
term [3:2]: note that this term is equal to
\begin{align*}
    \A^\top \Sxv \Px\odot (\Px\Sxz\Px) \Sxu\G^{-1}\Lambdax \A 
\end{align*}
which is similar to $(\trr 4) (\trrr 2) [1]$.
\\
term [3:3], [3:4], [3:5], [3:6]: similar to $(\trr 4) (\trrr 7)$.
\\
term [4:1] is similar to $(\trr 4) (\trrr 2) [7]$.
\\
term [4:2], [4:4], [4:5] similar to $(\trr 4) (\trrr 3) [3]$.
\\
term [4:3] similar to $(\trr 4) (\trrr 2) [8]$
\\
term [5] is similar to [4].
\\
term [6] is also similar to $(\trr 4) (\trrr 3) [5]$ and $(\trr 4) (\trrr 2) [10]$.

\subsubsection*{$(\trrr 8)$ term}
\begin{align*}
    \D(\A^\top \Px \odot \Px \D(\Rxv)(z)\Px \G^{-1}\Lambdax \A)(u)
\end{align*}
This term is similar to $\D(\trr 4)(z)$ as detailed in Lemma~\ref{lem:star4}.

\subsubsection*{$(\trrr 1)$ term}
\begin{align*}
    &\D(\A^\top \Sxz \Px\odot \Pxv \G^{-1} \Lambdax \A)(u)
\end{align*}
We have handled this term with regards to the differentiation of any term with respect to $u$, we can instead first take that derivation with respect to $u$ and then take the derivative of $A$ which respect to $z$ which spits out the $\Sxz$.

Now based on the form of the metric written in Lemma~\ref{lem:metricsecondform}, we first focus on the last term $\A^\top \Ptwo\G^{-1}\Ptwo\A$.
Note that above in handling all the derivatives in directions $z$ and $u$ of the $(\trr 4)$ term, we have bounded all the 3rd order derivative terms of $\D^3g(u,v,z)$ that has at least one derivative regarding the $\Ptwo$ terms in $\A^\top\Ptwo\G^{-1}\Ptwo\A$. Hence, regarding this term, it remains to take derivative with only with respect to $\G^{-1}$ and the $\A$'s which we do next. Again, the sums mean we are considering all the terms corresponding to all the permutations of $u, v,z$ regarding the current term.
\begin{align*}
    &\D(\A^\top \Ptwo\G^{-1}\Ptwo\A)(u,v,z) \rightarrow \sum \A^\top\Ptwo\G^{-1}\D\G(u)\G{-1}\D\G(v)\G^{-1}\D\G(z)\G^{-1}\Ptwo\A\\
    &+\sum \A^\top\Ptwo\G^{-1}\D\G(u)\G^{-1}\D\G(v)\G^{-1}\D\G(z)\G^{-1}\Ptwo\A\\
    &+\sum \A^\top \Ptwo\G^{-1}\D\G(u,v)\G^{-1}\D\G(z)\G^{-1}\Ptwo\A\\
    &+\sum \A^\top \Ptwo\G^{-1}\D\G(u,v,w)\G^{-1}\Ptwo\A\\
    &+\sum \A^\top \Sxu \Ptwo \Ptwo\big(\G^{-1}\D\G(v,z)\G^{-1} + \G^{-1}\D\G(v)\G^{-1}\D\G(z)\G^{-1}\big)\Ptwo\A\\
    &+\sum \A^\top \Ptwo \Ptwo\big(\G^{-1}\D\G(v,z)\G^{-1} + \G^{-1}\D\G(v)\G^{-1}\D\G(z)\G^{-1}\big)\Ptwo\Sxu\A.
\end{align*}
$(\trr 3)$ is handled in a similar way as $(\trr 4)$.

To handle the rest of the derivatives more conveniently at this point, we consider the second form of metric in Equation~\ref{lem:metricsecondform}. First we aim to handle all the possible derivatives in three directions which differentiate the $\Ptwo$ terms at least once.  Taking one time derivative in direction $v$ from the $\Ptwo$ term results in term $(\trr 4)$ and $(\trr 3)$ in Lemma~\eqref{lem:gderivative}.

But using Lemmas~\ref{lem:thirdderivativeGbound} and~\ref{lem:Gsecondderivative} and similar technique as we did, these terms are bounded by plus and minus of two constants times the matrix $\|s_u\|_\infty \|s_v\|_\infty \|s_z\|_\infty \A^\top \W \A$.
\\
Next, we move on to the other terms in the formulation of $g_2$ in~\ref{lem:metricsecondform}, namely
$\A^\top \W \A$, $\A^\top \Lambdax \A$, $\A^\top \G \A$, and $\A^\top \Ptwo\A$. third order self concordance of $\A^\top \W \A$ is a direct consequence of Lemma~\ref{lem:wprimesecondderivative}. Term $\A^\top \Lambdax \A$ and $\A^\top \G \A$ are handled by Lemma~\ref{lem:thirdderivativeGbound}, and $\A^\top \Ptwo\A$ is handled by Lemma~\ref{lem:thirdderivativep2}.
\end{proof}

\newpage
\section{Derivative Stability Lemmas}\label{sec:derivativestability}
\subsection{Infinity norm comparisons}\label{sec:infnormcomparison}
Here we show a control over the infinity to infinity norm, i.e. $\|.\|_{\infty \rightarrow \infty}$ of the matrix $G^{-1}W$, which is a crucial property that we use all over the proof to derive our derivative estimates with respect to the $\|.\|_{x,\infty}$ norm.
\begin{lemma}\label{lem:ginfnorm}
For $\gv = \G^{-1}\W\sv$, given any vector $\sv$ and $p < 4$, we have
\begin{align*}
    \|\gv\|_\infty \leq \frac{1}{4/p - 1}\|\sv\|_\infty.
\end{align*}
\end{lemma}
\begin{proof}
Set $\|\sv\|_\infty = \ell$. then
\begin{align*}
    \W\sv = \G\gv = \frac{2}{p}\W\gv + (1-\frac{2}{p})\Ptwo\gv.
\end{align*}
Now suppose $\|\gv\|_\infty \geq \frac{1}{4/p - 1}\ell$, which implies that for the maximizing index $i$ we have
\begin{align*}
    |\gv_i| \geq \frac{1}{4/p-1}\ell.
\end{align*}
But note that
\begin{align*}
    |\gv^\top{{\Ptwo}_{i,}}| \leq w_i \|\gv\|_\infty = w_i \gv_i,
\end{align*}
hence
\begin{align*}
    {\gv}^\top{\G}_i \geq \frac{2}{p}w_i\gv_i - (1-\frac{2}{p}) w_i \gv_i = (\frac{4}{p} - 1) w_i \gv_i > w_i \ell.
\end{align*}
On the other hand
\begin{align*}
    \gv^\top{\G}_i = w_i \sv_i \leq w_i \ell.
\end{align*}
The contradiction finishes the proof.
\end{proof}

\subsection{Lowner Inequalities}\label{sec:lownerineq}
In this section, we drive important estimates on the derivatives of fundamental matrix quantities that we arrive at such as $\G,\Lambdax,\Rxv,\Sxv$ that we defined, and use them in our proof for strong self-concordance.
\begin{lemma}\label{lem:ptildebound}
We have
\begin{align*}
    -\frac{1}{4/p - 1} \|s_{x,v}\|_\infty \W \preccurlyeq \Px \odot \Pxv \preccurlyeq \frac{1}{4/p - 1} \|s_{x,v}\|_\infty \W.
\end{align*}
\end{lemma}
\begin{proof}
For the matrix $\Pxv$ we have
\begin{align*}
    \Px \odot \Pxv & \preccurlyeq \|r_{x,v}\|_\infty \Px \odot \Px\\
    & \preccurlyeq \frac{1}{4/p - 1} \|s_{x,v}\|_\infty \Ptwo \preccurlyeq \frac{1}{4/p - 1} \|s_{x,v}\|_\infty \W,
\end{align*}
and similarly
\begin{align*}
     -\frac{1}{4/p - 1} \|s_{x,v}\|_\infty W \preccurlyeq \Px \odot \Pxv.
\end{align*}
\end{proof}

\begin{lemma}\label{lem:levelzerobounds}
    We have
    \begin{align*}
        &\Ptwo \preccurlyeq \W,\\
        &\Lambdax \preccurlyeq \W,\\
        &\frac{2}{p}\W \preccurlyeq \G \preccurlyeq \W.
    \end{align*}
\end{lemma}
\begin{proof}
    For the first inequality, note that the sum of entries of the $i$th row of matrix $\Ptwo$ is equal to ${\W}_{ii}$. Hence, the matrix $\W - \Ptwo$ is a Laplacian so it is positive semi-definite. The second inequality follows from the fact that $\Ptwo$ is PSD. The third inequality, using the fact that $\Ptwo \preccurlyeq \W$:
    \begin{align*}
        \frac{2}{p}\W \preccurlyeq \frac{2}{p}\W + (1 - \frac{2}{p})\Ptwo \preccurlyeq \W.
    \end{align*}
\end{proof}

\begin{lemma}\label{lem:Glambdabound}
For the derivatives of $\G$ and $\Lambdax$ at some point $x$ we have
\begin{align*}
   -\|s_{x,z}\|_\infty \W \preccurlyeq \D\G(z) \preccurlyeq \|s_{x,z}\|_\infty  \W,\\
-\|s_{x,z}\|_\infty  \W \preccurlyeq \D\Lambdax(z) \preccurlyeq \|s_{x,z}\|_\infty  \W.
\end{align*}
\end{lemma}
\begin{proof}
Directly from Lemmas~\ref{lem:ptildebound} and~\ref{lem:lem:wbound}.
\end{proof}

\begin{lemma}
\begin{align*}
    \|\D^2(\Rxv)(z,u)\| \leq \|s_{x,v}\|_\infty \|s_{x,u}\|_\infty \|s_{x,z}\|_\infty.
\end{align*}
\end{lemma}
\begin{proof}
We use the terms of the derivative of $R_v$ in direction $z$ (according to Lemma~\ref{lem:derivativeofr}) and differentiate them one by one with respect to $u$:
\begin{align*}
    \D(\G^{-1}\D\G(z)\G^{-1}\Wxz)(u) \rightarrow
    &-\G^{-1}\D\G(u)\G^{-1}\D\G(z)\G^{-1}\W s_{x,z}\\
    &+\G^{-1}\D\G(z,u)\G^{-1}\W s_{x,z}\\
    &+\G^{-1}\D\G(z)\G^{-1}\D\G(u)\G^{-1}\W s_{x,z}\\
    &+\G^{-1}\D\G(z)\G^{-1}\D(\W)(u)s_{x,z}\\
    &+\G^{-1}\D\G(z)\G^{-1}\W\Sxu s_{x,z}.
\end{align*}
Now from Lemmas~\ref{lem:ginfnorm} and~\ref{lem:infnorm2} and~\ref{lem:lem:wbound} 
we have
\begin{align*}
    \|\G^{-1}\D\G(u)\G^{-1}\D\G(z)\G^{-1}\W s_{x,z}\|_\infty & \leq \|s_{x,u}\|_\infty \|\G^{-1}\D\G(z)\G^{-1}\W s_{x,v}\|_\infty \\
    & \leq \|s_{x,u}\|_\infty \|s_{x,z}\|_\infty \|\G^{-1}\W s_{x,v}\|_\infty \leq \|s_{x,u}\|_\infty \|s_{x,z}\|_\infty \|s_{x,v}\|_\infty,
\end{align*}
\begin{align*}
    \|\G^{-1}\D\G(z,u)\G^{-1}\W s_{x,z}\|_\infty &\leq \|s_{x,z}\|_\infty \|s_{x,u}\|_\infty \|\G^{-1}\W s_{x,z}\|_\infty \leq \|s_{x,z}\|_\infty \|s_{x,u}\|_\infty \|s_v\|_\infty,
\end{align*}
the third and forth terms are similar to the first and second terms resp., for the fifth term
\begin{align*}
    \|\G^{-1}\D\G(v)\G^{-1}\D(\W)(u)s_{x,z}\|_\infty &\leq \|s_{x,z}|_\infty \|s_v\|_\infty \|s_u\|_\infty.
\end{align*}
the derivatives of the other terms are handled in a similar way.
\end{proof}

\begin{lemma}\label{lem:lownerhelper}
For a symmetric matrix $D$ with $-\W \leq D \leq \W$, we have
\begin{align*}
    -\W\|r_v\|_\infty \leq \Rxv D + D\Rxv \leq \W\|r_{x,v}\|_\infty.
\end{align*}
\end{lemma}
\begin{proof}
For arbitrary vectors $\qq_1,\qq_2$, using the inequality $\qq_1^\top D \qq_2 \leq \sqrt{\qq_1^\top \W \qq_1} \sqrt{\qq_2^\top \W \qq_2}$ with $\qq_1 = \Rxv \qq$ and $\qq_2 = \qq$:
\begin{align*}
    \qq^\top \Rxv D\qq \leq \sqrt{\qq^\top \Rxv\W\Rxv\qq}\sqrt{\qq^\top \W \qq} \leq \|r_{x,v}\|_\infty \qq^\top \W \qq.
\end{align*}
\end{proof}

\begin{lemma}\label{lem:hadamardlowner}
For diagonal matrices $D_1, D_2, D_3$(not necessarily positive) we have
\begin{align*}
    &-\|D_1\|\|D_2\|P^{(2)}\leq \Px D_1\Px \odot \Px D_2\Px \leq \|D_1\|\|D_2\|\Ptwo,\\
    &-\|D_1\|\|D_2\|\|D_3\|\Ptwo\leq \Px(D_2\Px D_3 + D_3\Px D_2)\Px \odot \Px D_1\Px \leq \|D_1\|\|D_2\|\|D_3\|\Ptwo.
\end{align*}
\end{lemma}
\begin{proof}
    Consider the Choleskey decomposition of $\Px$:
    \begin{align*}
        \Px = \sum_{i=1}^n u_i u_i^\top,
    \end{align*}
    Then for the first inequality, note that we can write $\Px D_1\Px$ as
    \begin{align}
      \Px D_1\Px = \sum_{i=1}^n (u_i^\top D_1 u_i)u_i u_i^\top.\label{eq:choleskey} 
    \end{align}
    Hence, for arbitrary vector $\ell$:
    \begin{align*}
        \Big|\qq^\top (\Px D_1\Px \odot \Px D_2 \Px)\qq \Big|&\leq \sum_{i=1}^n |u_i^\top D_1 u_i| |(\ell \odot u_i)^\top (\Px D_2\Px) (\qq \odot u_i)|\\
        &\leq \sum_{i=1}^n \|D_1\| |(\ell \odot u_i)^\top (\Px D_2\Px) (\qq \odot u_i)|\\
        &\leq \sum_i \|D_1\| \|D_2\| (\qq \odot u_i)^\top \Px (\qq \odot u_i)\\
        &= \|D_1\|\|D_2\|\qq^\top \Ptwo \qq.
    \end{align*}
    For the second inequality, note that
    \begin{align*}
        \qq^\top(D_2 \Px D_3 + D_3 \Px D_2)\qq \leq 
        2(D_2\qq)^\top \Px(D_3 \qq) \leq 2\|D_2 \qq\|_2 \|D_3 \qq\|_2 \leq 2\|D_2\|\|D_3\|\|\qq\|_2^2,
    \end{align*}
    which implies
    \begin{align*}
       -\|D_2\|\|D_3\| I \preccurlyeq D_2\Px D_3 + D_3\Px D_2 \preccurlyeq \|D_2\|\|D_3\| I.
    \end{align*}
    Therefore
    \begin{align*}
       -\|D_2\|\|D_3\|\Px \preccurlyeq \Px(D_2\Px D_3 + D_3\Px D_2)\Px \preccurlyeq \|D_2\|\|D_3\| \Px.
    \end{align*}
    Now again using Equation~\eqref{eq:choleskey}:
    \begin{align*}
        &\Big|\qq^\top(  \Px(D_2\Px D_3 + D_3\Px D_2)\Px)\odot (\Px D_1\Px)\qq\Big|\\
        &\leq \sum_{i=1}^n \big|u_i^\top D_1u_i\big| |(\qq\odot u_i)^\top (\Px D_2\Px D_3\Px + \Px D_3\Px D_2\Px)(\qq\odot u_i)| \\
        &\leq \|D_1\|D_2\|\|D_3\| (\qq\odot u_i)^\top \Px (\qq\odot u_i)\\
        &= \|D_1\|\|D_2\|\|D_3\|\qq^\top \Ptwo\qq.
    \end{align*}
\end{proof}

\begin{lemma}\label{lem:lownerhelper2}
Given a matrix $-\W \leq D \leq \W$ and arbitrary diagonal matrices $V_1$ and $V_2$ and arbitrary vector $\ell$:
\begin{align*}
    \qq^\top V_1 D V_2 \qq \leq \|V_1\|_{op}\|V_2\|_{op} \|\qq\|_w^2.
\end{align*}
\end{lemma}
\begin{proof}
simply by Cauchy Schwarz:
\begin{align*}
    LHS
    \leq \|V_1\qq\|_w \|V_2\qq\|_w \leq \|V_1\| \|V_2\| \|\qq\|_w^2.
\end{align*}
\end{proof}

\begin{lemma}\label{lem:Gsecondderivative}
For matrices $G$ and $\Lambdax$ we have
\begin{gather*}
    -\|s_{x,v}\|_\infty \|s_{x,z}\|_\infty \W \leq \D^2\G(v,z) \leq \|s_{x,v}\|_\infty \|s_{x,z}\|_\infty \W,\\
    -\|s_{x,v}\|_\infty \|s_{x,z}\|_\infty \W \leq \D\Lambdax(v,z) \leq \|s_{x,v}\|_\infty \|s_{x,z}\|_\infty \W.
\end{gather*}
\end{lemma}
\begin{proof}
Note that
\begin{align*}
    & \D\G(v,z) = \frac{2}{p}\D\Wxv(z) \\
    & + (1-\frac{2}{p})\Px\odot \Px \D(\Rxv)(z)\Px \\
    &+
    (1-\frac{2}{p})\Px\odot (\Rxz\Px\Rxv\Px+\Px\Rxz\Rxv\Px + \Px\Rxv\Px\Rxz + \Px\Rxv\Px\Rxz\Px + \Px\Rxz\Px\Rxv\Px) \\
    & + (1-\frac{2}{p})(\Px\Rxz\Px)\odot (\Px\Rxv\Px).
\end{align*}
But using Lemma~\ref{lem:hadamardlowner}:
\begin{align*}
    -\|r_{x,z}\|_\infty\|r_{x,v}\|_\infty \W \leq (\Px\Rxz\Px)\odot (\Px\Rxv\Px) \leq \|r_{x,z}\|_\infty\|r_{x,v}\|_\infty \W.
\end{align*}
On the other hand, note that
\begin{align*}
    \Rxz\Px\Rxv + \Rxv\Px\Rxz \leq \|r_{x,v}\|_\infty \|r_{x,z}\|_\infty I,
\end{align*}
simply by checking the operator norm of LHS. Hence
\begin{align*}
    \Px\Rxz\Px\Rxv\Px + \Px\Rxv\Px\Rxz\Px \leq \|r_{x,v}\|_\infty \|r_{x,z}\|_\infty \Px. 
\end{align*}
On the other hand, 
\begin{align*}
    \Px\Rxz \Rxv\Px \leq \|r_{x,v}\|_\infty \|r_{x,z}\|_\infty \Px.
\end{align*}
Therefore, by Schur product theorem
\begin{align*}
    \Px\odot (\Px\Rxz\Rxv\Px + \Px\Rxv\Px\Rxz\Px + \Px\Rxz\Px\Rxv\Px) &\leq \|r_{x,v}\|_\infty \|r_{x,z}\|_\infty \Ptwo\\
    & \|r_{x,v}\|_\infty \|r_{x,z}\|_\infty \W.
\end{align*}
Moreover, 
\begin{align*}
\Px\odot (\Rxz\Px\Rxv\Px + \Px\Rxv\Px\Rxz) = 
\Rxz(\Px\odot (\Px\Rxv\Px)) + (\Px\odot (\Px\Rxv\Px))\Rxz,
\end{align*}
and note that
\begin{align*}
    \Px\odot \Px\Rxv\Px \leq \|r_{x,v}\|_\infty \W.
\end{align*}
Hence, by Lemma~\ref{lem:lownerhelper}
\begin{align*}
    \Rxz(\Px\odot (\Px\Rxv\Px)) + (\Px\odot (\Px\Rxv\Px))\Rxz \leq \W\|r_{x,v}\|_\infty \|r_{x,z}\|_\infty.
\end{align*}
Finally, note that from Lemma~\ref{lem:derivativeofwprime}:
\begin{align*}
    \D\Wxv(z) \leq \|r_{x,v}\|_\infty \|r_{x,z}\|_\infty \W.
\end{align*}
All the inequalities that we wrote also hold in the other direction with a negative sign. Combining all the inequalities concludes the proof for $\G$. As $\Lambdax$ is also a linear combination of $W$ and $\Ptwo$, using the exact same bounds we can obtain the conclusion for $\Lambdax$ as well.
\end{proof}

\begin{lemma}\label{lem:thirdderivativep2}
We have
\begin{align*}
    -\|u\|_\infty \|v\|_\infty \|z\|_\infty \W \leq \D^3\Ptwo(u,v,z) \leq \|u\|_\infty \|v\|_\infty \|z\|_\infty W.
\end{align*}
\end{lemma}

\begin{proof}
We have
\begin{align*}
    &\D^3\Ptwo(u,v,z) \\
    &= \Px\odot \Big(
    \sum \Px\Sxu\Sxz\Px\Sxv\Px + \sum \Px\Sxu\Px\Sxv\Sxz\Px \\
    &+ \sum \Sxu\Px\Sxv\Px\Sxz + \sum \Sxu \Sxv\Px\Sxz\Px\\
    &+ \sum \Px\Sxz\Px\Sxu\Sxv + \sum \Sxu\Px\Sxv\Sxz \Px\\
    &+ \sum \Px\Sxu\Sxv\Px\Sxz + \sum \Px\Sxu\Sxv\Sxz \\&+ \sum \Sxu\Sxv\Sxz\Px\Big)\\
    & + \sum (\Px\Sxu\Px + \Sxu\Px + \Px\Sxu) \odot (\Px\Sxv\Sxz\Px + \Px\Sxv\Px\Sxz\Px + \Sxv\Px\Sxz\Px \\
    &+ \Px\Sxv\Px\Sxz + \Px\Sxv\Px\Sxu\Px).\\ 
\end{align*}
Note that from Lemma~\ref{lem:lownerhelper}, a generic term in the above is of the form 
\begin{align*}
    D_1 (\Px D_2\Px)\odot (\Px D_3\Px) D_4
\end{align*}
for diagonal matrices $D_1$ and $D_4$, such that 
\begin{align*}
    \|D_1\|\|D_2\|\|D_3\|\|D_4\|\leq \|s_{x,v}\|_\infty\|s_{x,u}\|_\infty\|s_{x,z}\|_\infty.
\end{align*}
Hence, combining Lemmas~\ref{lem:lownerhelper} and~\ref{lem:hadamardlowner}, we get
\begin{align*}
    \D^3\Ptwo(u,v,z) \leq \|s_{x,u}\|_\infty \|s_{x,v}\|_\infty \|s_{x,z}\|_\infty \W.
\end{align*}
Similarly, we can show
\begin{align*}
   -\|s_{x,u}\|_\infty \|s_{x,v}\|_\infty \|s_{x,z}\|_\infty \W \leq \D^3\Ptwo(u,v,z).
\end{align*}
\end{proof}
\begin{lemma}\label{lem:thirdderivativeGbound}
We have
\begin{align*}
    -\|s_{x,u}\|_\infty \|s_{x,v}\|_\infty \|s_{x,z}\|_\infty \G \leq \D^3\G(u,v,z) \leq \|s_{x,u}\|_\infty \|s_{x,v}\|_\infty \|s_{x,z}\|_\infty \G,\\
    -\|s_{x,u}\|_\infty \|s_{x,v}\|_\infty \|s_{x,z}\|_\infty \Lambdax \leq \D^3\Lambdax(u,v,z) \leq \|s_{x,u}\|_\infty \|s_{x,v}\|_\infty \|s_{x,z}\|_\infty \Lambdax.
\end{align*}
\end{lemma}
\begin{proof}
Directly from Lemmas~\ref{lem:wprimesecondderivative} and~\ref{lem:thirdderivativep2}.
\end{proof}

\begin{lemma}\label{lem:wprimesecondderivative}
We have
\begin{align*}
    -\|s_{x,z}\|_\infty \|s_{x,u}\|_\infty \|s_{x,v}\|_\infty \W \leq \D^2(\Wxv)(z,u) \leq \|s_{x,z}\|_\infty \|s_{x,u}\|_\infty \|s_{x,v}\|_\infty \W.
\end{align*}
\end{lemma}
\begin{proof}
For the first term of the first derivative in~\eqref{eq:tmp4}, further taking derivative. with respect to $u$:
\begin{align*}
   &\W^{-1}\D(\D\Lambdax(z)\G^{-1}\W s_z)(u)\\
   &=\W^{-1}\D^2\Lambdax(z,u)\G^{-1}\W s_{x,v}\\ &+ \W^{-1}\D\Lambdax(z)\G^{-1}\D\G(u)\G^{-1}\W s_{x,v}\\& + \W^{-1}\D\Lambdax(z)\G^{-1}\D\W(u)s_{x,v} \\
   &+ \W^{-1}\D\Lambdax(z)\G^{-1}\D \W(u)\Sxu s_{x,v}\\
   &\leq \|s_{x,z}\|_\infty \|s_{x,u}\|_\infty \|s_{x,v}\|_\infty I,
\end{align*}
where we used Lemmas~\ref{lem:infnorm2} and~\ref{lem:infnormhelper2}.
For the second term in~\eqref{eq:tmp4}:
\begin{align*}
    &\W^{-1}\D(\Lambdax \G^{-1}\D\G(z)\G^{-1}\W s_{x,v})(u)\\
    &= \W^{-1}\D(\Lambdax)(z)\G^{-1}\D\G(z)\G^{-1}\W s_{x,v}\\
    &+\W^{-1}\Lambdax \G^{-1}(\D\G(u)\G^{-1}\D\G(z) + \D\G(z)\G^{-1}\D\G(u))\G^{-1}\W s_{x,v},\\
    &+\W^{-1}\Lambdax \G^{-1}\D^2\G(z,u)\G^{-1}\W s_{x,v},\\
    &+\W^{-1}\Lambdax \G^{-1}\D\G(z)\G^{-1}\D(\W)(u)s_{x,v},\\
    &+\W^{-1}\Lambdax \G^{-1}\D\G(z)\G^{-1}\W\Sxu s_{x,v}\\
    & \leq  \|s_{x,z}\|_\infty \|s_{x,u}\|_\infty \|s_{x,v}\|_\infty I.
\end{align*}
For the third term:
\begin{align*}
    &\W^{-1}\D(\Lambdax \G^{-1}\W'^{z}s_{x,v})(u)\\
    &=\W^{-1}\D(\Lambdax)(u)\G^{-1}\W'^{z}s_{x,v}\\
    &-\W^{-1}\Lambdax \G^{-1}\D\G(u)\G^{-1}\W'^{z}s_{x,v}\\
    &+\W^{-1}\Lambdax \G^{-1}\D(\Wxz)(u)s_{x,v}\\
    &+W^{-1}\Lambdax \G^{-1}\W'^{z}\Sxu s_{x,v}\\
    &\leq  \|s_{x,z}\|_\infty \|s_{x,u}\|_\infty,
\end{align*}
where for this term we also used Lemma~\ref{lem:derivativeofwprime}.
\\
Finally the last term $\W^{-1}\D(\Lambdax \G^{-1}\W\Sxz s_{x,v})(u)$ is exactly similar to the proof of Lemma~\ref{lem:derivativeofwprime} for handling $\W^{-1}\D(\Wxv)(z)$.
\end{proof}

\begin{lemma}\label{lem:lem:wbound}
We have
\begin{align*}
   -\frac{1}{\frac{4}{p} - 1}\|s_{x,v}\|_\infty \W \preccurlyeq  \Wxv \preccurlyeq \frac{1}{\frac{4}{p} - 1}\|s_{x,v}\|_\infty \W.
\end{align*}
In particular, for random $s_v$ we have with high probability
\begin{align*}
    -\frac{1}{\frac{4}{p} - 1} \W \preccurlyeq \Wxv \preccurlyeq \frac{1}{\frac{4}{p} - 1} \W.
\end{align*}
Moreover
\begin{align}
    -\frac{1}{\frac{4}{p}-1} \|v\|_g \W \preccurlyeq \Wxv \preccurlyeq \frac{1}{\frac{4}{p}-1} \|v\|_g \W.
\end{align}
\end{lemma}
\begin{proof}
Note that $\Wxv = -2 \diag(\Lambdax r_{x,v})$.
Using Lemma~\ref{lem:ginfnorm}, we have $\|r_{x,v}\|_\infty \leq \frac{1}{4/p - 1} \|s_{x,v}\|_\infty$. Hence, for every $i$:
\begin{align*}
    |{\Lambdax}_{i,} r_{x,v}| \leq w_i {r_{x,v}}_i + {\Ptwo}_i r_{x,v} \lesssim w_i \|r_{x,v}\|_\infty,
\end{align*}
which completes the proof. For random $s_{x,v}$, just note that $$\|s_{x,v}\|_\infty \lesssim 1,$$
For $g$-norm also use Lemma~\ref{lem:infwithgnorm} to upper bound infinity norm with $w$-norm.
\end{proof}

\begin{lemma}\label{lem:derivativeofr}
For the derivative of $R_v$ in direction $z$ we have
\begin{align*}
    \|\D(r_{x,z})(v)\|_\infty \leq \|s_{x,v}\|_\infty \|r_{x,z}\|_\infty.
\end{align*}
\end{lemma}
\begin{proof}
We can write
\begin{align*}
    \D(R_{x,z})(v) = \diag{\G^{-1}\D\G(z)\G^{-1}\W s_{x,v}} + \diag{\G^{-1}\Wxv s_{x,v}} + \diag{\G^{-1}\W\Sxz s_{x,v}}.
\end{align*}
But note that from Lemma~\ref{lem:ginfnorm} we have $\|\G^{-1}\W s_{x,v}\|_\infty \leq \|s_{x,v}\|_\infty$ and from Lemma~\ref{lem:infnorm2} we have $$\|\G^{-1}\Wxv s_{x,v}\|_\infty, \|\G^{-1}\D\G(z)s_{x,v}\|_\infty \leq \|s_{x,v}\|_\infty \|s_{x,z}\|_\infty,$$
which completes the proof.
\end{proof}

\begin{lemma}\label{lem:derivativeofwprime}
We have
\begin{align*}
    \|\D(w)(v,z)/w\|_\infty \leq \|s_{x,v}\|_\infty \|s_{x,z}\|_\infty.
\end{align*}
\end{lemma}
\begin{proof}
We consider $Diag(D(w')(z)/w)$:
\begin{align}
    \text{LHS} = & \W^{-1}\D(\Lambdax \G^{-1}\W s_{x,v})(z) = \W^{-1}\D\Lambdax(z)\G^{-1}\W s_{x,v} + \W^{-1}\Lambdax
    \G^{-1}\D \G(z)\G^{-1}\W s_{x,v} \\
    &+ \W^{-1}\Lambdax \G^{-1}\Wxz s_{x,v} + {\W}^{-1}\Lambdax \G^{-1}\W \Sxz s_{x,v}.\label{eq:tmp4}
\end{align}
Now from Lemmas~\ref{lem:infnorm1} and~\ref{lem:ginfnorm}:
\begin{align*}
    &\|\W^{-1}\D\Lambdax(z)\G^{-1}\W s_{x,v}\|_\infty \leq \|s_{x,z}\|_\infty \|\G^{-1}\W s_{x,v}\|_\infty \leq \|s_{x,v}\|_\infty \|s_{x,z}\|_\infty,\\
    &\|\W^{-1}\Lambdax \G^{-1}\D\G(z)\G^{-1}\W s_{x,v}\|_\infty \leq \|\G^{-1}\D \G(z)\G^{-1}\W s_{x,v}\|_\infty \\
    &\leq
    \|s_{x,z}\|_\infty\|\G^{-1}\W s_{x,v}\|_\infty \leq \|s_{x,z}\|_\infty\|s_{x,v}\|_\infty,\\
    &\|\W^{-1}\Lambdax \G^{-1}\W'xz s_{x,v}\|_\infty \leq \|\G^{-1}\Wxz s_{x,v}\|_\infty \leq \|s_{x,z}\|_\infty\|s_{x,v}\|_\infty,\\
    &\|\W^{-1}\Lambdax \G^{-1}\W \Sxz s_{x,v}\|_\infty \leq \|\G^{-1}\W \Sxz \sxv\|_\infty \leq \|s_{x,z}\|_\infty \|s_{x,v}\|_\infty.  
\end{align*}
\end{proof}

\begin{lemma}\label{lem:infnormhelper1}
We have
\begin{align*}
    \|\W^{-1}\D^2(\Ptwo)(z,u) s_{x,v}\|_\infty \leq \|s_{x,z}\|_\infty\|s_{x,u}\|_\infty\|s_{x,v}\|_\infty.
\end{align*}
\end{lemma}
\begin{proof}
We use the $\sum$ notation below to consider all the permutations among $u$, $v$, and $z$. 
\begin{align*}
    \D^2(\Ptwo)(z,u)s_{x,v} & \rightarrow \sum \Px\odot (\Px\Sxz\Px\Sxu\Px) + (\Px\Sxz\Px)\odot (\Px\Sxu\Px) + \Sxu \Ptwo \Sxv \\
    & +\Sxu \Sxv \Ptwo + \Ptwo\Sxu \Sxv + \Sxv(\Px\odot \Px \Sxu\Px) + (\Px\odot \Px\Sxu\Px) \Sxv.
\end{align*}
But note that in general for diagonal matrices $D_1, D_2$ we have from Lemmas~\ref{lem:infnorm1},~\ref{lem:infnorm4}, and~\ref{lem:infnorm3}:
\begin{align*}
    \|\W^{-1}(\Px\odot(\Px\D_1\Px))s_{x,v}\|_\infty \leq \|s_{x,v}\|_\infty \|D_1\|_{op},\\
    \|\W^{-1}((\Px D_2\Px))\odot(\Px D_1\Px)\|_\infty \leq \|s_{x,v}\|_\infty \|D_1\|_{op}\|D_2\|_{op},\\
    \|\W^{-1}(\Px\odot(\Px D_1\Px D_2\Px))s_{x,v}\|_\infty \leq \|D_3\|_{op},
\end{align*}
as the proof of Lemma~\ref{lem:infnorm4} can be generalized to arbitrary diagonal matrices $D_1$ and $D_2$ in place of $R_z$ and $R_u$. The proof is complete.
\end{proof}

\begin{lemma}\label{lem:infnormhelper2}
We have
\begin{align*}
    \|\W^{-1}\D^2\G(z,u)s_{x,v}\|_\infty, \|\W^{-1}\D^2\Lambdax (z,u)s_v\|_\infty \leq \|s_{x,z}\|_\infty\|s_{x,u}\|_\infty\|s_{x,v}\|_\infty.
\end{align*}
\end{lemma}
\begin{proof}
Directly from Lemma~\ref{lem:infnormhelper1}, noting the fact that both $G$ and $\Lambda$ are linear combinations of $W$ and $P^{(2)}$.
\end{proof}

\begin{lemma}\label{lem:infnorm1}
We have
\begin{align*}
    \|\W^{-1}\Pxv \odot \tilde \Px s_{x,v}\|_\infty \leq \|s_{x,v}\|_\infty \|s_{x,z}\|_\infty.
\end{align*}
\end{lemma}
\begin{proof}
\begin{align*}
e_i^\top (\Px \odot \Px\Sxz\Px)s_{x,v} & \leq
(s_{x,z} \odot {\Px}_{i,})^\top \Px (s_{x,v} \odot {\Px}_{i,})\\
& \leq \|s_{x,z} \odot {\Px}_{i,}\|_2 \|s_{x,v} \odot {\Px}_{i,}\|_2 \leq \|s_{x,z}\|_\infty \|s_{x,v}\|_\infty w_i.
\end{align*}
\end{proof}

\begin{lemma}\label{lem:infnorm4}
We have
\begin{align*}
    \|\W^{-1}\big((\Px \Sxz\Px) \odot (\Px\Sxu\Px)\big)s_{x,v}\|_\infty \leq  \|s_{x,z}\|_\infty \|s_{x,u}\|_\infty \|s_{x,v}\|_\infty.
\end{align*}
\end{lemma}
\begin{proof}
Observe that the 2-norm of the $i$th row of the matrix $\Px\Rxz\Px$ is at most $\|s_{x,z}\|_\infty \sqrt{w_i}$. This is because
\begin{align*}
    \|\Px\Sxz\Px e_i\|^2 \leq \|\Sxz\Px e_i\|^2 = \sqrt{\sum_j {s^2_{x,z}}_j {\Px}_{i,j}^2} \leq 
    \|s_{x,z}\|_\infty \sqrt{w_i}.
\end{align*}
Now note that
\begin{align*}
    e_i^\top \big((\Px\Sxz\Px) \odot (\Px\Sxu\Px)\big)s_{x,v} 
    & = \big(e_i^\top (\Px\Sxz \Px)\odot e_i^\top (\Px\Sxu\Px)\big)^\top s_{x,v}\\
    & = (e_i^\top (\Px\Sxu\Px))^\top((\Px\Sxz \Px e_i) \odot s_{x,v})
    \\
    & =  ({\Px}_{i,}\odot s_{x,u}) \Px ((\Px\Sxz\Px e_i)\odot s_{x,v}) \\
    & \leq \|{\Px}_{i,}\odot s_{x,u}\|_2 \|(\Px\Sxz\Px e_i)\odot s_{x,v}\|_2 \\
    & \leq \|s_{x,u}\|_\infty \|s_{x,v}\|_\infty \|{\Px}_{i,}\|\|(\Px\Sxz\Px e_i)\| \\
    & =  \|s_{x,u}\|_\infty \|s_{x,v}\|_\infty \|{\Px}_{i,}\|\|\Px (s_{x,z}\odot {\Px}_{i,})\| \\
    & \leq \|s_{x,u}\|_\infty \|s_{x,v}\|_\infty \|{\Px}_{i,}\|\| s_{x,z}\odot {\Px}_{i,}\| \\
    & \leq \|s_{x,u}\|_\infty \|s_{x,v}\|_\infty \|{\Px}_{i,}\|\|s_{x,z}\|_\infty\|{\Px}_{i,}\| \\
    & = w_i \|s_{x,u}\|_\infty \|s_{x,v}\|_\infty\|s_{x,z}\|_\infty.
\end{align*}

\end{proof}

\begin{lemma}\label{lem:infnorm3}
We have
\begin{align*}
    \|\W^{-1} (\Px \odot (\Px\Rxu\Px\Rxv\Px)) s_{x,v}\|_\infty \leq \|s_{x,z}\|_\infty \|s_{x,u}\|_\infty \|s_{x,v}\|_\infty.
\end{align*}
\end{lemma}
\begin{proof}
Note that by Cauchy Schwarz
\begin{align*}
    e_i^\top (\Px \odot (\Px\Rxu\Px\Rxv\Px))s_{x,v} & \leq \|\Px\Rxu\Px\Rxv\Px e_i\|_2 \|{\Px}_{i,}\odot s_{x,v}\|_2 \\
    &\leq \|\Rxu\Px\Rxv\Px e_i\|_2\sqrt{w_i}\|s_{x,v}\|_\infty\\
    &\leq \sqrt{w_i}\|r_{x,u}\|_\infty \|\Px\Rxz\Px e_i\|_2 \|s_{x,v}\|_\infty\\
    & \leq \sqrt{w_i}\|r_{x,u}\|_\infty \|\Rxz \Px e_i\|_2 \|s_{x,v}\|_\infty\\
    & \leq \sqrt{w_i}\|r_{x,u}\|_\infty \|\Px e_i\|_2 \|s_{x,v}\|_\infty\|s_{x,z}\|_\infty\\
    & = w_i \|r_{x,u}\|_\infty\|s_{x,v}\|_\infty\|s_{x,z}\|_\infty.
\end{align*}
\end{proof}

\begin{lemma}\label{lem:infnorm2}
We have
\begin{align*}
    &\|\G^{-1}\D\G(z)s_{x,v}\|_\infty \leq \|s_{x,v}\|_\infty\|s_{x,z}\|_\infty,\\
    & \|\G^{-1}\D\Lambdax(z)s_{x,v}\|_\infty \leq \|s_{x,v}\|_\infty\|s_{x,z}\|_\infty,\\
    &\|\G^{-1}\Wxv s_v\| \leq \|s_v\|_\infty.
\end{align*}
\end{lemma}
\begin{proof}
Note that 
\begin{align*}
    \D\G(z) = \Wxv  + \Px\odot (\Px\Sxz\Px).
\end{align*}
Now from Lemma~\ref{lem:infnorm1}, we know $$|(\Px\odot (\Px\Sxz\Px) s_{x,v})_i| \leq \|s_{x,v}\|_\infty\|s_{x,z}\|_\infty w_i.$$
Now similar to Lemma~\ref{lem:ginfnorm}, we can show
\begin{align*}
    \|\G^{-1} \Px\odot \Pxv s_{x,v}\|_\infty \leq \|s_{x,v}\|_\infty. 
\end{align*}
On the other hand, note that 
\begin{align*}
    |(\Wxv s_{x,z})_i| \leq w_i \|s_{x,z}\|_\infty
\end{align*}
so similarly we can argue
\begin{align*}
    \|\G^{-1}\Wxv s_{x,z}\|_\infty \leq \|s_{x,z}\|_\infty.
\end{align*}
Finally, as both $\G$ and $\Lambdax$ are a combination of $\W$ and $\Ptwo$ matrices, this completes the proof.
\end{proof}

\begin{lemma}\label{lem:wsecondderivativeinf}
We have
\begin{align*}
    \|\W^{-1}\D\Wxv(u)\|_\infty \leq \|s_{x,v}\|_\infty \|s_{x,u}\|_\infty. 
\end{align*}
\end{lemma}

\subsection{Norm of the bias}
\begin{lemma}\label{lem:biasnorm}
We have
\begin{align*}
    \|\mu\|_g \leq (1 + \alpha \sqrt{\alpha_0})\sqrt n.
\end{align*}
\end{lemma}
\begin{proof}
For the first part
\begin{align*}
    \|\nabla \phi\|_g = \|\D\phi\|_{g^{-1}} \leq \alpha \sqrt{n\alpha_0}
\end{align*}
from Lemma~\ref{lem:philemma}. For the second part, writing $tr(g^{-1}Dg)$ as an expectation
\begin{align*}
    \tr(g^{-1}\D g) = \mathbb E_{v\sim \mathcal N(0,g^{-1})} \D g(v)v,
\end{align*}
we have for independent $v, v' \sim \mathcal N(0,g^{-1})$:
\begin{align*}
    \|g^{-1}\tr(g^{-1}\D g)\|_g^2 & = \mathbb E_{v,v'}v^\top \D g(v)g^{-1}\D g(v')v'\\
    & \leq \mathbb E_{v} v^\top \D g(v)g^{-1}\D g(v)v\\
    & \leq \mathbb E_{v} \|s_v\|_\infty^2 v^\top g v\lesssim n,
\end{align*}
where we used Lemma~\ref{lem:momentbound1}. This completes the proof.
\end{proof}

\subsection{Comparison between leverage scores}
\begin{lemma}\label{lem:mcomparisonlemma}
Let 
\begin{align*}
    \tilde \sigma_i = (\W^{1/2}\A g^{-1}\A^\top \W^{1/2})_{i,i}.
\end{align*}
Then
\begin{align*}
    \tilde \sigma_i/ w_i \leq (\frac{m}{n})^{2/p} w_i^{2/p},
\end{align*}
which implies
\begin{align*}
    \tilde{\sigma}_i/w_i \leq (\frac{m}{n})^{\frac{2/p}{1+2/p}}.
\end{align*}
\end{lemma}
\begin{proof}
Simply note that $g \geq (\frac{n}{m})^{2/p}\A^\top \W^{1-2/p}\A$, which implies
\begin{align*}
    (\W^{1/2}\A (\A^\top \W^{1-2/p}\A)^{-1}\A^\top \W^{1/2})_{i,i} \leq w_i^{2/p}
\end{align*}
\end{proof}

\subsection{Norm comparison between covariant and normal derivatives}
\begin{lemma}\label{lem:vderivative}
Given a family of Hamiltonian curves $\gamma_s$ in the interval $(0,\delta)$ where $\gamma_0$ is $(\delta, c)-$nice, with $v = v_s = \gamma'_s(t)$, we have
\begin{align*}
    \|\frac{d}{ds}v_s\|_g\Big|_{s=0} \leq c + 1/\delta.
\end{align*}
\end{lemma}
\begin{proof}
    From Lemma~\ref{lem:R1bound} we have $R_1 \leq \sqrt n$ along the curve, so by Lemma 23 in~\cite{lee2018convergence} (note that the condition $\delta^2 \lesssim 1/R_1$ is satisfied) we get
    \begin{align*}
        \|\nabla_{\frac{d}{ds}\gamma_s(t)} v_s\|_g \leq \frac{1}{\delta}. 
    \end{align*}
    But now from Lemma~\ref{lem:covarianttonormal}
    \begin{align*}
       \|\frac{d}{ds}v_s\|_g \leq \|s_v\|_\infty \|\frac{d}{ds}\gamma_s(t)\|_g + \|\nabla_{\frac{d}{ds}\gamma_s(t)} v_s\|_g, 
    \end{align*}
    As always, our parameterization in $s$ is always unit norm, so $\|\frac{d}{ds}\gamma_s(t)\|_g $, and from niceness of the curve $\|s_v\|_\infty \lesssim c$, which completes the proof.
\end{proof}

\begin{lemma}\label{lem:covarianttonormal}
For a vector field $v$ and arbitrary vector $z$ at a point $x$, denoting $\D v(z)$ by $v'$, we have
\begin{align*}
    \|v'\|_{g} \leq \|s_v\|_\infty\|z\|_{g} + \|\nabla_z(v)\|_{g}.
\end{align*}
\end{lemma}
\begin{proof}
We have
\begin{align*}
    \nabla_z(v) = v' + \frac{1}{2}g^{-1}\D g(v)z,
\end{align*}
so
\begin{align*}
    \|v'\|_{g} \leq \|\nabla_z(v)\|_{g} + \|g^{-1}\D g(v)z\|_{g} \leq \|\nabla_z(v)\|_{g} + \|s_v\|_\infty \|z\|_{g}.
\end{align*}
\end{proof}

\subsection{Log barrier infinity self-concordance}\label{sec:logbarrierself}

\begin{proof}[Proof of Lemma~\ref{lem:logbarrierself}]
    The log barrier metric is
    \begin{align*}
        g_2 = \nabla^2 \phi_\ell(x) = \A^\top \A.
    \end{align*}
     Its directional derivative is given by
    \begin{align*}
        \D g_2(v) = -2\A^\top \Sxv \A,
    \end{align*}
    which can be bounded as
    \begin{align*}
        -\|s_{x,v}\|_\infty \A^\top\A \preccurlyeq \A^\top \Sxv \A \preccurlyeq \|s_{x,v}\|_\infty \A^\top\A.
    \end{align*}
    Similarly, the second and third directional derivatives of $g_2$ are given by
    \begin{align*}
        &\D^2 g_2(v,z) = 6\A^\top \Sxv \Sxz \A,\\
        &\D^2 g_2(v,z,u) = -24\A^\top \Sxv \Sxz \Sxu \A,
    \end{align*}
    which can be bounded as
    \begin{gather*}
        -\|s_{x,v}\|_\infty\|s_{x,z}\|_\infty \A^\top \A \preccurlyeq \A^\top \Sxv \Sxz \A \preccurlyeq \|s_{x,v}\|_\infty \|s_{x,z}\|_\infty \A^\top \A,\\
        -\|s_{x,v}\|_\infty\|s_{x,z}\|_\infty\|s_{x,u}\|_\infty \A^\top \A \preccurlyeq \A^\top \Sxv \Sxz \Sxu\A  \preccurlyeq \|s_{x,v}\|_\infty \|s_{x,z}\|_\infty\|s_{x,u}\|_\infty \A^\top \A.
    \end{gather*}
    This completes the proof.
\end{proof}

\newpage
\subsection{Other helper Lemmas}
\begin{lemma}\label{lem:momentbound1}
    For vector $v \sim \mathcal N(0,g^{-1})$, we have with high probability
    \begin{align*}
    \mathbb \|s_{x,v}\|_{\infty} v^\top g v = O(n).
    \end{align*}
\end{lemma}
\begin{proof}
    Directly from Gaussian moment bounds.
\end{proof}

\begin{lemma}\label{lem:lsbarriergradient}
For the $p$-Lewis weights barrier $\phi_p = \log\det \A^\top \W^{1-2/p}\A$, we have
\begin{align*}
    \D\phi_p(x) =  \A^\top w.
\end{align*}    
\end{lemma}
\begin{proof}
    Proof is done in~\cite{lee2019solving}.
\end{proof}

\begin{lemma}\label{lem:exponentiallownerinequality}
    For any positive integer $n$, vector $v$, and matrix $\tilde g \preccurlyeq g$ we have
    \begin{align*}
    g^{1/2}(g^{-1/2}\tilde gg^{-1/2})^{n}g^{1/2} \preccurlyeq g.    
    \end{align*}
    \begin{proof}
        Directly from the fact that if $A \preccurlyeq B$, then for any matrix $C$ we have $C^\top A C \preccurlyeq C^\top B C$.
    \end{proof}
\end{lemma}
\begin{lemma}\label{lem:operatornormbound1}
    For operator $g^{-1}\D g(v)g^{-1}\D g(v)$, we have
    $\|g^{-1}\D g(v)g^{-1}\D g(v)\ell\|_g \leq \|s_{x,v}\|_\infty^2\|\ell\|_g$.
\end{lemma}
\begin{proof}
We have
    \begin{align}
    \|g^{-1}\D g(v)g^{-1}\D g(v)\ell\|_g 
    &\leq\sqrt{\ell^\top \D g(v)g^{-1}\D g(v)g^{-1}\D g(v)g^{-1}\D g(v)\ell} \\
    &\leq \|s_{x,v}\|_\infty^2 \sqrt{\ell^\top g \ell} \lesssim \|s_{x,v}\|_\infty^2 \|\ell\|_g.
\end{align}
\end{proof}

\begin{lemma}\label{lem:factoredlemma}
    For vector field $w$ on manifold $\mathcal M$, we have
   \begin{align*}
         \tr(g^{-1}\D g(w)) \lesssim \|w\|_g.
    \end{align*}
\end{lemma}
\begin{proof}
We have
    \begin{align*}
    \tr(g^{-1}\D g(w)) 
    &=\mathbb E_{v'\sim \mathcal N(0,g^{-1})} {v'}^\top \D g(w)v'\\
    &=\mathbb E_{v'\sim \mathcal N(0,g^{-1})} {v'}^\top \D g(v')w\\
    &\leq \mathbb E_{v'}\|v'\|_\infty \sqrt{{v'}^\top g v'} \sqrt{{w}^\top g w}\\
    &\lesssim \sqrt n\|w\|_g,
\end{align*}
where in the last line we used Lemma~\ref{lem:momentbound1}.
\end{proof}

\begin{lemma}\label{lem:factoredlemma2}
    For arbitrary vector field $w$ on $\mathcal M$ we have
    \begin{align*}
        \big|\tr(g^{-1}\D g(z, w))\big| \leq \sqrt n\|w\|_g\|z\|_g.
    \end{align*}
\end{lemma}
\begin{proof}
    We can write
    \begin{align*}
    &\big|\tr(g^{-1}\D g(z, w))\big|\\
    & =  \mathbb E_{v' \sim \mathcal N(0,g^{-1})} {v'}^\top \D g(z,w)v' \\
    & =
     \mathbb E_{v' \sim \mathcal N(0,g^{-1})} {v'}^\top \D g(z, v')w\\
    & \leq \mathbb E_{v' \sim \mathcal N(0,g^{-1})} \|z\|_\infty \|v'\|_\infty \sqrt{{v'}^\top g v'} \sqrt{{w}^\top g w}\\ 
    & \leq \|w\|_g \sqrt n\|z\|_\infty\\
     & \leq \|w\|_g \sqrt n\|z\|_g.
\end{align*}
\begin{lemma}\label{lem:factoredlemma3}
    For vector field $w$ we have
    \begin{align*}
        |\tr(g^{-1}\D g(z)g^{-1}\D g(w))| \leq \sqrt n \|z\|_g\|w\|_g.
    \end{align*}
\end{lemma}
\begin{align*}
    |\tr(g^{-1}\D g(z)g^{-1}\D g(w))| &=
    |\mathbb E_{v' \sim \mathcal N(0,g^{-1})} v'^\top \D g(z)g^{-1}\D g(w) v'|\\
    &=\mathbb E_{v'} |v'^\top \D g(z)g^{-1}\D g(v') w|\\
    &\leq\mathbb E_{v'} \sqrt{v'^\top \D g(z)g^{-1}\D g(v')g^{-1}\D g(v')g^{-1}\D g(z)v'}\|w\|_g\\
    &=\mathbb E_{v'} \sqrt{v'^\top \D g(z)g^{-1/2}(g^{-1/2}\D g(v')g^{-1/2})^2g^{-1/2}\D g(z)v'}\|w\|_g.
\end{align*}
But note that
\begin{align*}
    g^{-1/2}\D g(v')g^{-1/2} \leq \|v'\|_\infty I.
\end{align*}
Hence
\begin{align*}
    |\tr(g^{-1}\D g(z)g^{-1}\D g(w))| &\leq \mathbb E_{v'} \sqrt{v'^\top g^{1/2}(g^{-1/2}\D g(z)g^{-1/2})^2 g^{1/2}v'}\|w\|_g\\
    &\leq \mathbb E_{v'}\|s_z\|_\infty\|v'\|_g \|w\|_g \\
    &\leq \mathbb E_{v'}\|z\|_g\|v'\|_g \|w\|_g \\
    &\leq \sqrt n \|z\|_g\|w\|_g.
\end{align*}

where we used Lemma~\ref{lem:momentbound1} and Lemma~\ref{lem:infwithgnorm}.
\end{proof}

\newpage
\section{Remaining Proofs}\label{sec:remainingproofs}
\subsection{Proof of Theorem~\ref{prop:conductance}}\label{app:sconductance}

Consider a subset $S \subseteq \mathcal S$ with $0.5 \geq \pi(S) = s' \geq s \geq 2\rho$. Then, to show a lower bound for $s$-conductance, we need to lower bound
\begin{align*}
    P(S,S^c)/P(S),
\end{align*}
where $P(.,.) = \int_{x \in S}\mathcal T_x(S^c)\pi(x)dx$ is the probability that we are in set $S$ and the next step of the Markov chain we escape $S$ and $P$ is the probability measure corresponding to $\pi$. Recall that $\mathcal T_x(.)$ is the Markov kernel, specifying the distribution of the next step given we are at point $x$. Now assume that the conductance bound does not hold, i.e. there exists such $S$ with 
\begin{align*}
    P(S,S^c)/P(S) = O(\Delta \psi_M).
\end{align*}
Note that because the chain is reversible, we have
\begin{align*}
    P(S,S^c) = P(S^c,S),
\end{align*}
and because $\pi(S) \leq 0.5$, we have
\begin{align}
    P(S^c,S)/P(S^c) \leq P(S,S^c)/P(S) = O(\Delta \psi_M).\label{eq:contradictassumption}
\end{align}
Next, define the set $\tilde S \subseteq S$ to be the points $x$ from which our chance of escaping $S$ is at least $0.01$. Now if $\pi(S) \geq \Delta \psi_{\mathcal M}\pi(S)/2$, then given that we are in $S$, we have at least $\Delta \psi_{\mathcal M}$ chance of escaping $S$ which contradicts~\eqref{eq:contradictassumption}. This means
\begin{align}
    \pi(\tilde S) \leq \Delta \psi_{\mathcal M}/2. \pi(S).\label{eq:stmp1}
\end{align}
On the other hand, note that for point $x_1$ with $d(x_1,x_0) \leq \Delta$ for $x_0 \in S - \tilde S$, we have
\begin{align}
    TV(P_{x_0}, P_{x_1})\leq 0.9,\label{eq:TVbound1}
\end{align}
which means $x_1$ cannot be in $S - \tilde S$, hence it should be in $S^c$. Therefore, defining the set $S^{+\Delta}$ as the set of points outside $\tilde S$ which are $\Delta$ close to a point in $S - \tilde S - \mathcal M'^c$, we have
\begin{align}
    S^{+\Delta}\subseteq S^c \cup \mathcal M'^c.\label{eq:subsetdelta}
\end{align}
On the other hand, from isoperimetry (because $\pi(S) \leq \frac{1}{2}$) and the fact that $\Delta \psi_{\mathcal M} \leq 1/2$ we have
\begin{align*}
    \pi(S^{+\Delta}) \geq \Delta \psi_{\mathcal M} (\pi(S - \tilde S) - \pi(\mathcal M'^c))\geq \Delta \psi_{\mathcal M}(\pi(S)/2 -\rho) \geq \Delta \psi_{\mathcal M}(\pi(S)/4).
\end{align*}
Therefore, from the assumption $s \geq \rho/(8\Delta \psi_{\mathcal M})$:
\begin{align*}
    \pi(S^{+\Delta} - \mathcal M'^c) \geq \Delta \psi_{\mathcal M}(\pi(S)/4 - \pi(S)/8) \geq \Delta \psi_{\mathcal M}\pi(S)/8,
\end{align*}
which implies from Equations~\eqref{eq:TVbound1} and~\eqref{eq:subsetdelta}:
\begin{align*}
    P(S,S^c) \geq P(S, S^{+\Delta} - \mathcal M'^c) = P(S^{+\Delta} - \mathcal M'^c, S) \geq \Delta \psi_{\mathcal M}(\pi(S)/8) \times 0.99 \geq \Delta \psi_{\mathcal M} \psi(S)/16,
\end{align*}
which proves that the conducance is lower bounded by $\Omega(\Delta \psi_{\mathcal M})$.

\subsection{Properties of Lewis weights}
In this section, we recall some properties of Lewis weights which we use in the proof.

\begin{lemma}[Fixed point property of Lewis weights]
The Lewis weights of the matrix $A_x$ is the unique vector $w$ in $\R^m_{\geq 0}$ with $W = \diag{w}$ such that
\begin{align*}
    \sigma(W^{1/2-1/p}\A) = W,
\end{align*}
where $\sigma(.)$ denotes the leverage scores of the matrix.
\end{lemma}
\begin{proof}
Recall the definition of Lewis weights as the optimum of the objective in Equation~\eqref{eq:lewisweightdefinition}. 
Taking derivative with respect to $W$, we get
\begin{align*}
    -(1-2/p)\sigma/w + (1-2/p)\indic^\top w = 0, 
\end{align*}
where $\sigma \triangleq(W^{1/2 - 1/p}A_x)$ is the vector of leverage scores defined as
\begin{align*}
    \sigma = \diagtwo{W^{1/2-1/p}\A (\A^\top W^{1-2/p}\A)^{-1} {A_x}^\top W^{1/2-1/p}}.
\end{align*}
\end{proof}

\begin{proof}[Proof of Lemma~\ref{lem:metricsecondform}]
    The first form of the Lewis weight metric $g_1$ directly follows from Equation 5.5 in Lemma 31. in~\cite{lee2019solving}. To see why the second form in Equation~\eqref{eq:metricotherform} holds, note that
    \begin{align*}
        \G = \frac{2}{p}\W + (1-\frac{2}{p})\Ptwo.
    \end{align*}
    Hence
    \begin{align*}
        \G - \frac{2}{p}\Lambdax = \Ptwo,
    \end{align*}
    which implies
    \begin{align}
        \Lambdax = \frac{p}{2}(\G - \Ptwo).\label{eq:LambdatoGP2}
    \end{align}
    Plugging Equation~\eqref{eq:LambdatoGP2} into the first form in Equation~\eqref{lem:initialform} completes the proof.
\end{proof}

\begin{proof}[Proof of Lemma~\ref{lem:metricsecondform}]\label{app:metricform}
The first formulation follows from~\cite{lee2014path}. To show the second formulation, recall the definition of $\Lambdax$:
\begin{align*}
    \Lambdax = \frac{p}{2}(\G - \Ptwo),
\end{align*}
Plugging the above into the first formulation results in the second formulation.
\end{proof}

\begin{proof}[Proof of Lemma~\ref{lem:lewisweightbarrier}]
    Directly from Lemma 31 in~\cite{lee2019solving}.
\end{proof}
\begin{lemma}[Gradient of the $p$ Lewis weights barrier]\label{lem:gradientoflewisbarrier}
    The gradient of the $p$ Lewis weights barrier $\phi_p$ is given by
    \begin{align*}
        \D\phi_p(x) = \A^\top w_x.
    \end{align*}
\end{lemma}
\begin{proof}
    Taking directional derivative in direction $v$, using the chain rule
    \begin{align*}
       \D\phi_p(x)[v] &= 2\tr\big((\A^\top \W^{1-2/p}\A)^{-1} (\A^\top\Sxv \W^{1-2/p} \A\big) \\
       &+ \D(w_x)^\top \frac{\partial \big(-\logdet{\A^\top \W^{1-2/p}\A} + (1-2/p)\indic^T w\big)}{\partial w_x},
    \end{align*}
    But because $w_x$ is the maximizer of $\big(-\logdet{\A^\top W^{1-2/p}\A} + (1-2/p)\indic^\top w\big)$, the second term is zero and the proof is complete.
\end{proof}

\subsubsection{Proof of Lemma~\ref{lem:gderivative}}
To differentiate $g_1$ in direction $v$, we differentiate each of the matrices in the product regarding the formula of $g_1$ one by one. Starting from $\A$, we use the first formulation in Equation~\eqref{lem:initialform} and we get $(\trr 1)$ term. Next, differentiating $W$ and $2\Lambdax$ in $\A^\top (\W + 2\Lambdax)\A$ we get
\begin{align*}
    \A^\top \D(\W + 2\Lambdax)(v) \A
    &= \A^\top (3\D\W(v) - 2\D\Lambdax(v))\A\\
    &= 3\A^\top \Wxv \A - 4\A^\top (\Px\odot \Big(\Px \odot (-\Rxv\Px - \Px\Rxv + 2 \Pxv) \Big))\A,
\end{align*}
which is the $(\star 2)$ term. Furthermore, differentiating $\A^\top \Lambdax \G^{-1}\Lambdax \A$ with respect to $\Lambda$, we get $(\star 3)$ and $(\star 4)$ terms. Finally note that the derivative of $G^{-1}$ is:
\begin{align*}
    \D(\G^{-1})(v) = -\G^{-1}\D\G(v)\G^{-1}
    = \G^{-1}\Big(\frac{2}{p}\Wxv + 2\Big(1-\frac{2}{p}\Big(\Px\odot (-\Rxv\Px - \Px\Rxv + 2\Pxv)\Big)\Big)\Big)\G^{-1}.
\end{align*}
Therefore, differentiating the $G^{-1}$ part in $A^T \Lambda G^{-1}\Lambda A$ we get the $(\trr 5)'$, $(\trr 6)$, and $(\trr 7)$ terms.

\subsection{Derivative of }
\begin{lemma}\label{lem:star5terms}
 The derivative of the term $(\star 5)$ defined in Lemma~\ref{lem:gderivative}, ignoring the constants is equal to
\begin{align*}
    D(\trr 5)(z) & \rightarrow\\
    &\A^\top \Sxz \W' \G^{-1}\Lambdax \A &&(\trrr 1)\\
    &+ \A^\top \D(\W')(z)\G^{-1}\Lambdax \A &&(\trrr 2)\\
    &+\A^\top \W' \G^{-1}\D\G(z)\G^{-1}\Lambdax \A &&(\trrr 3)\\
    &+\A^\top \W' \G^{-1}\D\Lambdax(z)\A &&(\trrr 4)\\
    &+\A^\top \W' \G^{-1}\Lambdax \Sxz \A &&(\trrr 5).
\end{align*}
\end{lemma}

\section{Self-concordance Parameter of $\phi$}
Here we provide a bound for the self-concordance parameter of $\phi$.
\begin{lemma}[Self-concordance parameter of $\phi$]\label{lem:barrierparameter}
For our hybrid barrier $\phi$, the self-concordance parameter is defined as
\begin{align*}
    \nu = \sup_{x\in \mathcal P} \D\phi(x)^\top (\D^2 \phi(x))^{-1} \D\phi(x),
\end{align*}
is bounded by $\alpha_0 n$.
\end{lemma}
\begin{proof}
    Note that for the Lewis weights and log barrier parts of the barrier $\phi = \alpha_0 \phi_p + \alpha_0\frac{n}{m}\phi_\ell$ we can bound the barrier parameter separately as 
    \begin{align*}
        \sqrt{\D\phi(x)^\top (\D^2 \phi(x))^{-1} \D\phi(x)} &\leq \alpha_0\sqrt{\D\phi_p(x)^\top (\D^2 \phi(x))^{-1} \D\phi_p(x)} + \alpha_0\frac{n}{m}\sqrt{\D\phi_\ell(x)^\top (\D^2 \phi(x))^{-1} \D\phi_\ell(x)}\\
        &\leq \sqrt{\alpha_0}\sqrt{\D\phi_p(x)^\top (\D^2 \phi_p(x))^{-1} \D\phi_p(x)} + \sqrt{\alpha_0\frac{n}{m}}\sqrt{\D\phi_\ell(x)^\top (\D^2 \phi_\ell(x))^{-1} \D\phi_\ell(x)}.
    \end{align*}
    Now for the log barrier part, we have
    \begin{align}
        \D\phi_\ell(x)^\top (\D^2 \phi_\ell(x))^{-1} \D\phi_\ell(x) = \indic^T \A (\A\top \A)^{-1}\A^T \indic \leq m,\label{eq:firstt}
    \end{align}
    and for the $p$ Lewis weight barrier part, from Lemmas~\ref{lem:gradientoflewisbarrier} and~\ref{lem:lsbarrierapprox}:
    \begin{align}
        \D\phi_p(x)^\top (\D^2 \phi_p(x))^{-1} \D\phi_p(x) \leq w_x^{\top} \A (\A^\top \W\A)^{-1} \A^T w_x \leq \|w_x\|_2^2 = n.\label{eq:secondd}
    \end{align}
    Combining Equations~\eqref{eq:firstt} and~\eqref{eq:secondd} completes the proof.
\end{proof}
\subsection{Iteration complexity of Gaussian Cooling}
\begin{proof}[Proof of Corollary~\ref{cor:warmstart}]
    First, note that from Lemma~\ref{lem:barrierparameter}, $\phi$ is self-concordant with self-concordant parameter $\nu = \alpha_0 n$.
    The Gaussian cooling schedule introduce by authors in~\cite{lee2018convergence} can be used to relax the requirement of a warm start for our sampling algorithm, hence obtain an efficient volume algorithm. The idea is that sampling from Gibbs distributions $e^{-\alpha\phi(x)}$ with smaller variance or larger $\alpha$ is easier, so one can start from sampling a large temperature $\alpha$ and gradually decrease it. The Gaussian cooling of~\cite{lee2018convergence} evolves in phases where in the $i$th phase it generates $k_i$ approximate samples from the density proportional to $e^{-\phi(x)/\sigma_i^2}$ inside the polytope, where
    \begin{align*}
        k_i &= \Theta(\frac{\sqrt n}{\epsilon^2}\log(\frac{\sqrt n}{\epsilon})) && \text{if} \ \sigma_i^2 \leq \frac{\nu}{n}\\
        k_i &= \Theta((\frac{\sqrt \nu}{\sigma} + 1)\epsilon^2\log(\frac{n}{\epsilon}), && \text{O.W.}
    \end{align*}
    and the update rule for $\sigma_i$ is 
    \begin{align*}
        &\sigma_{i+1}^2 = \sigma_i^2 (1 + \frac{1}{\sqrt n}) && \text{if} \ \sigma_i^2 \leq \frac{\nu}{n}\\
        &\sigma_i^2 = (1 + \min\{\frac{\sigma_i}{\sqrt \nu}, \frac{1}{2}\}). &&\text{O.W.}
    \end{align*}
    starting from $\sigma_0^2 = \Theta(\epsilon^2 n^{-3}\log^{-3}(n/\epsilon))$ until $\sigma$ goes above $\Theta(\frac{\nu}{\epsilon} \log(\frac{n\nu}{\epsilon}))$. Note that the temperature parameter is given by $\alpha = 1/\sigma^2$.
    Now at each phase $i$ going from temperature $\sigma_i^2$ to $\sigma_{i+1}^2$ we have a an approximate samples from $e^{-\phi(x)/\sigma_i^2}$ which can be used as warm starts for sampling from $e^{-\phi(x)/\sigma_{i+1}^2}$, specially as $k_{i+1} \leq k_i$. Hence, our main Theorem~\ref{thm:mixing} implies that the mixing time of sampling at each phase is of order 
    \begin{align*}
    \tilde O\big(\min\{ \alpha^{-1}n^{2/3} + \alpha^{-1/3}n^{5/9}m^{1/9} + n^{1/3}m^{1/6}, \  m^{1/3}n^{4/3} \}\big) = \tilde O(\alpha^{-1}n^{2/3} + \alpha^{-1/3}n^{5/9}m^{1/9} + n^{1/3}m^{1/6}).    
    \end{align*}
    Now in the first case when $\sigma_i^2 \leq \frac{\nu}{n} = \alpha_0$, we have $\alpha \geq \frac{1}{\alpha_0}$. On the other hand, due to the update rule of $\sigma_i$ in this case, it takes $\sqrt n$ phase to double $\sigma$ and in each phase we take samples $k_i = \tilde{\Theta}(\frac{\sqrt n}{\epsilon^2})$. Hence, the total number of RHMC steps to double $\sigma$ in this case is bounded by
    \begin{align*}
        \tilde O\big((\alpha^{-1}n^{2/3} + \alpha^{-1/3}n^{5/9}m^{1/9} + n^{1/3}m^{1/6})\times \frac{\sqrt n}{\epsilon^2} \times \sqrt n\big)
        &=
        \tilde O\big((\alpha_0 n^{2/3} + \alpha_0^{1/3}n^{5/9}m^{1/9} + n^{1/3}m^{1/6}) \frac{n}{\epsilon^2}\big)
        \\
        &= \tilde O(\frac{n^{4/3}m^{1/3}}{\epsilon^2}).
    \end{align*}
    In the other case when $\sigma_i^2 \geq \frac{\nu}{n} = \alpha_0$, we have $\alpha \leq \frac{1}{\alpha_0}$. Then, the total RHMC steps to double $\sigma$ in this case can be upper bounded after substituting $\nu = n\alpha_0$ as
    \begin{align*}
        &\tilde O\big((\alpha^{-1}n^{2/3} + \alpha^{-1/3}n^{5/9}m^{1/9} + n^{1/3}m^{1/6}) \times \frac{1}{\epsilon^2}(\frac{\sqrt \nu}{\sigma} + 1) \times (\frac{\sqrt \nu}{\sigma} + 1)\big) = \tilde O(\frac{n^{4/3}m^{1/3}}{\epsilon^2}).
    \end{align*}
    This means we can calculate the integral of $e^{-\alpha \phi(x)}$ for any $\alpha$ using $\tilde O(\frac{n^{4/3}m^{1/3}}{\epsilon^2})$ steps of RHMC up to $1 \pm \epsilon$. Moreover, if we just want to sample from $e^{-\alpha \phi(x)}$ in the polytope, we don't require to take $k_i$ number of samples at phase $i$ but only need one sample, so the $\epsilon^2$ in the complexity is omitted and we end up with the complexity $\tilde O(n^{4/3}m^{1/3})$ for sampling without warm start.
\end{proof}

\end{document}